\documentclass[a4paper,11pt]{article}
\pdfoutput=1

\usepackage{geometry}										
\usepackage[english]{babel}									
\usepackage[utf8]{inputenc}									
\usepackage[T1]{fontenc}									

\usepackage{lmodern}										

\usepackage{graphicx}
\usepackage[normalem]{ulem} 
\usepackage{xcolor}
\usepackage{siunitx}
\usepackage{placeins} 
\usepackage{bbold} 

\usepackage{amsmath}										
\usepackage{amstext}										
\usepackage{amsfonts}										
\usepackage{amssymb}										
\usepackage{enumitem}
\usepackage{amsthm,amscd}

\usepackage{authblk}										

\title{Global sensitivity analysis using derivative-based sparse Poincar\'e chaos expansions}
\author[1]{Nora L\"uthen}
\author[2,3]{Olivier Roustant}
\author[3]{Fabrice Gamboa}
\author[4,3]{Bertrand Iooss}
\author[1]{Stefano Marelli}
\author[1]{Bruno Sudret}
\affil[1]{Chair of Risk, Safety, and Uncertainty Quantification, ETH Z\"urich,
	Stefano-Franscini-Platz 5, 
	8093 Z\"urich, Switzerland}
\affil[2]{INSA Toulouse, 135 avenue de Rangueil, 31077 Toulouse cedex 4, France}
\affil[3]{Institut de Mathématiques de Toulouse, 31062 Toulouse, France}
\affil[4]{EDF Lab Chatou, 6 Quai Watier, 78401 Chatou, France}

\date{\today}

\usepackage[caption=false]{subfig}

\usepackage{natbib} 
\usepackage{hyperref}
\hypersetup{
	pdftitle = {},
	pdfauthor = {},
	pdfsubject = {},
	pdfkeywords = {},
	colorlinks = true,										
			linkcolor={blue},
			citecolor={blue},
			urlcolor={blue}
}

\usepackage[capitalize]{cleveref}
\crefname{equation}{}{}

\newcommand{\curlyA}{{\mathcal A}}

\newcommand{\cl}{{\mathcal L}}
\newcommand{\cm}{{\mathcal M}}
\newcommand{\cn}{{\mathcal N}}

\newcommand{\cu}{{\mathcal U}}
\newcommand{\cx}{{\mathcal X}}
\newcommand{\Rr}{{\mathbb R}}
\newcommand{\Nn}{\mathbb{N}}

\newcommand{\ve}[1]{\boldsymbol{#1}}
\newcommand{\enum}{ , \, \dots \,,}
\newcommand{\Var}{{\rm Var}}
\newcommand{\Vare}[2]{{\rm Var}_{#1}\left[#2\right]}
\newcommand{\Esp}[1]{{\mathbb E}\left[ #1 \right]}
\newcommand{\Espe}[2]{{\mathbb E}_{#1}\left[#2\right]}

\newcommand{\norme}[2]{\left\| #1 \right\|_{#2}}
\newcommand{\innprod}[2]{\left\langle #1, #2 \right\rangle}

\newcommand{\alp}{{\ve{\alpha}}}
\newcommand{\Phal}{\Phi_{\alp}}

\newcommand{\vpali}[1]{\varphi_{#1, \alpha_{#1}}}
\newcommand{\eig}[1]{\lambda_{#1, \alpha_{#1}}}

\renewcommand{\citep}[2][]{\cite[#1]{#2}}
\renewcommand{\citet}[2][]{\cite[#1]{#2}}

\theoremstyle{definition}

\theoremstyle{remark}
\newtheorem{rem}{Remark}

\theoremstyle{theorem}
\newtheorem{thm}{Theorem}

\newtheorem{prop}{Proposition}
\newtheorem{assum}{Assumption}

\begin{document}
	
\maketitle

\begin{abstract}
 Variance-based global sensitivity analysis, in particular Sobol' analysis, is widely used for determining the importance of input variables to a computational model.
 Sobol' indices can be computed cheaply based on spectral methods like polynomial chaos expansions (PCE). Another choice are the recently developed Poincar\'e chaos expansions (PoinCE), whose orthonormal tensor-product basis is generated from the eigenfunctions of one-dimensional Poincar\'e differential operators.
 In this paper, we show that the Poincar\'e basis is the unique orthonormal basis with the property that partial derivatives of the basis form again an orthogonal basis with respect to the same measure as the original basis. 
 This special property makes PoinCE ideally suited for incorporating derivative information into the surrogate modelling process. 
 Assuming that partial derivative evaluations of the computational model are available, we compute spectral expansions in terms of Poincar\'e basis functions or basis partial derivatives, respectively, by sparse regression.
 We show on two numerical examples that the derivative-based expansions provide accurate estimates for Sobol' indices, even outperforming PCE in terms of bias and variance. 
 In addition, we derive an analytical expression based on the PoinCE coefficients for a second popular sensitivity index, the derivative-based sensitivity measure (DGSM), and explore its performance as upper bound to the corresponding total Sobol' indices.   
\end{abstract}

%
%

\section{Introduction}

Computer models simulating physical phenomena and industrial systems are commonly used in engineering and safety studies, for prediction, validation or optimisation purposes.
These numerical models often take as inputs a high number of physical parameters, whose values are variable or not perfectly known, creating the need for uncertainty quantification on model computations \citep{smi14}. 
Uncertainty quantification typically becomes more challenging the higher the input dimension is (curse of dimensionality).
In this situation, global sensitivity analysis (GSA) 
is an invaluable tool that allows the analyst to rank the relative importance of each input of the model and to detect non-influential inputs \citep{borpli16,razjak21}.
Most often relying on a probabilistic modeling of the model input variables, GSA tries to explain model output uncertainties on the basis of model input uncertainties, accounting for the full range of variation of the variables.

A well-known and widely used GSA method is Sobol' analysis \citep{sob93}, which relies on the functional ANOVA (analysis of variance) decomposition \citep{efrste81}. For a square-integrable model and independent input variables, Sobol' analysis determines which part of the model output variance can be attributed to each input and to each interaction between inputs. The overall contribution of each input, including interactions with other inputs, is provided by the total Sobol' index \citep{homsal96}. 
Sobol' indices can be estimated efficiently using various Monte Carlo-based techniques as well as metamodel-based techniques \citep{pritar17}. The latter save on expensive model evaluations by first performing a small number of model runs, which are
used to compute an accurate approximation to the original model -- the meta- or surrogate model -- {from which the Sobol' indices are finally computed} \citep{fanli06,LeGratiet2017}.

One of the most popular and powerful metamodelling methods is the polynomial chaos expansion (PCE) \citep{Xiu2002}. 
PCE represents the model in a specific basis consisting of polynomials that are orthonormal with respect to the input distribution. 
Orthogonal polynomial systems have been studied throughout the last century and they have many useful properties (see, e.g., Szeg\"o \citet{Szeg1939} and Simon \citet{Sim2010}).
One particular strength of PCE is that once it is computed, it easily gives all the variance-based quantities defined through the ANOVA decomposition, and in particular the Sobol' indices at all orders \citep{sud08}. 
In practice, the expansion cannot use infinitely many terms and must be truncated. 
Among the many approaches available to compute the expansion coefficients, sparse regression techniques combined with adaptive basis selection appear to be especially promising (see \cite{LuethenSIAMJUQ2020, LuethenIJUQ2022} for an overview). Here, a small number of terms is selected which is able to best represent the computational model based on the available model evaluations.

In some practical situations, partial derivatives of the model output with respect to each input are easily accessible, for example by algorithmic differentiation of the numerical model in the reverse (adjoint) mode \citep{griwal08}. This technique allows for computing all partial derivatives of the model output at a cost independent of the number of input variables.
Since PCEs are such a well-established metamodelling tool, there have been many efforts to leverage the additional information contained in model derivatives to improve the performance of PCE. 
The idea of including derivative information into sparse regression problems, often called \textit{gradient-enhanced $\ell^1$-minimization}, 
is tested by Jakeman et al.~\cite{Jakeman2015} for one numerical example with uniform inputs, and analyzed theoretically and numerically by Peng et al.~\cite{penham16} for Hermite PCE. Both report favorable results. 
Roderick et al.~\citet{rodani10} and Li et al.~\citet{liani11} apply polynomial regression (PCE) in the context of nuclear engineering. They include derivative information into the least-squares regression formulation and observe 
that most polynomial families are not orthogonal with respect to the $H^1$ inner product. This may deteriorate certain properties of the regression matrix.
To alleviate this issue, Guo et al.~\cite{Guo2018} develop a preconditioning procedure for gradient-enhanced sparse regression with certain polynomial families, with the goal of improving the orthogonality properties of the regression matrix.
In all these approaches, the utilization of derivative information is not straightforward, but requires specific polynomial families and/or specialized sampling and preconditioning, because the partial derivatives of a PCE basis do in general not form an orthogonal system.
Gejadze et al.~\citet{gejmal19} have derived derivative-enhanced projection methods to compute the PCE coefficients but their method is restricted to Hermite polynomials and low polynomial degree.

On a different note, the availability of model derivatives has implications also for GSA. 
The so-called \textit{Derivative-based Global Sensitivity Measures} (DGSM) are computed by integrating the squared partial derivatives of the model output 
over the domain of the inputs \citep{Sobol_Gresham_1995,Sobol2009}. 
These indices have been shown to be efficiently estimated by sampling techniques (as Monte Carlo or quasi-Monte Carlo) \citep{Kucherenko2009} as well as from PCE \citep{sudmai15}, and have been proven to be an excellent screening technique (i.e., detecting all the non-influential inputs among a large number), see e.g.\ the review by Kucherenko and Iooss \cite{kucioo17}.
Indeed, the interpretation of DGSM indices is straightforward due to their inequality relationship with Sobol' indices: multiplied with the associated Poincar\'e constant, DGSM indices provide an upper bound of the total Sobol' index \citep{Lamboni_et_al_2013}, regardless of the input probability distribution.

Another way to utilize model derivatives, which solves the issues present for the polynomial chaos formulation, and naturally provides sharp lower bounds as well as upper bounds on total Sobol' indices, 
is to compute
\textit{Poincar\'e chaos expansions} \citep{rougam20}, which we will abbreviate by \textit{PoinCE} in the sequel.
Similar to PCE, PoinCE is a spectral expansion
in terms of an orthonormal basis whose elements are eigenfunctions of the so-called \textit{Poincar\'e differential operator}. 
The eigenfunctions are in general non-polynomial, except for the special case of the Gaussian distribution, where they coincide with the Hermite polynomials.
The key property of PoinCE is that the partial derivatives of the basis form again an orthogonal basis with respect to the input distribution. 
This allows to conveniently expand the derivative of the computational model in terms of partial derivatives of the basis (PoinCE-der), which yields another estimator for partial variances and Sobol' indices. 
If the partial derivatives of the model have smaller variability
than the model itself, the estimates based on the model derivatives might be more accurate. This makes PoinCE(-der) an efficient tool for screening \citep{rougam20}.

Our present contribution to the field of generalized chaos expansions and GSA is two-fold. 
On the theoretical side, we provide a proof that the Poincar\'e basis is in fact characterized uniquely as the orthonormal basis which remains an orthogonal basis (w.r.t.\ the same probability measure) after differentiation.
Furthermore, we show how PoinCE naturally generalizes an analytical formula for DGSM originally developed for Hermite PCE \citep{sudmai15}, which implies that PoinCE simultaneously and efficiently provides lower and upper bounds to all partial variances.
On the computational side, we improve on Roustant et al.~\cite{rougam20}, which introduced projection-based Poincar\'e chaos and demonstrated that small Sobol' indices were approximated particularly well by the derivative expansion. 
In this contribution, we compute PoinCE by sparse regression, thus generalizing the powerful and cost-effective sparse PCE methodology to non-polynomial functions.
We explore the performance of PoinCE as an estimator for partial variances (upper and lower bounds) and compare it to standard PCE.

This paper is organized as follows.
\Cref{sec:theory} revisits the mathematical foundations of PoinCE and presents several analytical results related to Sobol' indices and DGSM.
\Cref{sec:computation} explains the computation of PoinCE basis functions, and the sparse regression methodology adapted from PCE to PoinCE. 
The methodology is applied in \cref{sec:numerical}, where two example problems are investigated to demonstrate its performance for sensitivity analysis and screening.
Finally, we summarize our conclusions in \cref{sec:conclusion}.

\section{Mathematical background}
\label{sec:theory}

\subsection{Orthonormal bases in $L^2$}
\label{sec:orth_bases}
In this section, we recall some important facts about orthonormal bases in $L^2(E,\mu)$ where $E \subset \Rr^d$ and $\mu$ is a probability measure on $E$. We first outline the general theory in 
\cref{susugege}.
The particular cases of polynomial and Poincar\'e bases in several dimensions are developed in \cref{sec:PCE,sec:PoinCE}. 

\subsubsection{General theory}
\label{susugege}
To begin with, recall that $L^2(E,\mu)$ endowed with the inner pro\-duct
\begin{equation}
\innprod{f}{g}=\int_E f(x)g(x) \mu(dx),\;\; \text{ for } f,g\in L^2(E,\mu)
\label{eq:innprod}
\end{equation}
is a Hilbert space.
Recall that a sequence of functions $(\Phi_\alpha)_{\alpha\in \mathcal{I}}$ ($\mathcal{I} \subset \Nn$) is an \textit{orthonormal system} in $L^2(E,\mu)$ if it satisfies the two following assumptions:
\begin{itemize}
	\item[1)] For all $\alpha\neq\alpha'$, $\innprod{\Phi_\alpha}{\Phi_{\alpha'}}=0$, (orthogonality)
	\item[2)] For all $\alpha$, $\innprod{\Phi_\alpha}{\Phi_\alpha}=1$ (unit norm).
\end{itemize}
An orthonormal system in $L^2(E,\mu)$ is called \textit{complete} if the closure of the span generated by $(\Phi_\alpha)$ is $L^2(E,\mu)$.  In this case, the system $(\Phi_\alpha)$ is called an \textit{Hilbertian} or \textit{orthonormal basis} of $L^2(E,\mu)$ and for any function $f\in L^2(E,\mu)$ the following expansion holds:
\begin{equation}
\label{eq:chaos_expansion}
f=\sum_{\alpha}\innprod{\Phi_\alpha}{f}\Phi_\alpha,\;\;\;\; (\mu\mbox{ almost surely}).
\end{equation}

When used to represent random variables in terms of a basis of uncorrelated random variables, such an expansion is often called \textit{chaos expansion} 
in the uncertainty quantification literature
\citep{Wiener1938, Ghanembook1991, Ernst2012}.

An archetype example of chaos expansion is given by the so-called \textit{Fourier expansion}. This corresponds to the case where the set $E=[0,1]$ is endowed with the Lebesgue measure and we have for $\alpha\in\mathbb{Z}$,
\begin{equation*}
\Phi_\alpha(x)=\sqrt{2}\cos(2\pi\alpha x) \text{ if $\alpha < 0$, }
\Phi_0(x)=1 \text{, and } 
\Phi_\alpha(x)=\sqrt{2}\sin(2\pi\alpha x) \text{ if $\alpha > 0$.}
\end{equation*}
In this frame, any square-integrable function $f$ may be expanded as 
\begin{equation*}
f(x)=a_0+\sqrt{2}\sum_{\alpha>0}\left(a_\alpha\cos(2\pi\alpha x)+b_\alpha\sin(2\pi\alpha x)\right).
\end{equation*}
Here, for all $\alpha\in\mathbb{Z}_*$,
\begin{equation*}
a_0=\int_0^1f(x)dx, \;\;
a_\alpha = \sqrt{2} \int_0^1f(x)\cos(2\pi\alpha)dx,\; 
b_\alpha = \sqrt{2} \int_0^1f(x)\sin(2\pi\alpha)dx.
\end{equation*}

When the probability measure is a product measure $\mu = \mu_1 \otimes \dots \otimes \mu_d$ on a product space of intervals
$E=E_1\times E_2\times \cdots \times E_d$, 
there is a canonical way to build a Hilbertian basis from a collection of univariate Hilbertian ones. 
Indeed, for $i=1,\ldots, d$ assume that $(\Phi^{(i)}_{\alpha_i})$ is a Hilbertian basis of $L^2(E_i,\mu_i)$. Then, setting
$\alp:=(\alpha_1,\ldots,\alpha_d)$ and defining the tensor product functions $\Phal:=\prod_{i=1}^{d} \Phi^{(i)}_{\alpha_i}$,
we obtain that $(\Phal)$ is an orthonormal  basis of $L^2(E,\mu)$.

In the following we describe two particular chaos types, namely the classical \textit{polynomial chaos} and the recently developed \textit{Poincar\'e chaos}, for a probability measure $\mu$ on $E \subset \Rr$.

\subsubsection{Polynomial chaos}
\label{sec:PCE}

A classical family of chaos expansions on an interval $E$ of $\Rr$ endowed with a probability measure are \textit{polynomial chaos expansions} (PCE) given by orthonormal polynomial bases.
A well-known example is the Hermite expansion for which the set $E$ is the whole line $\Rr$ endowed with the standard Gaussian distribution. In this example, for $\alpha\in\Nn$, $\Phi_\alpha=H_{\alpha}$ is the Hermite polynomial of degree $\alpha$. The first Hermite polynomials are
\begin{equation*}
H_0(x)=1,\; H_1(x)=x,\; H_2(x)=\frac{x^2-1}{\sqrt{2}}, H_3(x)=\frac{x^3-3x}{\sqrt{6}},\;
H_4(x)=\frac{x^4-6x^2+3}{\sqrt{24}} \quad (x\in\Rr).
\end{equation*}

In general, there exists an orthonormal polynomial basis for $L^2(\mu)$ whenever the moment problem for $\mu$ is determinate. This includes the uniform, Gaussian, Beta and Gamma distributions, as well as all distributions with 
compact support \citep{Ernst2012}.

\subsubsection{Poincar\'e chaos}
\label{sec:PoinCE}

The Poincar\'e basis is another example of an orthonormal basis of $L^2(\mu)$, consisting of functions that admit weak derivatives, i.e. that belong to $H^1(\mu) = \{f \in L^2(\mu) \text{ s.t. } f' \in L^2(\mu) \}$.
Recall that $H^1(\mu)$, endowed with the norm $\Vert f \Vert_{H^1(\mu)}^2 = \Vert f \Vert^2 + \Vert f' \Vert^2$, is a Hilbert space. 
This short summary is based on Roustant et al.~\cite{roubar17} in which more details can be found. We assume that:

\begin{assum} \label{assum:measure}
	The probability measure $\mu$ is supported on a bounded interval $(a,b)$ and admits a density of the form $\rho = e^{-V}$, where $V$ is continuous and piecewise $C^1$ on $[a, b]$ with respect to the Lebesgue measure.
\end{assum}

This assumption is sufficient to guarantee the existence of a Poincar\'e basis. On the topological side, it implies that the Hilbert space $L^2(\mu)$ (resp. $H^1(\mu)$) is equal to $L^2(a, b)$ (resp. $H^1(a, b)$), with an equivalent norm. Indeed, $\mu$ is a bounded perturbation of the uniform measure on $[a, b]$, meaning that the pdf $\rho$ is bounded from below and above by strictly positive constants (by continuity of $V$ on the compact support $[a, b]$).

\begin{thm}[1D Poincar\'e basis]
	\label{thm:1DPoincare}
	Under \cref{assum:measure},  
	there exists an orthonormal basis $\left( \varphi_\alpha \right)_{\alpha \geq 0}$ of $L^2(\mu)$ such that for all $f \in H^1(\mu)$ and 
	for all integer $\alpha \geq 0$, we have: 
	\begin{equation} \label{eq:PoincareWeakForm1D}
	\innprod{ f'}{ \varphi_\alpha' } = \lambda_\alpha \innprod{ f}{ \varphi_\alpha },
	\end{equation}
	where $(\lambda_\alpha)_{\alpha \geq 0}$ is an increasing sequence that tends to infinity: 
	\begin{equation*}
	0 = \lambda_0 < \lambda_1 < \lambda_2 < \dots 
	< \lambda_\alpha \underset{\alpha \rightarrow \infty}{\longrightarrow} +\infty.
	\end{equation*}
	Here, the inner product $\innprod{\cdot}{\cdot}$ is the one on $L^2(\mu)$ as defined in \eqref{eq:innprod}. 
	The basis functions $\varphi_\alpha$ are unique up to a sign change, and form the so-called Poincar\'e basis. Notice that $\varphi_0$ is a constant function equal to $\pm 1$; by convention, we choose $\varphi_0 = 1$.\\
	Furthermore, the Poincar\'e basis functions are the eigenfunctions of the differential operator
	\begin{equation*}
	L(f) = f'' - V'f'
	\end{equation*}
	i.e. satisfy $L(f) = - \lambda f$, subject to Neumann conditions $f'(a) = f'(b) = 0$.
	The $(\lambda_\alpha)_{\alpha \geq 0}$ are the corresponding  eigenvalues. \\
	Finally, $\alpha^{-2} \lambda_\alpha \to \pi^2$ when $\alpha$ tends to infinity, and for all $\alpha \in \mathbb{N}^\star$, the eigenfunction $\varphi_\alpha$ has exactly $\alpha$ zeros in $(a, b)$.
\end{thm}

\begin{proof}
	The main part of the Theorem can be found in Roustant et al.~\cite{roubar17} or Bakry et al.~\cite{BGL_book}. 
	The two last assertions come by rewriting the differential equation $f'' - V'f' = - \lambda f$ in the Sturm-Liouville form
	\begin{equation}
	-(pf')' + qf = \lambda w f
	\label{eq:sturm-liouville}
	\end{equation}
	with $p =w = \rho$ and $q = 0$. Then, by the Sturm-Liouville theory (see e.g.\ Zettl~\citet[Theorem 4.3.1, (1), (6) and (7)]{zettl2010sturm}), we have that $\alpha^{-2} \lambda_\alpha \to \pi^2$ when $\alpha$ tends to infinity, and for all $\alpha \in \mathbb{N}^\star$, the eigenfunction $\varphi_\alpha$ has exactly $\alpha$ zeros in $(a, b)$.
\end{proof}

The Poincar\'e basis shares some similarity with both the polynomial chaos and the Fourier basis in terms of oscillations:
by \cref{thm:1DPoincare}, the higher the order of the eigenvalue, the more oscillating the corresponding eigenfunction.  

For some specific cases, the Poincar\'e basis is known analytically. For instance, for the uniform distribution, the Poincar\'e basis is a kind of Fourier basis (see e.g.\ Roustant et al.~\citet[\S 4]{rougam20}). Otherwise it has to be computed numerically, e.g., by a finite element technique (see \cref{sec:implementation}). 

Note that \cref{assum:measure} is a convenient sufficient condition which guarantees the existence of a Poincar\'e basis.
It is satisfied for a large range of truncated parametric probability distributions. The set of probability distributions for which the Poincar\'e chaos exists is larger, but not well known.
For instance, the Poincar\'e chaos is defined for the Gaussian distribution, and then coincides with polynomial chaos, corresponding to Hermite polynomials. This is the only case where Poincar\'e chaos and polynomial chaos coincide \citep[\S 2.7]{BGL_book}. On the other hand, Poincar\'e chaos is not defined for the Laplace distribution, since the eigenvalues of the associated operator do not form a countable set \citep[\S 4.4.1]{BGL_book}. 

The Poincar\'e basis is useful for sensitivity analysis. First, it is linked to the Poincar\'e inequality
\begin{equation} \label{eq:PoincareIneq}
\Var_{\mu}(f) \leq C_P(\mu) \int f'^2 d\mu,
\end{equation}
which holds for all functions $f \in H^1(\mu)$ under the assumptions on $\mu$. Indeed, the smallest constant $C_P(\mu)$ such that \eqref{eq:PoincareIneq} is satisfied is equal to $C_P(\mu) = 1/\lambda_1$, and choosing $f = \varphi_1$ corresponds to the equality case \citep{roubar17}. Roughly speaking, the Poincar\'e basis function associated to the first non-zero eigenvalue is the function with the largest possible variance for a given amount of integrated squared derivative (in the $L^2$ sense). 
A second appealing property for the analysis of variance is that the derivatives of the Poincar\'e basis remain orthogonal functions:
\begin{prop} \label{prop:PoincDerOrtho1D}
	Under \cref{assum:measure}, the sequence $\left( \frac{1}{\sqrt{\lambda_\alpha}} \varphi_\alpha' \right)_{\alpha \geq 1}$ is an orthonormal basis of $L^2(\mu).$ 
\end{prop}
\begin{proof}
	The orthonormality of the sequence is a consequence of \eqref{eq:PoincareWeakForm1D} by choosing $f = \varphi_{\beta}$, with $\beta \in \Nn^*.$
	It remains to show that the system is dense in $L^2(\mu)$, or equivalently, that its orthogonal complement is the null element. Let thus $f \in L^2(\mu)$ such that 
	\begin{equation*}
	\innprod{ f}{ \varphi_\alpha' } = 0, \quad \text{ for all } \alpha \geq 1.
	\end{equation*}
	As explained when stating \cref{assum:measure}, $L^2(\mu)$ (resp.\ $H^1(\mu)$) is equal to $L^2(a, b)$ (resp. $H^1(a, b)$), with an equivalent norm. 
	Now, there exists $g \in H^1(\mu)$ such that $f = g'$. 
	Indeed, let us define $g$ by $g(x) = g(a) + \int_{a}^x f(t) dt$.
	As $f \in L^2(\mu) = L^2(a, b)$, then $g$ belongs to $H^1(a, b) = H^1(\mu)$, and $g' = f$.
	Then we have 
	\begin{equation*}
	\innprod{ g'}{ \varphi_\alpha' } = 0, \quad \text{ for all } \alpha \geq 1.
	\end{equation*}
	By \eqref{eq:PoincareWeakForm1D}, we obtain $\innprod{ g}{ \varphi_\alpha } = 0$ for all $\alpha \geq 1$ (as $\lambda_\alpha > 0$ for $\alpha \geq 1$).
	As the functions $\varphi_\alpha$ form an orthonormal basis of $L^2(\mu)$ with $\varphi_{0} = 1$, this implies that $g$ is a constant function, and finally $f=0$. The proof is completed.
\end{proof}
In fact, the property in \cref{prop:PoincDerOrtho1D}, i.e., that the derivatives of the Poincar\'e basis form again an orthogonal basis in $L^2(\mu)$, uniquely characterizes the Poincar\'e basis: 
\begin{prop} \label{prop:PoincareBasesCaracterisation}
	Under \cref{assum:measure}, Poincar\'e bases are the only orthonormal bases $(\varphi_\alpha)$ of $L^2(\mu)$ in $H^1(\mu)$ such that $(\varphi'_\alpha)$ is an orthogonal basis of $L^2(\mu)$.
\end{prop}

This result seems difficult to find in the literature. German-speaking readers can find a similar proposition in Mikolas~\cite{mikolas1955}, stated in the frame of Sturm-Liouville theory for twice-differentiable functions satisfying boundary conditions. 
See also Kwon and Lee~\citet{Kwon2003} for a similar result under the assumption that all functions involved in the Sturm-Liouville problem \cref{eq:sturm-liouville} are of class $C^\infty$.
We provide below a proof based on Hilbertian arguments. 

As a corollary, if there exists a basis different from the Poincar\'e basis for which derivatives form an orthogonal system, then that system is not dense in $L^2(\mu)$. As an example, for the uniform probability measure on $[0, 2\pi]$, consider the usual Fourier basis formed by $\{\cos(nx), \sin(nx): n \geq 0 \}$ (up to multiplicative constants). Taking derivatives results in the same set of functions (up to multiplicative constants) -- except for the constant function $\cos(0x) = 1$.
Thus, the derivatives form an orthogonal system which covers the orthogonal of constant functions in $L^2(\mu)$, which is a strict subspace of $L^2(\mu)$. Meanwhile, the Poincar\'e basis for this probability measure is formed by functions proportional to $\cos\left(\frac{n}{2}x \right)$ for $n \geq 0$.  \Cref{prop:PoincareBasesCaracterisation} guarantees that all functions of $L^2(\mu)$, including the constant functions, are spanned by the derivatives. Indeed, this is explained intuitively by the presence of half-frequencies: when $n$ is odd, the functions $\sin \left(\frac{n}{2}x \right)$ are not orthogonal to $1$.

\begin{proof}[Proof of \cref{prop:PoincareBasesCaracterisation}]
	The fact that a Poincar\'e basis remains an orthogonal basis by derivation has been proved in \cref{prop:PoincDerOrtho1D}.
	Conversely, let $(\varphi_\alpha)_{\alpha \geq 0}$ be a system of $H^1(\mu)$, with $\varphi_0 = 1$, such that $(\varphi_\alpha)$ is an orthonormal basis of $L^2(\mu)$ and $(\varphi'_\alpha)_{\alpha \geq 1}$ is an orthogonal basis of $L^2(\mu)$. 
	Let us first prove that $(\varphi_\alpha)$ is an orthogonal basis of $H^1(\mu)$. The orthogonality is a direct consequence of the definition of the inner product of $H^1(\mu)$:
	\begin{equation*}
	\innprod{ \varphi_\alpha}{ \varphi_\beta }_{H^1(\mu)} 
	= \innprod{ \varphi_\alpha}{\varphi_\beta }_{L^2(\mu)}
	+ \innprod{ \varphi'_\alpha}{ \varphi'_\beta }_{L^2(\mu)} 
	= (1 + \Vert \varphi'_\alpha \Vert_{L^2(\mu)}^2) \delta_{\alpha, \beta}.
	\end{equation*}
	Let us prove that $(\varphi_\alpha)$ is dense in $H^1(\mu)$. As explained when stating \cref{assum:measure}, $L^2(\mu)$ (resp. $H^1(\mu)$) is equal to $L^2(a, b)$ (resp. $H^1(a, b)$), with an equivalent norm. Hence, it is equivalent to prove that $(\varphi_\alpha)$ is dense in $H^1(a,b)$. Now, let $f$ be in $H^1(a, b)$. As $(\varphi'_\alpha)$ is dense in $L^2(a, b)$ (equivalently in $L^2(\mu)$), then $f'$ expands as
	$f' = \sum_{\alpha \in \Nn} c_\alpha \varphi'_\alpha$.
	In $H^1(a,b)$ each function is equal to the primitive function of its derivative, hence we have: 
	\begin{eqnarray*}
		\left \vert f(t) - f(a) - \sum_{\alpha=1}^N c_\alpha (\varphi_\alpha(t) - \varphi_\alpha(a)) \right \vert 
		& = & \left \vert \int_a^t \left( f'(x) - \sum_{\alpha=1}^N c_\alpha \varphi'_\alpha(x) \right) dx \right \vert \\
		& \leq & (b - a) \left \Vert f' - \sum_{\alpha=1}^N c_\alpha \varphi'_\alpha \right \Vert_{L^2(a, b)} 
	\end{eqnarray*}
	where the inequality comes from the Cauchy-Schwarz inequality.
	We deduce that 
	$\Vert f - f(a) - \sum_{\alpha = 1}^N c_\alpha (\varphi_\alpha - \varphi_\alpha(a)) \Vert_{L^2(a, b)} \to 0$ when $N$ tends to infinity. Together with $f' = \sum_{\alpha \in \Nn} c_\alpha \varphi'_\alpha$, this implies that $\Vert f - f(a) - \sum_{\alpha = 1}^N c_\alpha (\varphi_\alpha - \varphi_\alpha(a)) \Vert_{H^1(a, b)} \to 0$. As $\varphi_0 = 1$, this proves that $(\varphi_\alpha)$ is dense in $H^1(a,b)$, which was to be proved.\\
	Now, let us fix $\alpha \geq 0$. Consider the linear form $L_\alpha$ defined on $H^1(\mu)$ by $L_\alpha(f) = \innprod{ f'}{ \varphi'_\alpha }_{L^2(\mu)}$. 
	The Cauchy-Schwarz inequality gives 
	$\vert L_\alpha(f) \vert 
	\leq 
	\Vert f' \Vert_{L^2(\mu)}
	\Vert \varphi'_\alpha \Vert_{L^2(\mu)}
	\leq 
	\Vert f \Vert_{H^1(\mu)}
	\Vert \varphi'_\alpha \Vert_{L^2(\mu)}$.
	This proves that $L_\alpha$ is continuous.
	Hence, by the Riesz representation theorem, there exists a unique $\zeta_\alpha \in H^1(\mu)$ such that for all $f \in H^1(\mu)$,  $L_\alpha(f) = \innprod{ f}{ \zeta_\alpha }_{H^1(\mu)}$, i.e. 
	$\innprod{ f'}{ \varphi'_\alpha }_{L^2(\mu)}
	= \innprod{f}{ \zeta_\alpha }_{H^1(\mu)}.$
	Choosing $f = \varphi_\beta$ with $\beta \neq \alpha$, we obtain by orthogonality of $(\varphi'_\alpha)$ that for all $\beta \neq \alpha$, 
	$\innprod{ \varphi_\beta}{ \zeta_\alpha }_{H^1(\mu)} = 0$.
	As $(\varphi_\beta)_{\beta \geq 0}$ is an orthogonal basis of $H^1(\mu)$, this implies that $\zeta_\alpha$ is collinear to $\varphi_\alpha$, i.e., there exists $\Tilde{\lambda}_\alpha \in \mathbb{R}$ such that $\zeta_\alpha = \Tilde{\lambda}_\alpha \varphi_\alpha$. 
	Thus, for all $f \in H^1(\mu)$, we have 
	$\innprod{ f'}{ \varphi'_\alpha }_{L^2(\mu)}
	= \Tilde{\lambda}_\alpha \innprod{ f}{ \varphi_\alpha }_{H^1(\mu)}.$
	Choosing $f = \varphi_\alpha$, we get 
	$\Tilde{\lambda}_\alpha = \frac{\Vert \varphi'_\alpha \Vert_{L^2(\mu)}^2} 
	{1 + \Vert \varphi'_\alpha \Vert_{L^2(\mu)}^2}$, which belongs to $[0, 1)$. 
	Finally, we obtain that 
	$\innprod{ f'}{ \varphi'_\alpha }_{L^2(\mu)}
	= \lambda_\alpha \innprod{ f}{ \varphi_\alpha }_{L^2(\mu)},$
	where $\lambda_\alpha = \frac{\Tilde{\lambda}_\alpha}{1 - \Tilde{\lambda}_\alpha}$ is a non-negative real number.
	As it is true for all $f$ in $H^1(\mu)$ and all $\alpha \in \mathbb{N}$, this implies, by uniqueness of the Poincar\'e basis (under \cref{assum:measure}), that $(\varphi_\alpha)_{\alpha \geq 0}$ is a Poincar\'e basis.
\end{proof}

Turning to higher dimensions, we assume that for all $i=1, \dots, d$, the probability measure $\mu_i$ satisfies \cref{assum:measure}, and we denote by $(\vpali{i})_{\alpha_i \geq 0}$ the sequence of 1-dimensional Poincar\'e basis functions, and by $(\lambda_{i, \alpha_i})_{\alpha_i \geq 0}$ the sequence of associated eigenvalues. The Poincar\'e chaos basis is then defined by the tensor product $\Phal = \vpali{1} \otimes \dots \otimes \vpali{d}$. Using the properties of $L^2$ bases, \eqref{eq:PoincareWeakForm1D} thus implies that for all $f \in H^1(\mu)$, for all $i = 1, \dots, d$:
\begin{equation} \label{eq:PoincareWeakForm}
\innprod{\frac{\partial f}{\partial x_i}}{\frac{\partial \Phal}{\partial x_i}} = \lambda_{i, \alpha_i} \innprod{f}{\Phal}.
\end{equation}
Similarly, applying \cref{prop:PoincDerOrtho1D}, we get:
\begin{prop} \label{prop:PoincDerOrtho}
	Let $\{\Phal\}_{\alp}$ be a multivariate Poincar\'e chaos basis.
	Under \cref{assum:measure}, for all $i=1, \dots, d$, the sequence $\displaystyle \left( \frac{1}{\sqrt{\eig{i}}}\frac{\partial \Phal}{\partial x_i} \right)_{\alp, \alpha_i \geq 1}$ is an orthonormal basis of $L^2(\mu)$. 
\end{prop}

\subsection{Variance-based indices, derivative-based indices}
We first recall the definition of variance-based sensitivity indices, which quantify the importance of each input variable in terms of function response variability.

Let $f$ be a real-valued function defined on $E = E_1 \times \dots \times E_d \subseteq \mathbb{R}^d$. The uncertainty of the inputs is represented by a random vector $\ve X = (X_1, \dots, X_d)^T$ with probability measure $\mu$ on $E$. We further assume that the $X_i$'s are independent and that $f(\ve X)$ belongs to $L^2(E, \mu)$. Denoting by $\mu_i$ the marginal distribution of $X_i$ on $E_i$ ($i=1, \dots, d$), we then have $\mu = \mu_1 \otimes \dots \otimes \mu_d$.
In this framework, $f(\ve X)$ can be decomposed uniquely as a sum of terms of increasing complexity
\begin{equation} \label{eq:FANOVA}
f(\ve X) = f_0 + \sum_{1 \leq i \leq d} f_i(X_i) + \sum_{1 \leq i <j \leq d} f_{i,j}(X_i, X_j) 
+ \dots + f_{1, \dots, d}(X_1, \dots, X_d)
\end{equation}
under centering conditions $\Esp{f_I(X_I)} = 0$ and non-overlapping conditions $\Esp{f_I(X_I) \vert X_J} =0$, for all sets $I \subseteq \{1, \dots, d \}$ and all strict subsets $J$ of $I$. We have used the set notation $X_I$ to represent the subvector of $\ve X$ obtained by selecting the coordinates belonging to $I$. These conditions imply that all the terms of \eqref{eq:FANOVA} are orthogonal, leading to the variance decomposition
\begin{equation} \label{eq:ANOVA}
\Var{f(\ve X)} = \sum_{1 \leq i \leq d} \Var{f_i(X_i)} + \sum_{1 \leq i <j \leq d} \Var{f_{i,j}(X_i, X_j)} 
+ \dots + 
\Var{f_{1, \dots, d}(X_1, \dots, X_d)}
\end{equation}
Due to this property, the functional decomposition \eqref{eq:FANOVA} is often called ANOVA (ANalysis Of VAriance) decomposition. Originating from Hoeffding~\citet{Hoeffding_1948}, it was revisited by Efron and Stein~\citet{efrste81}, Antoniadis~\citet{Antoniadis1984}, and Sobol and Gresham~\citet{sob93}. 
For a given set $I \subseteq \{1, \dots, d\}$, we call the corresponding term of \eqref{eq:ANOVA} \textit{partial variance} (denoted $D_I$), and call its normalized version \textit{Sobol' index} (denoted $S_I$):
\begin{equation*}
D_I = \Var(f_I(X_I)), \qquad S_I = \frac{D_I}{D},
\end{equation*}
where $D = \Var{f(\ve X)}$ is the overall variance (\textit{total variance}). 
In particular, for $i \in \{1, \dots, d\}$, the first-order Sobol' index $S_i$  corresponds to the proportion of variance of $f(\ve X)$ explained by $X_i$ only. In order to include also the interactions of $X_i$ with the other variables, the total partial variance and the total Sobol' index are defined by
\begin{equation*}
D_i^\textrm{tot} = \sum_{I \supseteq \{i \}} \Var(f_I(X_I)), \qquad S_i^\textrm{tot} = \frac{D_i^\textrm{tot}}{D}.
\end{equation*}
Note that practitioners also call the (total) partial variances unnormalized (total) Sobol' indices. In the sequel, we will use these two words interchangeably.\\
The total Sobol' index can be used for screening. Indeed, under mild conditions, if $S_i^\textrm{tot} = 0$ then the function $f$ does not depend on $x_i$ over $E$ (in the pointwise sense).\\

When the derivatives are available, a global sensitivity index can be obtained by integration. The so-called \textit{derivative-based sensitivity measure} (DGSM) index of $f$ with respect to $X_i$ \citep{Sobol_Gresham_1995,Kucherenko2009} is defined by 
\begin{equation} \label{eq:DGSM}
\nu_i = \Esp{\left( \frac{\partial f}{\partial x_i}(\ve X) \right)^2 } 
= \int_{E} \left( \frac{\partial f}{\partial x_i}(\ve x) \right)^2 d\mu(\ve x) 
= \left\Vert \frac{\partial f}{\partial x_i} \right\Vert^2 .
\end{equation}

Contrarily to variance-based indices, DGSM are not associated to a variance decomposition. Nevertheless, they can be used for screening. Indeed, under mild conditions, $\nu_i = 0$ implies that $f$ does not depend on $x_i$ over $E$.

\subsection{Chaos expansion serving  sensitivity analysis}
One main advantage of using an orthonormal basis for sensitivity analysis is that, once the expansion has been obtained, the variance-based indices can be computed in a straightforward way as a sum of squared coefficients \citep{SudretCSM2006, sud08}. More precisely, let $f$ be in $L^2(\mu)$, and let $(\Phal)_{\alp\in \Nn^d}$ be a multivariate orthonormal basis obtained by tensorization as described in \cref{sec:orth_bases}.
The multi-indices $\alp \in \Nn^d$ are obtained from the enumeration of the univariate bases as described in \cref{sec:orth_bases}. We assume that each univariate basis contains the constant function, which is without loss of generality given the index zero. This assumption is fulfilled for PCE and for Poincar\'e chaos.
The expansion of $f$ in this basis is given by
\begin{equation} \label{eq:ChaosCoef}
f = \sum_{\alp \in \Nn^d} c_\alp \Phal.
\end{equation}
By using the orthonormality we obtain the expression of the total variance
\begin{equation} \label{eq:VarWithChaos}
D = \sum_{\alp \neq \mathbf{0}} c_\alp^2.
\end{equation}
The expression of the total Sobol' index $S_i^{\rm{tot}}$ is obtained by only considering the terms of the decomposition \eqref{eq:ChaosCoef} that contain the variable $x_i$, i.e. such that $\alpha_i \geq 1$. Hence, we have
$S_i^{\rm{tot}} = \frac{D_i^{\rm{tot}}}{D}$ with $D_i^{\rm{tot}}$ the total partial variance
\begin{equation} \label{eq:VarTotWithChaos}
D_i^{\rm{tot}} = \sum_{\alp, \alpha_i \geq 1} c_\alp^2.
\end{equation}
The first-order Sobol' index $S_i^1$ relies on the terms that include $x_i$ only, i.e., $S_i^{1} = \frac{D_i^{1}}{D}$ with
\begin{equation} \label{eq:VarFirstWithChaos}
D_i^{1} = \sum_{\substack{\alp, \alpha_i \geq 1, \\ \alpha_j = 0 \text{ for }j\neq i}} c_\alp^2.
\end{equation}

Let us now consider the case where the gradient of $f$ is available. The Poincar\'e basis is particularly suited to this situation. Indeed, we can derive in a straightforward way expressions of both variance-based and derivative-based indices, involving the derivatives of $f$.
Due to orthonormality, the coefficients of the basis expansion in \eqref{eq:ChaosCoef} are given by the projection of $f$ onto the associated basis element:
\begin{equation}
\label{eq:coeff}
c_\alp = \innprod{f}{\Phal}.
\end{equation}

From now on, let $(\Phal)_\alp$ denote the Poincar\'e basis. 
Combining \eqref{eq:PoincareWeakForm} and \eqref{eq:coeff}, and assuming that $\alpha_1\geq 1$, $c_\alp$ can be written using the partial derivatives w.r.t\ variable $X_1$ \citep{rougam20}:
\begin{equation} \label{eq:ChaosCoefWithDer}
c_\alp = \innprod{f}{\Phal} 
= \frac{1}{\eig{1}} 
\innprod{\frac{\partial f}{\partial x_1}}
{\frac{\partial \Phal}{\partial x_1}}
= \frac{1}{\eig{1}}  
\innprod{ \frac{\partial f}{\partial x_1}}
{\frac{\partial \vpali{1}}{\partial x_1} \otimes \vpali{2} \otimes \dots \otimes \vpali{d}}
\end{equation}
and equivalently using partial derivatives w.r.t\ variable $X_i$ if $\alpha_i \geq 1$.
Thus, \eqref{eq:VarWithChaos}, \eqref{eq:VarTotWithChaos} and \eqref{eq:VarFirstWithChaos} can also be computed using 
the various partial derivatives of $f$. Whereas the theoretical expressions are equal, their estimators have different properties. 
For example, if the integral is evaluated by Monte Carlo simulation, the expression whose integrand has smaller variance will be more accurate.
We describe in \cref{sec:computation} the computation of the expansion coefficients by regression,
and we empirically compare the two estimation procedures in \cref{sec:numerical}.

Furthermore,  
DGSM can be computed directly from the Poincar\'e expansion. More precisely, we have the following proposition.
\begin{prop}[DGSM formula for Poincar\'e chaos]
	Let $f \in H^1(\mu)$. 
	Let $f = \sum_{\alp} c_\alp \Phal$ be the expansion of $f$ in the Poincar\'e chaos basis, with ${c_\alp = \innprod{ f}{ \Phal} }$.
	Then the DGSM index of $f$ with respect to $X_i$ is equal to:
	\begin{equation} \label{eq:DGSM_Poincare}
	\nu_i = \sum_{\alp, \, \alpha_i \geq 1} \lambda_{i, \alpha_i} \left(c_\alp \right)^2.
	\end{equation}
\end{prop}

\begin{proof}
	Write $f = \sum_{\alp} c_\alp \Phal$. Then by \cref{prop:PoincDerOrtho}, we get
	\begin{equation*}
	\frac{\partial f}{\partial x_i} = \sum_{\alp, \, \alpha_i \geq 1} c_\alp \frac{\partial \Phal}{\partial x_i},
	\end{equation*}
	where we can constrain the sum to multi-indices $\alp$ such that $\alpha_i \geq 1$, since ${\varphi_{i, \alpha_i} = 1}$ for $\alpha_i = 0$.
	Now, using again the orthogonality of Poincar\'e basis derivatives (\cref{prop:PoincDerOrtho}), it follows that
	\begin{equation*}
	\nu_i = \left\Vert \frac{\partial f}{\partial x_i} \right\Vert ^2
	= \sum_{\alp, \, \alpha_i \geq 1} (c_\alp)^2 \left\Vert \frac{\partial \Phal}{\partial x_i} \right\Vert^2
	= \sum_{\alp, \, \alpha_i \geq 1} \lambda_{i, \alpha_i} \left(c_\alp \right)^2.
	\end{equation*}
	
\end{proof}

Formula \eqref{eq:DGSM_Poincare} extends a previous result given by Sudret and Mai~\cite{sudmai15} when all the $\mu_i$ are standard Gaussian. 
Indeed, in that case, Poincar\'e chaos coincides with polynomial chaos, and $\lambda_{i, \alpha_i} = \alpha_i$.

Using the expressions provided in \eqref{eq:VarTotWithChaos} and \eqref{eq:DGSM_Poincare} and an inequality derived by Sobol' and Kucherenko~\citet{Sobol2009} and Lamboni et al.~\citet{Lamboni_et_al_2013}, we obtain lower and upper bounds to total partial variances as follows:
\begin{equation}
\sum_{\alp \in \curlyA, \alpha_i \geq 1} \left(c_\alp \right)^2
\ \leq \ D_i^{\rm{tot}} 
\ \leq \ C_P(\mu_i) \nu_i 
\ = \sum_{\alp \in \Nn^d, \, \alpha_i \geq 1} \frac{\lambda_{i, \alpha_i}}{\lambda_{i, 1}} \left(c_\alp \right)^2,
\label{eq:DGSM_poincare_upper_bound}
\end{equation}
where $\curlyA \subset \Nn^d$ is the subset of multi-indices included in the truncated expansion (see \cref{sec:truncation}).
The lower bound is an obvious consequence of the truncation. 
The upper bound holds only for the full infinite expansion and is otherwise underestimated.
Comparing the form of the right-hand side of \cref{eq:DGSM_poincare_upper_bound} with the total Sobol' formula \cref{eq:VarTotWithChaos} gives insight into how tight this upper bound is: equality is attained only if the Poincar\'e chaos expansion does not contain terms of higher degree than $1$ for $X_i$ (then, $\frac{\lambda_{i, \alpha_i}}{\lambda_{i, 1}}=1$).
Else, depending on the decay behavior of $c_\alp$
the gap can be significant, since the eigenvalues 
are diverging to infinity
(see \cref{thm:1DPoincare}).

\section{Computation of sparse Poincar\'e expansions}
\label{sec:computation}

Let $f \in H^1(\mu, E)$ be a computational model defined on the input space $E \subset \Rr^d$, with independent input random variables and with the input probability measure $\mu$ admitting a probability density function $\rho$ fulfilling \cref{assum:measure} for each of the marginals.
In the remainder of this paper, we assume that $\rho$ is known.
We also assume that we are provided with an i.i.d.\ sample from the input distribution and with the corresponding model evaluations and model gradient values at each of the points.

With \textit{Poincar\'e expansion} (PoinCE) we denote the expansion of the computational model onto the Poincar\'e basis  
\begin{equation}
f(\ve x) = \sum_{\alp} c_\alp \Phal(\ve x),
\label{eq:Poincare_expansion}
\end{equation}
and with \textit{Poincar\'e derivative expansion in direction $i$} (PoinCE-der-$i$) the expression
\begin{equation}
\frac{\partial f}{\partial x_i}(\ve x) = \sum_{\alp, \alpha_i \geq 1} {c}^{\, \partial, i}_\alp \ \frac{\partial \Phal}{\partial x_i}(\ve x),
\label{eq:Poincare_deriv_expansion}
\end{equation}
or the equivalent expansion using normalized basis derivatives that have unit norm in $L^2(\mu)$.
Note that \eqref{eq:Poincare_deriv_expansion} is the partial derivative of \eqref{eq:Poincare_expansion} w.r.t.\ variable $X_i$. Because the zeroth order basis function of a Poincar\'e basis is the constant function, basis terms for which $\alpha_i = 0$ have zero partial derivative w.r.t.\ $X_i$ and are not included in \eqref{eq:Poincare_deriv_expansion}.
While in theory by Equation \eqref{eq:ChaosCoefWithDer}, the two expressions \eqref{eq:Poincare_expansion} and \eqref{eq:Poincare_deriv_expansion} provide identical coefficients for corresponding basis elements, i.e., $c_\alp = {c}^{\, \partial, i}_\alp$ for $\alp \in \{\alp' \in \curlyA: \alpha_i' \geq 1\}$, in practice they will not coincide when estimated from a data set of finite size.
This will be investigated in \cref{sec:numerical} for a number of numerical examples.

In this section, we describe how such expansions are computed in practice: this concerns the computation of the Poincar\'e basis functions, the choice of truncation, the location of the sampled points, and the method for computing the coefficients.
The implementation relies on and integrates into the UQLab framework \citep{MarelliUQLab2014}.

\subsection{Implementation of Poincar\'e basis functions}
\label{sec:implementation}
As described in \cref{sec:PoinCE}, Poincar\'e basis functions are tensor products of univariate Poincar\'e basis functions. Each 1D basis consists of the eigenfunctions of the Poincar\'e differential operator associated with the respective marginal distribution (\cref{thm:1DPoincare}).

A Poincar\'e basis is guaranteed to exist for marginal distributions fulfilling \cref{assum:measure} and for the Gaussian distribution. 
Other distributions have to be transformed or truncated to allow for a Poincar\'e basis. Since an isoprobabilistic transformation to standard variables can be highly nonlinear \citep{TorreJCP2019,Oladyshkin2012},
we opt for truncation:
if the distribution is not Gaussian and has (one- or two-sided) unbounded support, we truncate it to its $\num{e-6}$- and $(1-\num{e-6})$-quantiles, respectively.%

\begin{rem}[Truncation]
	One might argue that this can distort the results obtained with PoinCE, especially in the tails. It is true that this truncation introduces a small error. However, as all such methods, PoinCE by design approximates accurately mainly the bulk, not the tails (for this, specialized techniques like subset simulation shall be used). Furthermore, in practical applications it is a modelling choice how to represent the input distribution. Choosing an unbounded parametric distribution is common, but not necessarily the most sensible choice, since for virtually every quantity in the real world there is an upper bound that cannot be exceeded.
\end{rem}

We consider here only standard parametric families of probability densities (bounded and unbounded), although a Poincar\'e basis can be computed for any input distribution which after truncation fulfills \cref{assum:measure}. In particular, without any changes to the methodology PoinCE could be used in a data-driven framework \citep{TorreJCP2019} by computing the Poincar\'e basis for a dimensionwise kernel density estimate of the input distribution  (assuming independence) given the available data.

As can be seen from applying the change-of-variables formula for a linear transformation to \eqref{eq:PoincareWeakForm1D}, the eigenvalues of the Poincar\'e differential operator scale with the inverse of the squared support interval length. To avoid numerical difficulties, we therefore linearly transform (i.e., shift and rescale) parametric families to standard parameters using
\begin{itemize}[itemsep=0pt, parsep=0pt,topsep=0pt]
	\item their bounds in the case of uniform, beta, triangular;
	\item their location and scale parameter in the case of Gaussian, Gumbel, Gumbel-min, Laplace, logistic;
	\item their (inverse) scale parameter in the case of exponential, gamma, Weibull, lognormal.
\end{itemize}
In the current implementation, distributions not belonging to this group of families are not being rescaled. 

For standard uniform ($\cu([-0.5, 0.5]$) and standard Gaussian ($\cn(0,1)$) marginals, the Poincar\'e basis can be analytically computed and is given by the Fourier (cosine) basis and the Hermite polynomial basis, respectively \citep{rougam20}.
Therefore, in the special case of uniform or Gaussian marginals, we always (after rescaling) use the analytical solution.

For all other marginals, the Poincar\'e basis is computed numerically using linear finite elements.
We use a fine uniform grid within the bounds and piecewise linear functions with local support, commonly called `hat' functions.
Using the weak formulation of the eigenvalue problem of the Poincar\'e differential operator given in \eqref{eq:PoincareWeakForm1D}, we arrive at the shifted generalized eigenvalue problem
\begin{equation}
\ve K \ve a^{(n)} = (\lambda_n+1) \ve M \ve a^{(n)}
\end{equation}
as described in Roustant et al.~\citet[section 4.3]{roubar17}, where the eigenvector $\ve a^{(n)}$ denotes the vector of coefficients used to express eigenfunction $\varphi_n$ in terms of `hat' functions. Here $\ve M$ is the mass matrix, and $\ve K$ is the sum of mass- and stiffness matrix. 
After solving this problem using Matlab's builtin function \texttt{eigs}, we interpolate the discrete eigenvectors with piecewise cubic splines,
prescribing zero derivatives at the interval boundaries. Then, the basis derivatives are computed using centered finite differences.
While more sophisticated techniques (e.g., Hermitian $C^1$ elements, or Haar wavelets \citep{Bujurke2008}) could of course be used to improve this numerical computation procedure, it is accurate enough for our purposes of demonstrating the usefulness of PoinCE.
Eigenfunctions and eigenfunction derivatives are scaled to have unit norm with respect to the measure $\mu_i$.

\subsection{Choice of the basis truncation}
\label{sec:truncation}

In practice, the series in \eqref{eq:Poincare_expansion} and \eqref{eq:Poincare_deriv_expansion} cannot include an infinite number of terms, but must be truncated to a finite expansion. We denote by $\curlyA \subset \Nn^d$ the subset of multi-indices that are included in the expansion. For PCE, $\curlyA$ is typically chosen to include terms up to a certain degree $p$, resulting in the so-called \textit{total degree basis}
\begin{equation}
\curlyA^p = \{\alp \in \Nn^d: \sum_{i=1}^d |\alpha_i| \leq p\}
\end{equation}
containing $P = \binom{p+d}{d}$ polynomials \citep{SudretREEF2006,sud08}.
To further restrict the number of terms used in the expansion, another common truncation method is \textit{hyperbolic truncation} \citep{blasud11} 
\begin{equation}
\curlyA^{p,q} = \{\alp \in \Nn^d: \norme{\alp}{q} \leq p\}.
\end{equation}
with the $\ell^q$-(quasi-)norm $\norme{\alp}{q} = \left( \sum_{i=1}^d \alpha_i^q \right)^\frac{1}{q}$ for $q\in (0,1]$.

Since the Poincar\'e basis is in general not polynomial, the concept of polynomial degree cannot be used to characterize the basis functions. Instead, we use the natural order of the basis functions corresponding to the increasing sequence of Poincar\'e eigenvalues, which also corresponds to an increasing number of oscillations (\cref{thm:1DPoincare}; recall that the $n$th eigenfunction has $n$ zeros). 
Therefore, we use the PCE terminology ``degree'' also for PoinCE. In particular, a degree of $\alpha_i = 0$ denotes the constant basis function $\varphi_{i, 0}(x_i) = 1$ associated to the eigenvalue $\lambda_{i, 0} = 0$.

Often, in practice it is not known which degree is needed for a given problem.
While in theory the expansion is more accurate the larger the total degree is, in practice accuracy is limited by the number of available sample points, since the quality of the regression solution (see \cref{sec:sparse_regression}) depends on the ratio of sample points to basis elements.
In that case, a successful strategy consists of applying \textit{degree adaptivity}, i.e., choosing the best degree for the expansion by cross-validation \citep{blasud11, LuethenIJUQ2022}.
This procedure is computationally inexpensive, since it only requires a new surrogate model fit for each new total degree, but no additional model evaluations.
We apply leave-one-out (LOO) cross-validation together with a modification factor introduced by Chapelle et al.~\citet{Chapelle2002, blasud11}.

Both hyperbolic truncation and degree adaptivity contribute to the sparsity of the resulting expansion by identifying a suitable subset of basis functions necessary for a good approximation.
Sparsity is a successful concept in regression-based PCE \citep{LuethenSIAMJUQ2020}.
Denote by $P = |\curlyA|$ the number of basis elements in the truncated expansion.
$\curlyA$, also called \textit{candidate basis}, contains the basis elements available for approximation.
We describe below how sparse regression further selects only a subset of $\curlyA$ to be \textit{active}, i.e., have a nonzero coefficient. The final expansion might (and indeed often will) have less than $P$ active terms.

\subsection{Computation of the coefficients by sparse regression}
\label{sec:sparse_regression}
For computing the coefficients of an orthogonal expansion as in \eqref{eq:Poincare_expansion} and \eqref{eq:Poincare_deriv_expansion}, there exist two main approaches. 
One is \textit{projection}: the model $f$ is projected onto the basis functions, see \eqref{eq:chaos_expansion}. 
The resulting integral may be evaluated by Monte Carlo (MC) simulation, as done by Roustant et al.~\cite{rougam20} for Poincar\'e chaos, or by (sparse) quadrature methods \citep{LeMaitre2002, Matthies2005, Constantine2012}. However, note that in general MC converges slowly, while quadrature (even when sparse) is affected by the curse of dimensionality.

The second approach is \textit{regression}, introduced for PCE by Blatman and Sudret~\cite{BlatmanCras2008}. Here, after choosing an \textit{experimental design} (ED) $\cx = \{\ve x^{(1)}, \ldots , \ve x^{(N)}\}$ of input points, \eqref{eq:Poincare_expansion} is discretized as
\begin{equation}
\ve y \approx \ve\Psi \ve c
\end{equation}
where $\ve y = (f(\ve x^{(1)}), \ldots, f(\ve x^{(N)}))^T$ is the vector of model evaluations, $\ve\Psi \in \Rr^{N\times P}$ is the regression matrix with entries
$
\Psi_{kj} = \Phi_{j}(\ve x^{(k)})
$
where $j$ refers to an enumeration of the multivariate basis $(\Phal)_{\alp\in\curlyA}$, and $\ve c$ is the vector of expansion coefficients. 
The discretization of \eqref{eq:Poincare_deriv_expansion} is analogous, with a vector 
\begin{equation}
\ve y_{\partial, i} = \left( \frac{\partial f}{\partial x_i}(\ve x^{(1)}) \enum  \frac{\partial f}{\partial x_i}(\ve x^{(N)}) \right)^T
\label{eq:regression_deriv_rhs}
\end{equation}
containing model partial derivatives and a regression matrix $\ve\Psi_{\partial,i}$ with entries 
\begin{equation}
\Psi^{\partial,i}_{kj} = \frac{1}{\sqrt{\lambda_{i, \alp_j(i)}}}\frac{\partial \Phi_{j}}{\partial x_i} (\ve x^{(k)}),
\label{eq:regression_deriv_regrmatrix}
\end{equation}
where $\alp_j(i)$ denotes the $i$th component of the $j$th basis element characterized by the multi-index $\alp_j$ (see also \cref{prop:PoincDerOrtho}).

The regression problem can be solved by ordinary least squares as
\begin{equation}
\label{eq:OLS}
\hat{\ve c} = \arg\min_{\ve c} \norme{\ve\Psi \ve c - \ve y}{2}^2,
\end{equation}
provided that enough model evaluations are available -- at least $N \geq P$, or better $N \geq kP$ with $k = 2,3$ to avoid overfitting. 
Due to the rapid growth of the total-degree basis with increasing dimension and degree, this requirement on model evaluations is often too restrictive for real-world problems. 

To avoid this problem, sparse regression can be used, which regularizes the problem by encouraging solutions with few nonzero coefficients \citep{Candes2008a, Kougioumtzoglou2020}. An example is $\ell^1$-minimization:
\begin{equation}
\label{eq:ex_l1min}
\hat{\ve c} = \arg\min_{\ve c} \norme{\ve\Psi \ve c - \ve y}{2}^2 + \lambda \norme{\ve c}{1}.
\end{equation}
The $\ell^1$-norm penalizes the coefficient vector so that sparse solutions are preferred.
The sparse regression formulation 
allows for accurate solutions even in the case $N < P$. There exist many sparse regression methods utilizing different formulations of the sparse regression problem, see e.g.\ L\"uthen et al.~\cite{LuethenSIAMJUQ2020} for an overview of available sparse regression solvers in the context of PCE.
In this work, we use the sparse solver Hybrid Least Angle Regression (Hybrid-LARS) \citep{Efron2004, blasud11} in the implementation of UQLab \citep{MarelliUQLab2014,UQdocPCE}.

A result by Cand\`es and Plan~\cite{Candes2011} on sparse recovery emphasizes the importance of \textit{isotropy} of the row distribution of the regression matrix, i.e., the requirement that for a row $\ve a = (\Phi_{\alp_1}(\ve x), \ldots, \Phi_{\alp_P}(\ve x))$ of the regression matrix $\ve\Psi$ it holds that $\Esp{\ve a^T \ve a} = I_P$, where $I_P$ is the identity matrix of size $P$, and the expectation is with respect to the distribution of the experimental design points. If the experimental design points are chosen to follow the input distribution, the distributions of regression matrix rows for Poincar\'e as well as for normalized Poincar\'e derivative expansions are isotropic by construction due to orthonormality of the bases w.r.t.\ the input distribution. 
To improve the space-filling property of the experimental design, we use Latin Hypercube Sampling (LHS) \citep{McKay1979} with maximin distance optimization.

\subsection{Coefficients and Sobol' indices for Poincar\'e derivative expansions}
\label{sec:PDOder_coeffs_sobol}
Let $\hat {\ve c}^{\, \partial, i}$ be the solution to the sparse regression problem corresponding to the $i$-th Poincar\'e derivative expansion (PoinCE-der-$i$) \eqref{eq:Poincare_deriv_expansion} with regression matrix \eqref{eq:regression_deriv_regrmatrix} and data vector \eqref{eq:regression_deriv_rhs}.%
\footnote{Note that in practice, we normalize and rescale the regression matrix as described in \cref{sec:implementation} to improve the estimation of the coefficients.}
By construction, this expansion only provides coefficients corresponding to the basis elements from the set $\curlyA^i := \{\alp \in \curlyA: \alpha_i \geq 1\}$, since the partial derivatives w.r.t.\ $X_i$ of the basis elements $\{\alp \in \curlyA: \alpha_i = 0\}$ are zero and therefore no coefficient value can be computed for those elements.
Theoretically, for $\alp \in \curlyA^i$ the coefficient ${c}_{\alp}^{\, \partial, i}$ from \eqref{eq:Poincare_deriv_expansion} is equal to the PoinCE solution $c_\alp$ from \eqref{eq:Poincare_expansion}, however when estimated from a data set of finite size they will in general not coincide.

The coefficients from the set $\curlyA^i$ are sufficient for computing partial variances for variable $i$ as in \eqref{eq:VarTotWithChaos} and \eqref{eq:VarFirstWithChaos}, 
but not enough for computing the total variance \eqref{eq:VarWithChaos}, which requires all coefficients $c_\alp, \alp \in \curlyA$, and which is needed for normalizing the partial variances to Sobol' indices. 

To compute the total variance from PoinCE-der expansions, we therefore aggregate the coefficients of all $d$ PoinCE-der-$i$ expansions into one vector $\hat{\ve c}^{\, \partial, \text{avg}}$ as follows:
\begin{equation}
\hat{c}_\alp^{\, \partial, \text{avg}} = \frac{1}{\vert\{i:\alpha_i \geq 1\}\vert}\sum_{i: \alpha_i \geq 1} \hat{c}_\alp^{\, \partial, i} \quad \text{ for each } \alp \in \curlyA \setminus \{\ve{0}\},
\end{equation}
i.e., every PoinCE-der-$i$ expansion which computed a coefficient value for the basis element with index $\alp$ contributes equally to the averaged value.
It follows that in theory, the averaged coefficient $\hat {c}_{\alp}^{\, \partial, \text{avg}}$ is equal to the PoinCE solution $c_\alp$ from \eqref{eq:Poincare_expansion}, too.
It can therefore be used to estimate the total variance according to \eqref{eq:VarWithChaos}.

The averaging procedure yields PoinCE-der estimates for all coefficients except for the coefficient $c_{\ve 0}$ corresponding to the constant term $\Phi_{\ve 0}$.
Let $\hat{\ve{c}}_{\alp}^{\, \partial, \text{avg}} = (\hat{c}_\alp^{\, \partial, \text{avg}})_{\alp \in \curlyA \setminus \ve{0}}$ be in the form of a column vector in $\Rr^{(P-1) \times 1}$.
In order to use the averaged PoinCE-der expansion also as a surrogate model, we estimate the remaining coefficient $\hat{c}_{\ve{0}}^{\, \partial, \text{avg}}$ corresponding to the constant term by ordinary least-squares on the residual $\ve{y}_\text{res}$:
	\begin{align*}
	\ve{y}_\text{res} &= \ve{y} - \ve\Psi \begin{pmatrix} 0 \\ \hat{\ve{c}}_{\alp}^{\, \partial, \text{avg}} \end{pmatrix}, \\
	\hat{c}_{\ve{0}}^{\, \partial, \text{avg}} &= \frac{1}{N} \sum_{k=1}^{N} \ve{y}_\text{res}^{(k)}
	\end{align*}

Note that the described construction uses model evaluations and partial derivatives separately. An obvious question is whether one could use these simultaneously to compute an estimate for the coefficients. Although tempting, the simple stacking of regression matrices $\ve\Psi$ and $\ve \Psi_{\partial, i}$ into a big regression matrix (as done by Peng et al.~\citet{penham16} for Hermite PCE) is not satisfactory, since the increasing norm of the basis partial derivatives (see \cref{prop:PoincDerOrtho}) introduces an undesired weighting into the problem. The simultaneous use of evaluation and derivative data is a topic of further research.

\section{Numerical results} \label{sec:numerical}

We investigate the performance of PoinCE (both based on model evaluations and on derivatives) on two numerical examples. The focus of our study is on Sobol' sensitivity analysis, but we also investigate DGSM-based upper bounds to partial variances and validation error (relative mean-squared error). 
Our implementation is based on UQLab \citep{MarelliUQLab2014} and integrates into its PCE module \citep{UQdocPCE}.

We use the following estimation techniques to compute the Sobol' indices of the models:
\begin{itemize}
	\item PoinCE-LARS / PoinCE-der-LARS: Poincar\'e expansion and Poincar\'e derivative expansion computed by LARS as proposed in \cref{sec:sparse_regression}
	\item PoinCE-MC / PoinCE-der-MC: As a baseline, we compare to MC-based computation using the Poincar\'e basis/the Poincar\'e partial derivative basis as in Roustant et al.~\citet{rougam20}
	\item PCE-LARS: As a second baseline, we compare to PCE computed by LARS (with generalized polynomial chaos adapted to the respective input) \citep{blasud11, UQdocPCE}. Sparse PCE is a state-of-the-art method for computing Sobol' indices for real-world models \citep{LeGratiet2017}.
\end{itemize}
We do not compare to any sample-based estimates of Sobol' indices, since it is known that ANOVA-based estimation outperforms sample-based estimation.
For example, Sudret~\citet{sud08} and Crestaux et al.~\citet{cremar08} have shown that polynomial chaos-based estimators of Sobol' indices are much more efficient  than Monte Carlo or quasi-Monte Carlo-based estimators (for smooth models and dimensions up to $20$).
Recently, Becker~\citet{Becker2020} has shown that certain sample-based approaches can be more efficient than metamodel-based ones for screening with total Sobol' indices. However, the screening performance metrics of Becker~\citet{Becker2020} are only based on input ranking. In contrary, our practical purpose is to perform a so-called quantitative screening which aims at providing a correct screening and a good estimation of Sobol' indices.

We do not include a comparison to gradient-enhanced PCE \citep{penham16, Guo2018} because so far these methods are developed only for Gaussian, uniform and Beta input and are not immediately usable for other input distributions. Furthermore, the code of the relatively involved sampling- and preconditioning approach is not readily available. The development and comparison of gradient-enhanced PoinCE to gradient-enhanced PCE is a topic of future research.

Partial variances are normalized to Sobol' indices using the total variance. For PCE-LARS and PoinCE-LARS, the total variance is computed from the expansion coefficients as in \eqref{eq:VarWithChaos}.
For PoinCE-MC and PoinCE-der-MC, we use the sample variance as done by Roustant et al.~\cite{rougam20}.
For PoinCE-der-LARS, the total variance is obtained by the procedure detailed in \cref{sec:PDOder_coeffs_sobol}.

The DGSM-based upper bound to the total partial variances is computed from \eqref{eq:DGSM_poincare_upper_bound} using the coefficients of the PoinCE derivative expansions as described in \cref{sec:PDOder_coeffs_sobol}.
Note that the inequalities in \eqref{eq:DGSM_poincare_upper_bound} are analytical bounds that do not necessarily hold for the estimated quantities.

For uniform and Gaussian input variables, the analytical expression for the Poincar\'e basis functions is used, while for all others, the basis functions are computed numerically using a resolution of $10^3$ points for the uniform grid within the given bounds (see \cref{sec:implementation}).%
\footnote{For the flood model, the change in the resulting Sobol' indices when instead using a grid with $10^2$ or $10^4$ points is in the order of $10^{-4}$ or $10^{-6}$, respectively.}

\subsection{Dyke cost model}
\label{sec:dyke_cost}
Our first application is a simplified analytical model computing the cost associated to a dyke that is to be constructed along a stretch of river to prevent flooding \citep{ioolem15,rougam20}. Its output is the cost in million euros given by
\begin{equation}
Y = \mathbb{1}_{S > 0} + \left[ 0.2 + 0.8 \left( 1 - \exp^{-\frac{1000}{S^4}} \right) \right] \cdot \mathbb{1}_{S \leq 0} + \frac{1}{20} \left( 8 \cdot \mathbb{1}_{H_d \leq 8} + H_d \cdot \mathbb{1}_{H_d > 8} \right)
\label{eq:flood_model}
\end{equation}
where $S$ is the maximal annual overflow and $H_d$ is the dyke height. Here, the first term represents the cost of the consequences of a flooding event, the second describes the maintenance costs, and the third is associated to the construction cost. $S$ is computed from the river characteristics detailed in \cref{tab:flood_params} via the 1D Saint-Venant equations under several simplifying assumptions as follows:
\begin{equation}
S = \left( \frac{Q}{B K_s \sqrt{\frac{Z_m - Z_v}{L}}}\right)^{\! 0.6} + Z_v - H_d - C_b
\label{eq:flood_model_overflow}
\end{equation}
The model \cref{eq:flood_model} is continuous and piecewise $C^1$, and therefore in $H^1$.
It has 8 input variables, of which $Q, K_s, Z_v$ and $H_d$ are important, and $C_b, Z_m, L$ and $B$ are unimportant (see also the last two columns of \cref{tab:flood_params}). 

\begin{table}[htbp]
	\caption{Input variables to the dyke cost model \citep{ioolem15,rougam20} and reference values of first-order and total Sobol' indices. The reference values were obtained to good precision by Monte-Carlo-based Sobol' index estimation using a large sample \citep{rougam20}.}
	\label{tab:flood_params}
	\centering
	\begin{tabular}{llcp{5cm}ll}
		\hline
		Input & Function & Unit & Distribution & $S_i$ & $S_i^\text{tot}$ \\
		\hline
		$Q$ & Maximal annual flowrate & $\si{m^3/s}$ & 
		Gumbel $\mathcal{G}(1013,558)$ \newline truncated to $[500, 3000]$ 
		& $0.358$ & $0.483$ \\
		$K_s$ & Strickler coefficient & $-$ & 
		Gaussian $\cn(30, 8^2)$ \newline truncated to $[15, +\infty]$ 
		& $0.156$ & $0.252$ \\
		$Z_v$ & River downstream level & $\si{m}$ &
		Triangular $\mathcal{T}(49, 51)$
		& $0.167$ & $0.223$ \\
		$Z_m$ & River upstream level & $\si{m}$ &
		Triangular $\mathcal{T}(54,56)$ 
		& $0.003$ & $0.008$ \\
		$H_d$ & Dyke height & $\si{m}$ &
		Uniform $\cu([7,9])$
		& $0.119$ & $0.177$ \\
		$C_b$ & Bank level & $\si{m}$ & 
		Triangular $\mathcal{T}(55,56)$
		& $0.029$ & $0.040$ \\
		$L$ & Length of river stretch & $\si{m}$ &
		Triangular $\mathcal{T}(4990,5010)$
		& $0.000$ & $0.000$ \\
		$B$ & River width & $\si{m}$ &
		Triangular $\mathcal{T}(295,305)$
		& $0.000$ & $0.000$ \\
		
		\hline
	\end{tabular}
\end{table}

The dyke cost model has been used by Roustant et al.~\citet{rougam20} to demonstrate the performance of projection-based PoinCE. We compare the new regression-based methods PoinCE-LARS and PoinCE-der-LARS with the projection-based counterparts PoinCE-MC and PoinCE-der-MC, and additionally with the standard PCE method PCE-LARS.
The projection-based estimates use a basis of total degree 2, while the regression-based estimates use degree adaptivity with a degree of up to 5 (remember that for PoinCE, the degree corresponds to the ordering of the eigenfunctions by the magnitude of the eigenvalues).
The experimental design (ED) is sampled by LHS with maximin distance optimization.
Gradients are computed here by finite differences.
For each size of the experimental design, we perform 50 independent repetitions. We display the resulting estimates in the form of boxplots.
We show results only for three input variables: the most important variable $Q$, the low-importance variable $C_b$, and the unimportant variable $B$. The results for the remaining input variables can be found in \cref{app:additional}.

\subsubsection{Comparison of MC-based and regression-based computation of PoinCE(der)}
First we investigate the two different ways to compute PoinCE: projection-based as in Roustant et al.~\cite{rougam20} versus sparse regression-based as described in \cref{sec:computation}. 
\Cref{fig:flood_Sobol_firstorder_RegVsMC,fig:flood_Sobol_total_RegVsMC} show estimates for first-order and total Sobol' indices. We observe that in all cases the regression-based estimates have a smaller variance than the corresponding projection-based estimates. Also, the median of the regression-based estimates is closer to the true Sobol' index value than the median of the projection-based estimates. 
Note that while the regression-based estimates use a degree-adaptive basis of $p \leq 5$, the projection-based estimates use a fixed degree of only $p=2$. While this choice introduces a certain bias to the projection-based estimates, a larger value for $p$ leads to unfeasibly large variance for those estimates. This is because the coefficients of higher-order terms cannot be estimated precisely with few experimental design points, which makes the overall estimate less precise.

We also observe that regression-based estimates are often clustered around the true Sobol' index already for very small experimental design sizes.
This is in agreement with the observation that sparse-regression-based coefficient estimates have generally a smaller variance compared to MC-based estimates \citep[Chapter 3.4.6]{BlatmanThesis}.
Since regression generally leads to more precise estimates than projection, in the remainder of this paper we focus on regression-based PoinCE estimates.

Furthermore, as already observed by Roustant et al.~\cite{rougam20}, PoinCE-der estimates for Sobol' indices have a smaller variance than PoinCE estimates. In the case of projection-based estimates, this is the case if the derivative has a smaller variance than the original model. In the case of regression, the explanation might be that PoinCE-$i$-der has to compute less coefficients than PoinCE for the same number of experimental design points ($\{\alp \in \curlyA: \alpha_i > 0\}$ vs.\ $\curlyA$), which can result in a more precise estimate of the true coefficient values.

\begin{figure}[htbp]
	\centering
	{\includegraphics[width=.3\textwidth]{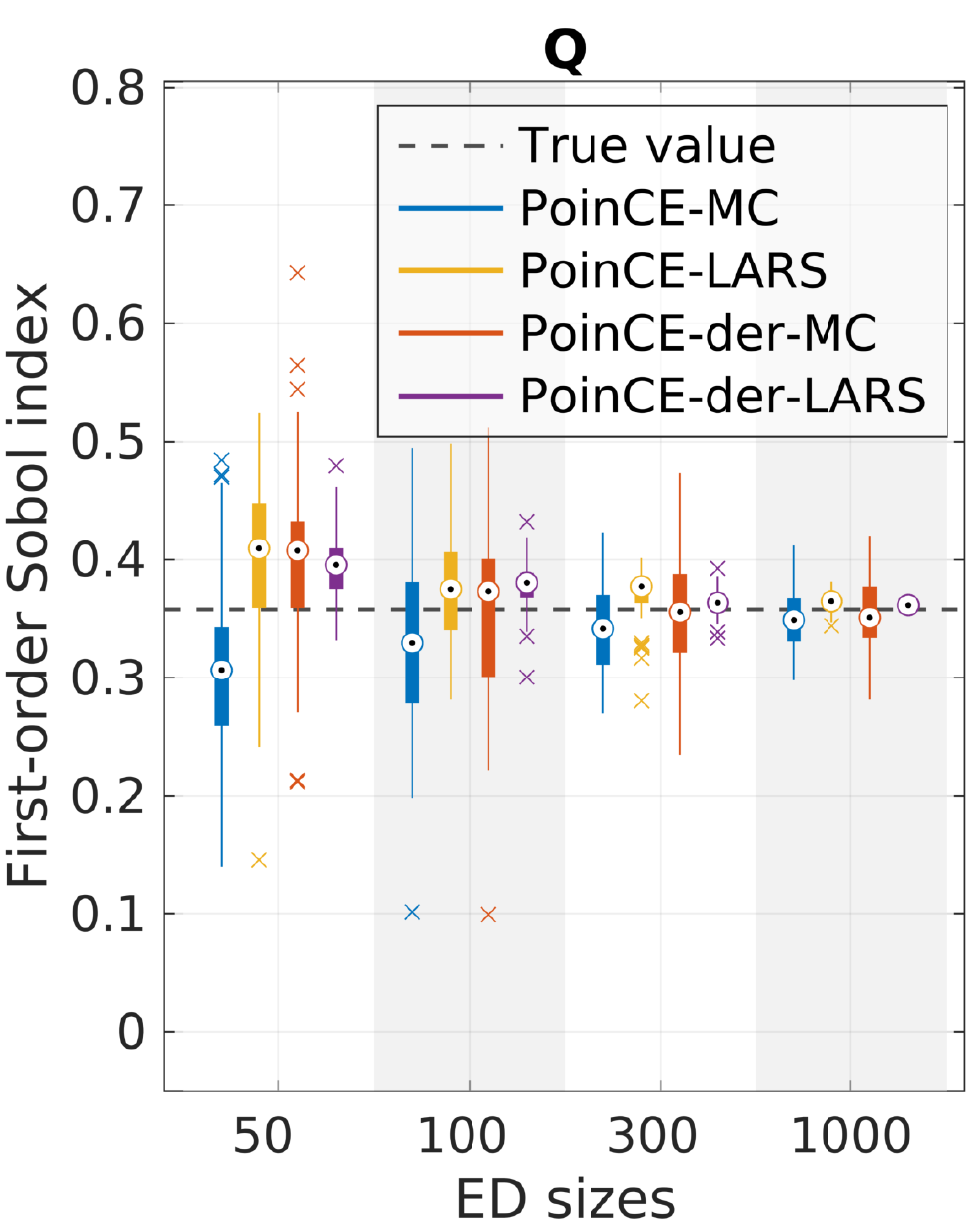}}
	\hfill
	{\includegraphics[width=.3\textwidth]{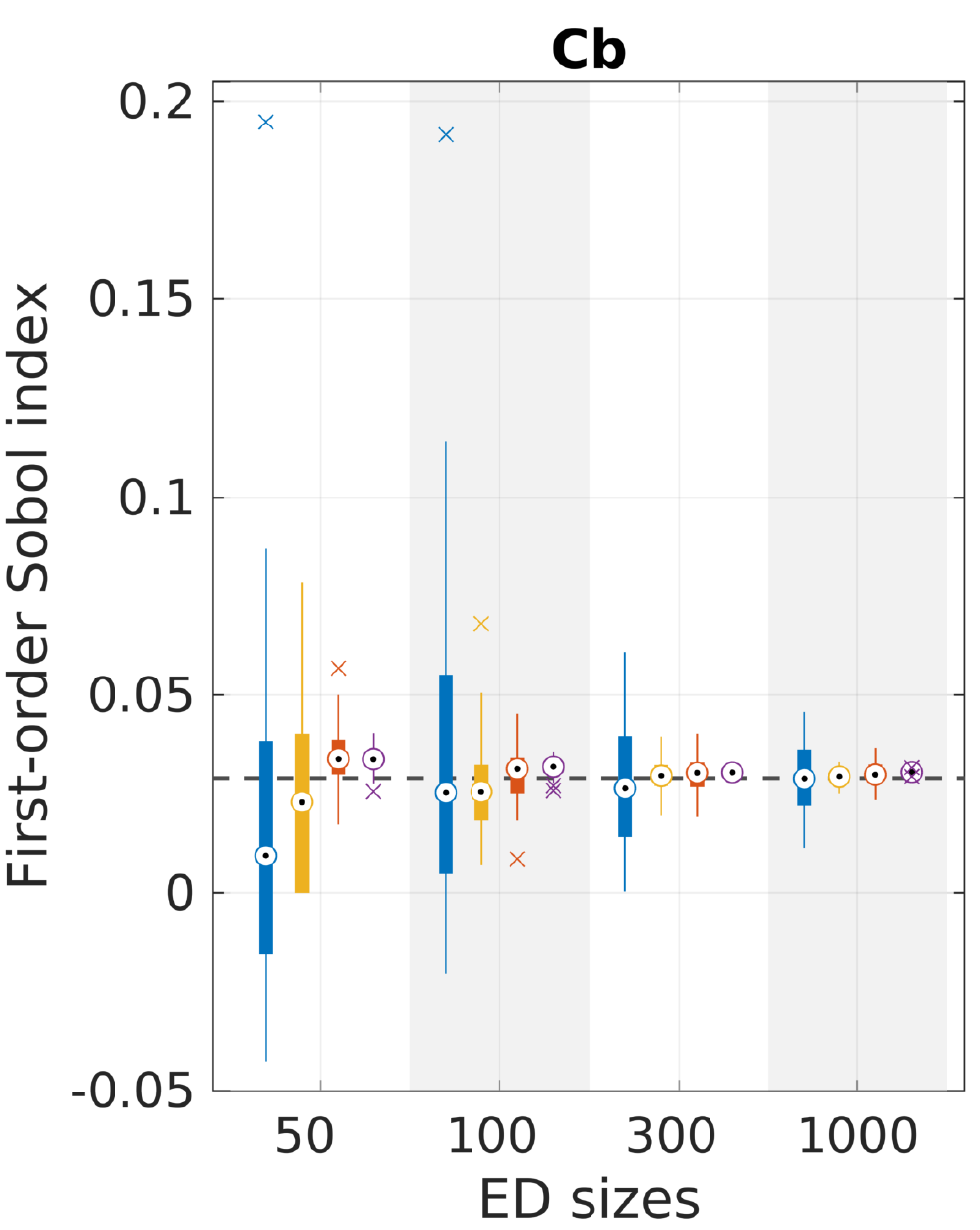}}
	\hfill
	{\includegraphics[width=.3\textwidth]{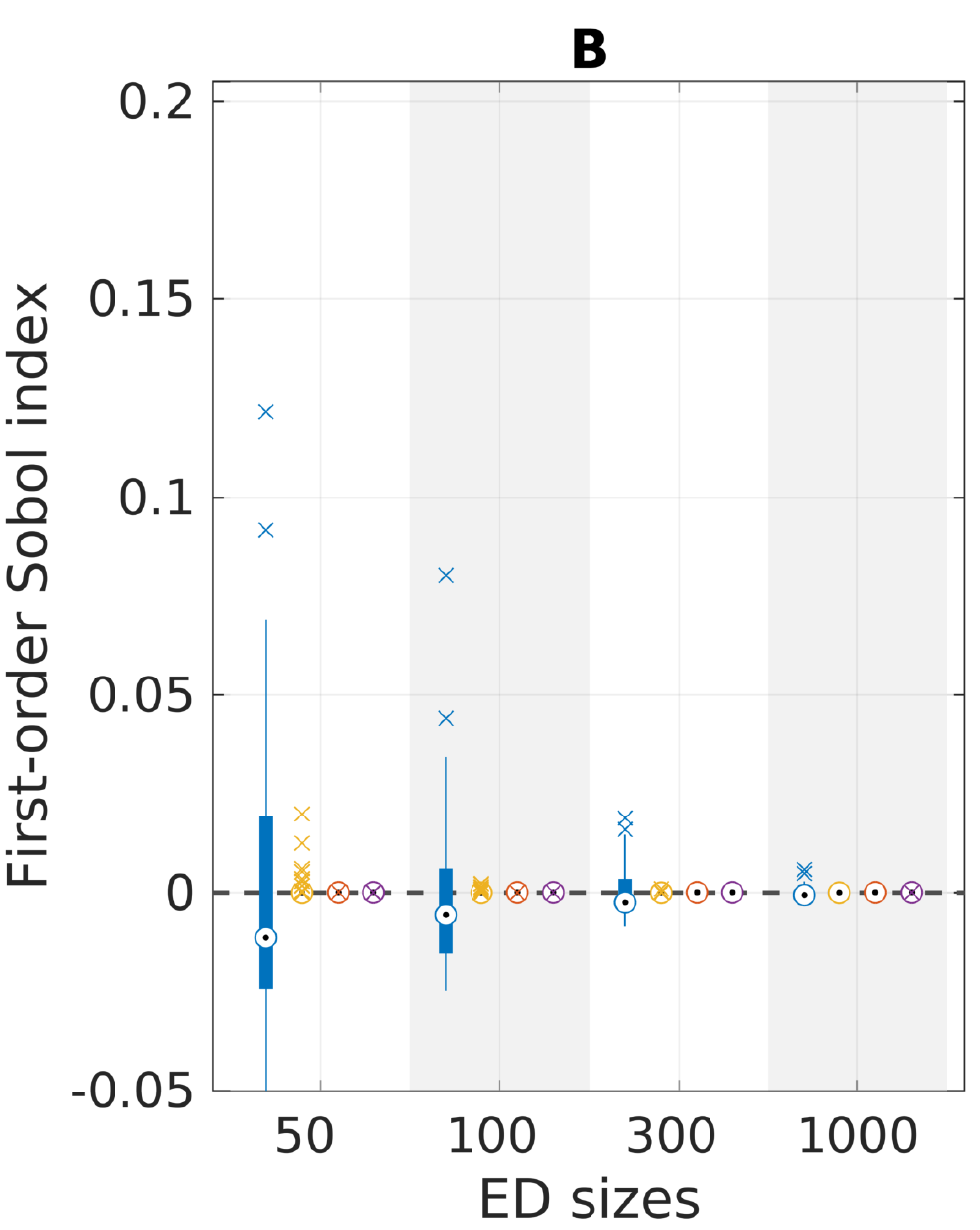}}
	\caption{Comparison of PoinCE estimates of {first-order Sobol' indices} for the dyke cost model. Degree $p = 2$ for the MC-based estimates and $p \leq 5$ (degree-adaptive) for the regression-based estimates. Results for the remaining variables are displayed in \cref{fig:flood_Sobol_firstorder_RegVsMC_appendix} in the appendix.}
	\label{fig:flood_Sobol_firstorder_RegVsMC}
\end{figure}

\begin{figure}[htbp]
	\centering
	{\includegraphics[width=.3\textwidth]{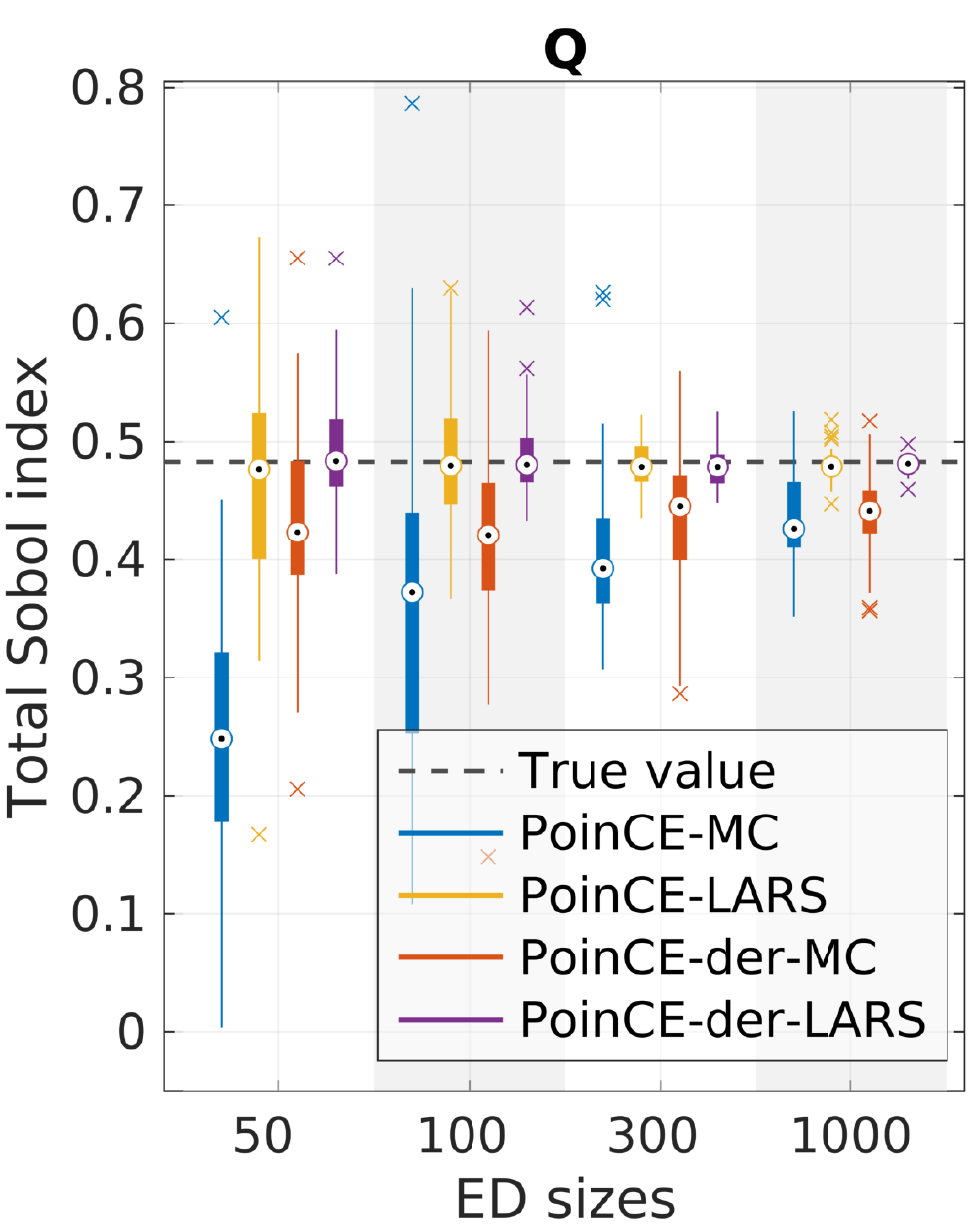}}
	\hfill
	{\includegraphics[width=.3\textwidth]{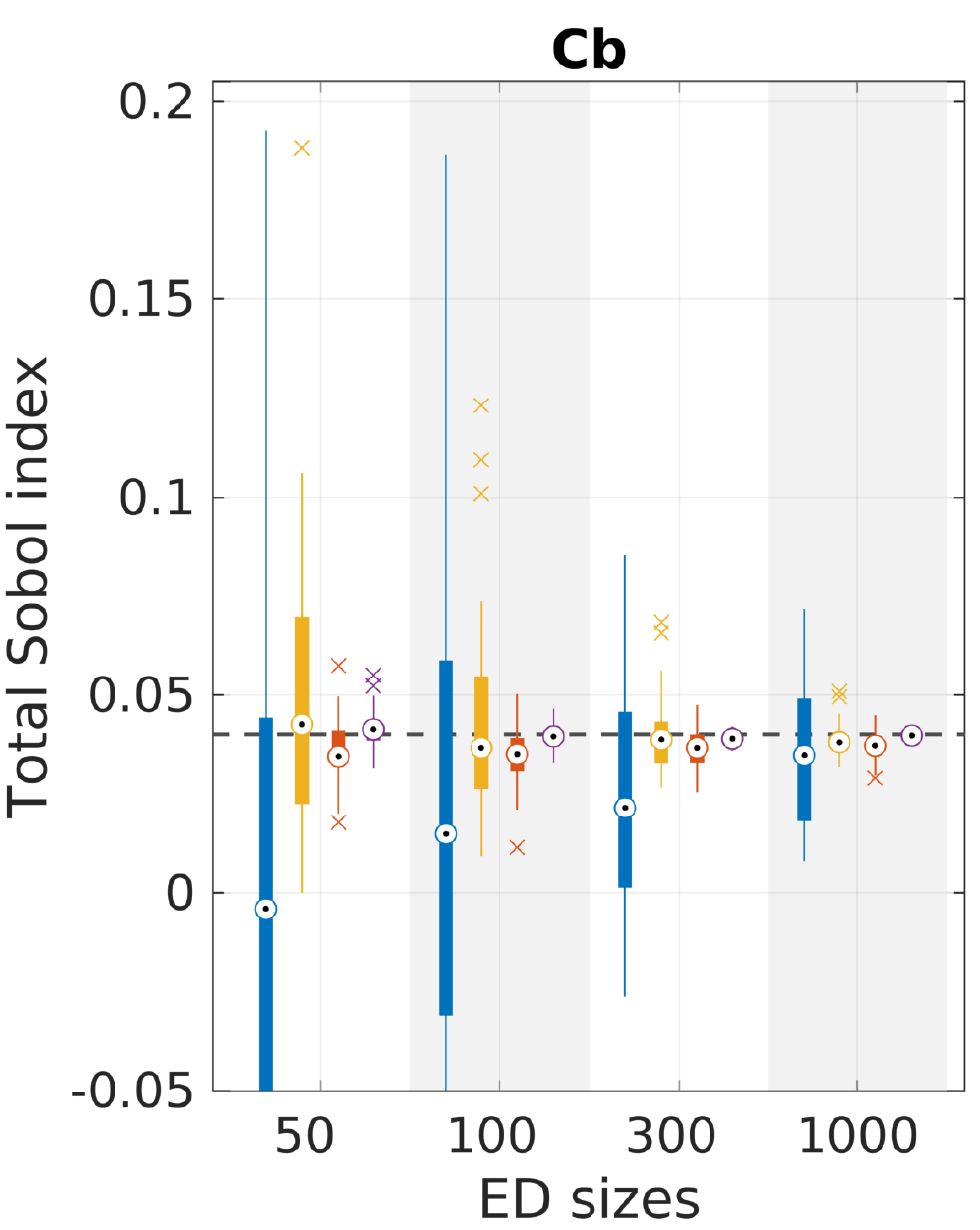}}
	\hfill
	{\includegraphics[width=.3\textwidth]{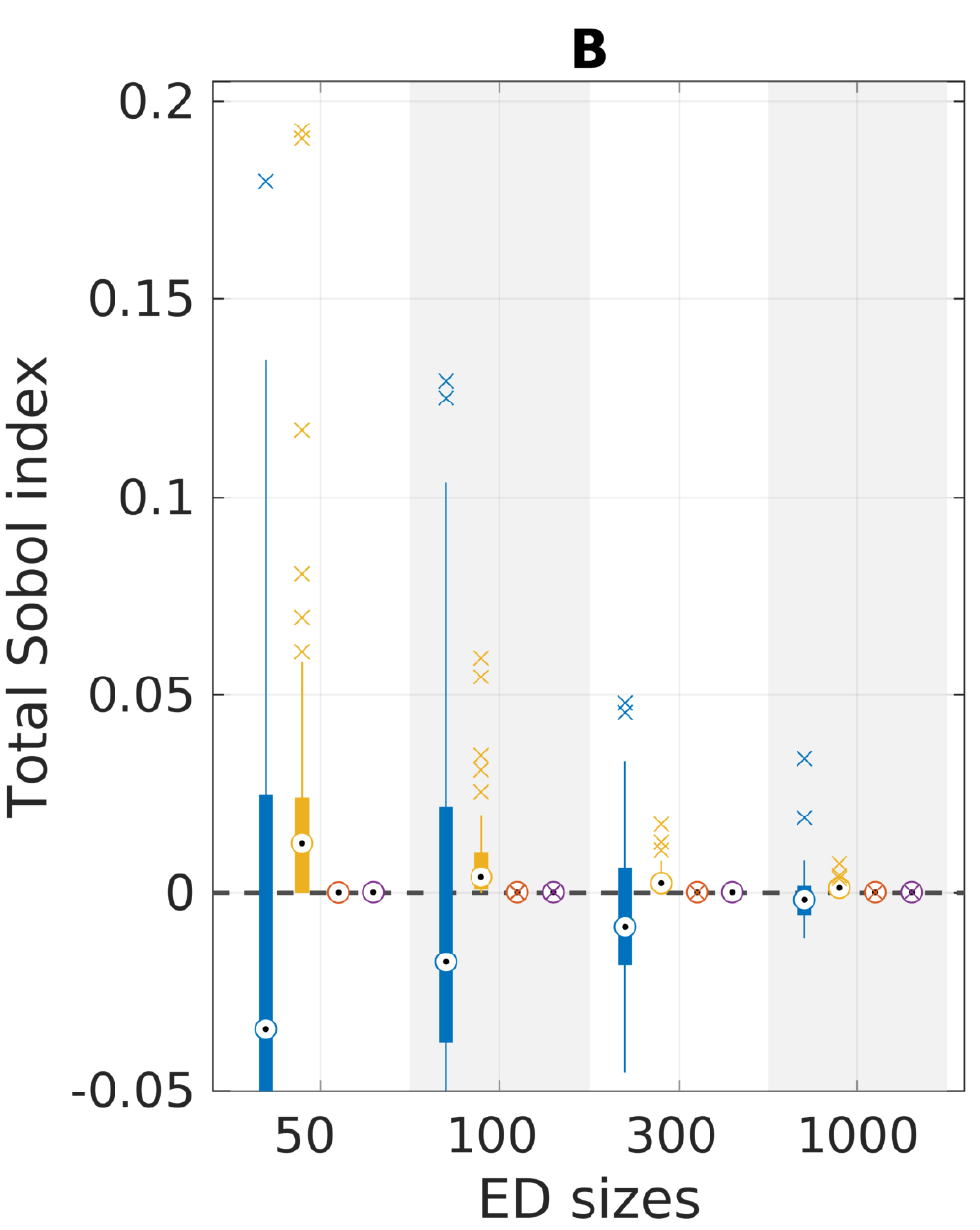}}
	\caption{Comparison of PoinCE estimates of {total Sobol' indices} for the dyke cost model. Degree $p = 2$ for the MC-based estimates and $p \leq 5$ (degree-adaptive) for the regression-based estimates. Results for the remaining variables are displayed in \cref{fig:flood_Sobol_total_RegVsMC_appendix} in the appendix.}
	\label{fig:flood_Sobol_total_RegVsMC}
\end{figure}

\subsubsection{Comparison of regression-based PoinCE-der with PCE and the DGSM-based upper bound}
Next, we investigate the performance of regression-based PoinCE compared to state-of-the-art PCE, and the usefulness of the DGSM-based upper bound to partial variances derived in \eqref{eq:DGSM_poincare_upper_bound}.
The corresponding results, unnormalized%
\footnote{We show unnormalized indices because the DGSM-based upper bound is not normalized.}
estimates for first-order and total Sobol' indices, are displayed in \cref{fig:flood_unnormalized_Sobol_firstorder,fig:flood_unnormalized_Sobol_total}.
Because PoinCE-der achieves more accurate estimates than PoinCE, we compute the DGSM-based upper bound using the PoinCE-der-$i$ coefficients.
For total Sobol' indices, we also include a precise Monte Carlo estimate for the DGSM-based upper bound (using $10^7$ derivative samples) computed from \eqref{eq:DGSM} and the second inequality of \eqref{eq:DGSM_poincare_upper_bound}.

We make the following observations:
the PCE-LARS estimates are generally very similar to the PoinCE-LARS estimates, but the latter often have a slightly larger range. The similarity might be because both rely on model evaluations only. However, the respective basis functions have a very different shape (for inputs that do not follow a Gaussian distribution).
In particular, the PoinCE basis functions by construction obey Neumann boundary conditions, i.e., have zero derivative on the boundary. 

As observed before for normalized indices, PoinCE-der performs better than PoinCE: the median is closer to the true value, and the range is smaller. This effect is especially pronounced for low-importance variables. In 8 dimensions, a PoinCE-der expansion of degree 5 has, due to derivation, $\binom{8+(5-1)}{(5-1)} = 495$ terms, while the total-degree basis of PCE and PoinCE has $\binom{8+5}{5} = 1287$ terms. This means that here, the PoinCE-der expansion has to estimate less than half of the coefficients.
PoinCE-der generally gives a tighter ``lower bound'' than PCE (but note that 
the estimates are not guaranteed to be a lower bound).

By construction \eqref{eq:DGSM_poincare_upper_bound}, the DGSM-based upper bound estimate is larger than or equal to the corresponding total Sobol' index estimate. However, it would be an upper bound to the true Sobol' index value only if the full infinite expansion was used. This is visible in \cref{fig:flood_unnormalized_Sobol_total,fig:flood_unnormalized_Sobol_total_appendix}: for some inputs, the upper bound estimate almost coincides with the Sobol' index estimate, and is smaller than the true Sobol' index.

For some inputs, such as $K_s$ and especially $H_d$ (see \cref{fig:flood_unnormalized_Sobol_total_appendix}), the DGSM-based upper bound is not tight.
From the comment following Equation \eqref{eq:DGSM_poincare_upper_bound} at the end of Section \ref{sec:theory}, this indicates that the Poincar\'e chaos expansion of the cost model contains terms of higher degree than 1 especially for $H_d$. This is explained by the difficulty of approximating a nonsmooth function with a small number of smooth basis functions. 
Indeed, although the dyke cost model admits weak derivatives everywhere as required by the theory, it is only piecewise $C^1$. In particular, there is a jump at $H_d=8$ for the partial derivative with respect to $H_d$, whereas the basis functions for $H_d$ are $C^\infty$ (cosine functions).

Finally, we see that for the variables with larger Sobol' index,
the upper bound and the Sobol indices are underestimated, with a larger negative bias for the upper bound. As remarked above, this indicates from \eqref{eq:DGSM_poincare_upper_bound}
that some higher-order terms are still missing from the considered expansion. Due to the eigenvalue factor involved in the estimate of the upper bound, this has a larger influence on the upper bound than on the Sobol' index estimate.

\begin{figure}[htbp]
	\centering
	{\includegraphics[width=.3\textwidth]
		{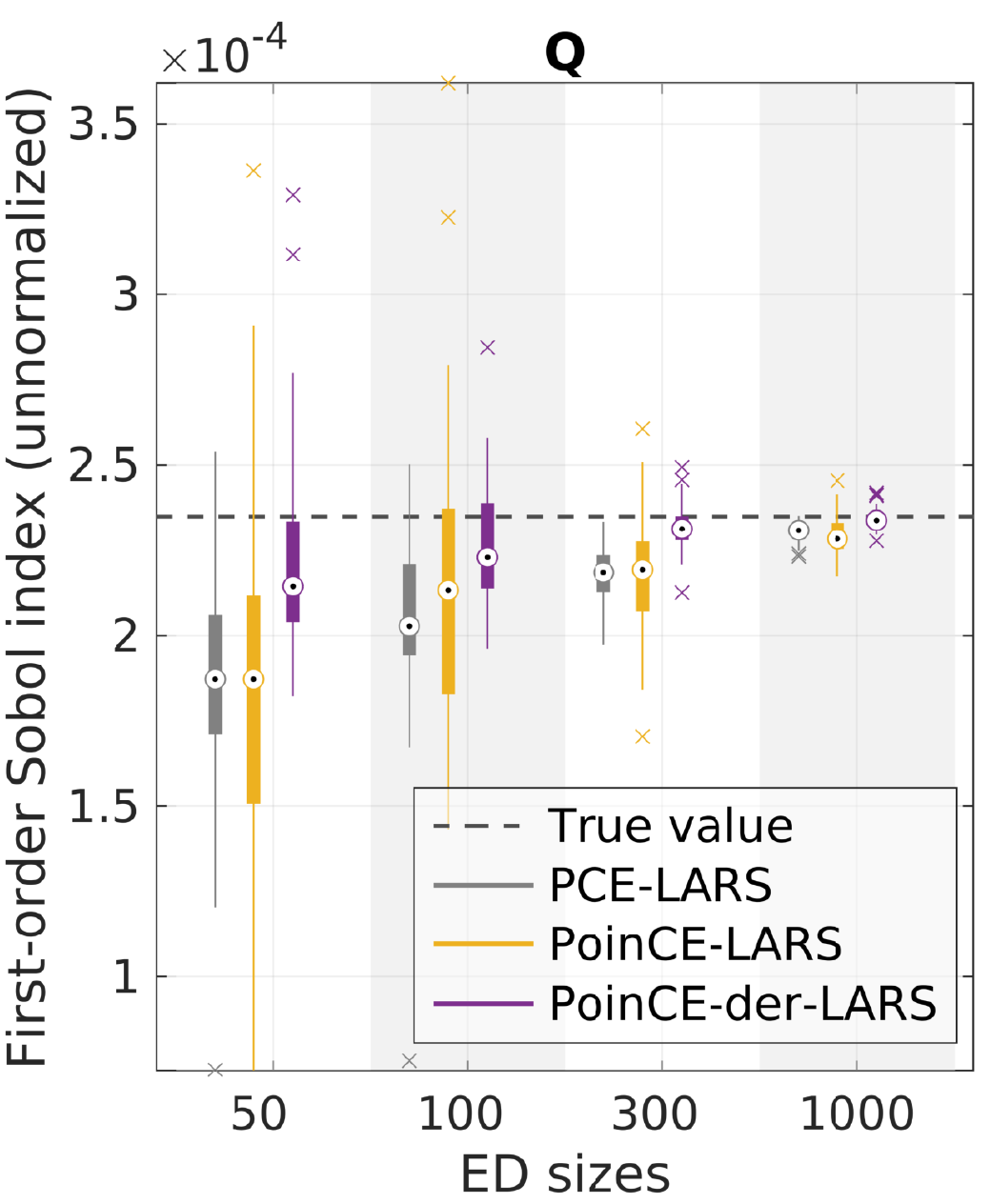}}
	\hfill
	{\includegraphics[width=.3\textwidth]
		{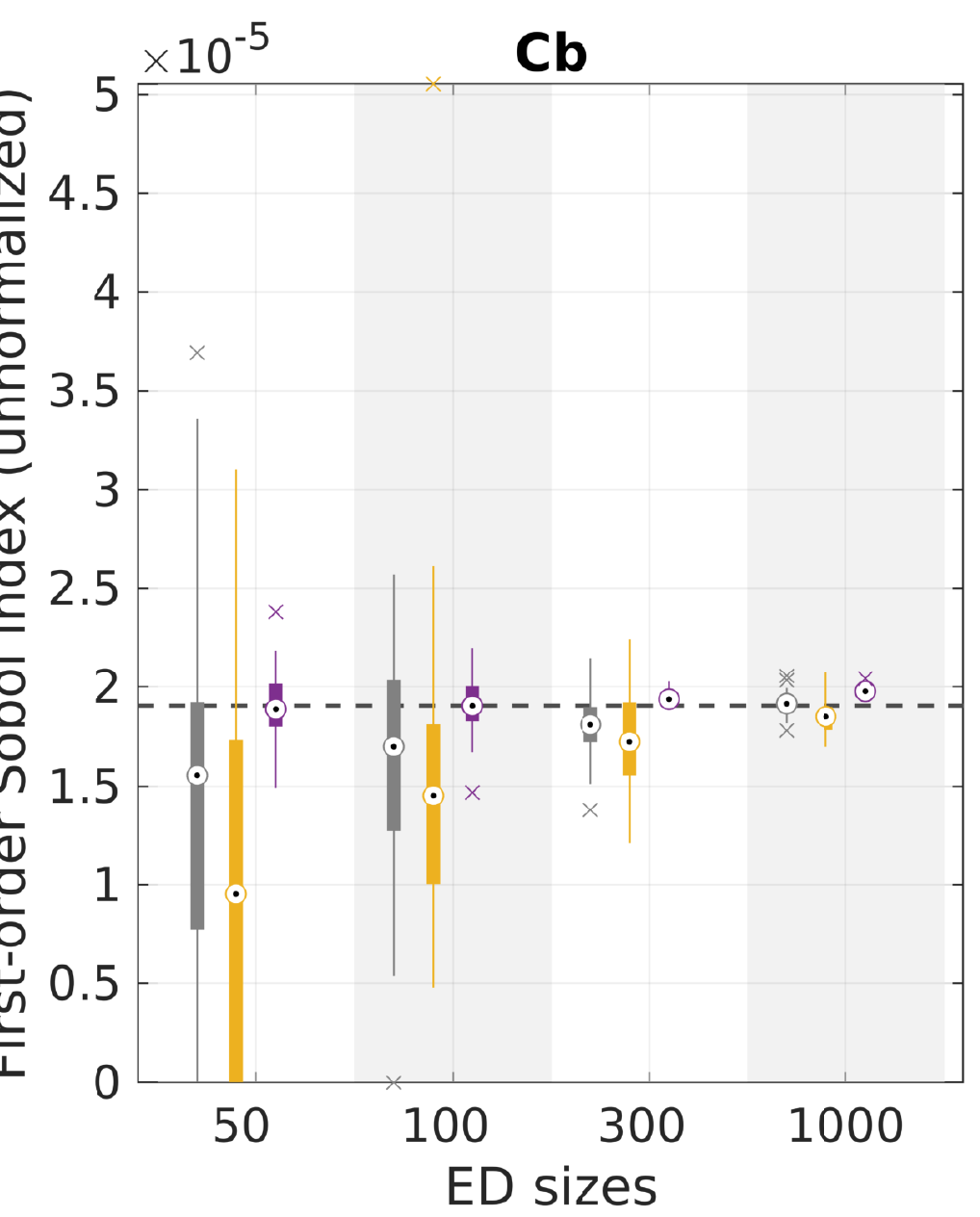}}
	\hfill
	{\includegraphics[width=.3\textwidth]
		{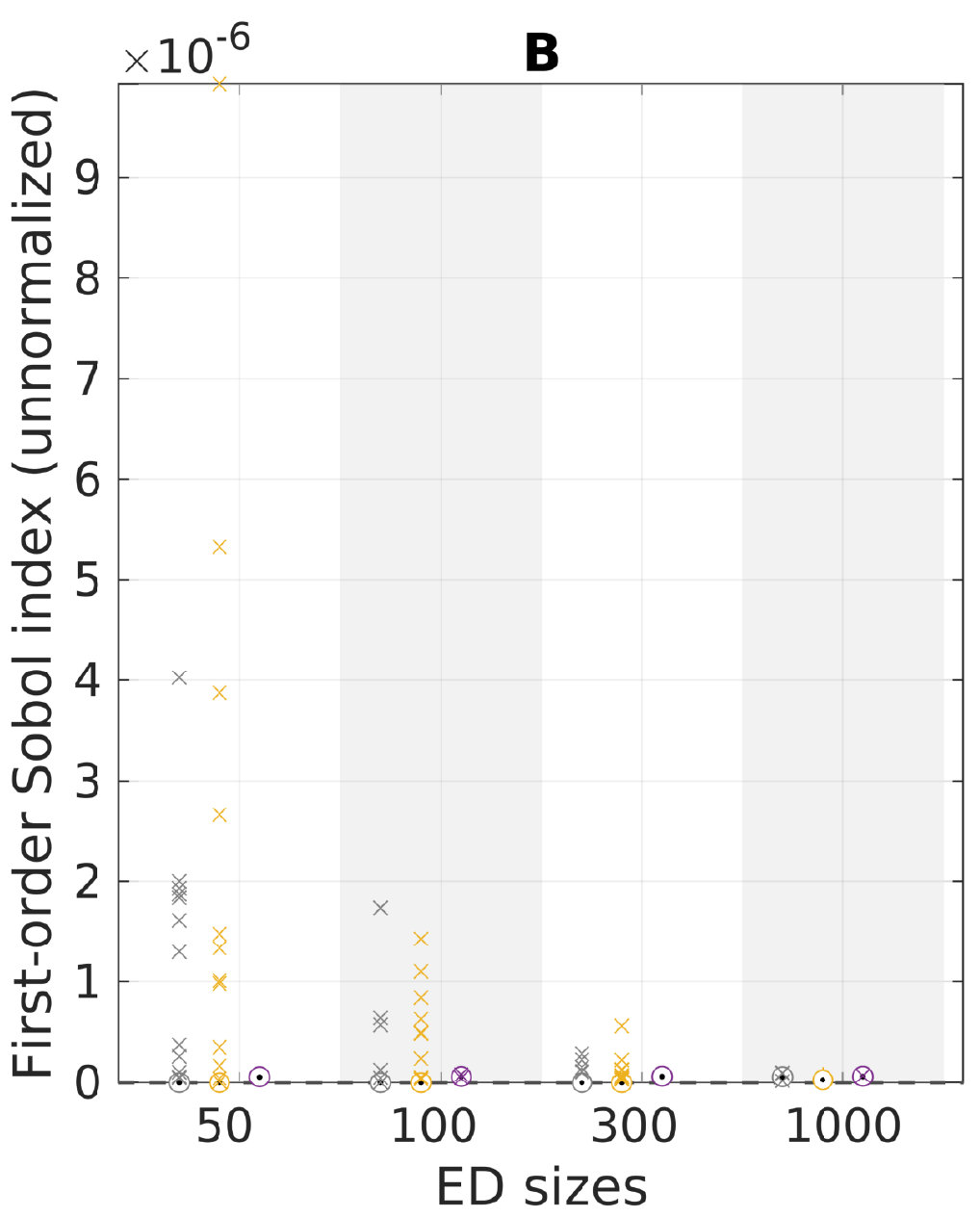}}
	\caption{Estimates of {unnormalized first-order Sobol' indices} for the dyke cost model ($p \leq 5$). Boxplots: in grey the PCE-based estimates. The dashed line (``True value'') denotes a high-precision estimate for the unnormalized first-order Sobol' index. Results for the remaining input variables can be found in \cref{fig:flood_unnormalized_Sobol_firstorder_appendix} in the appendix.}
	\label{fig:flood_unnormalized_Sobol_firstorder}
\end{figure}

\begin{figure}[htbp]
	\centering
	{\includegraphics[width=.3\textwidth]
		{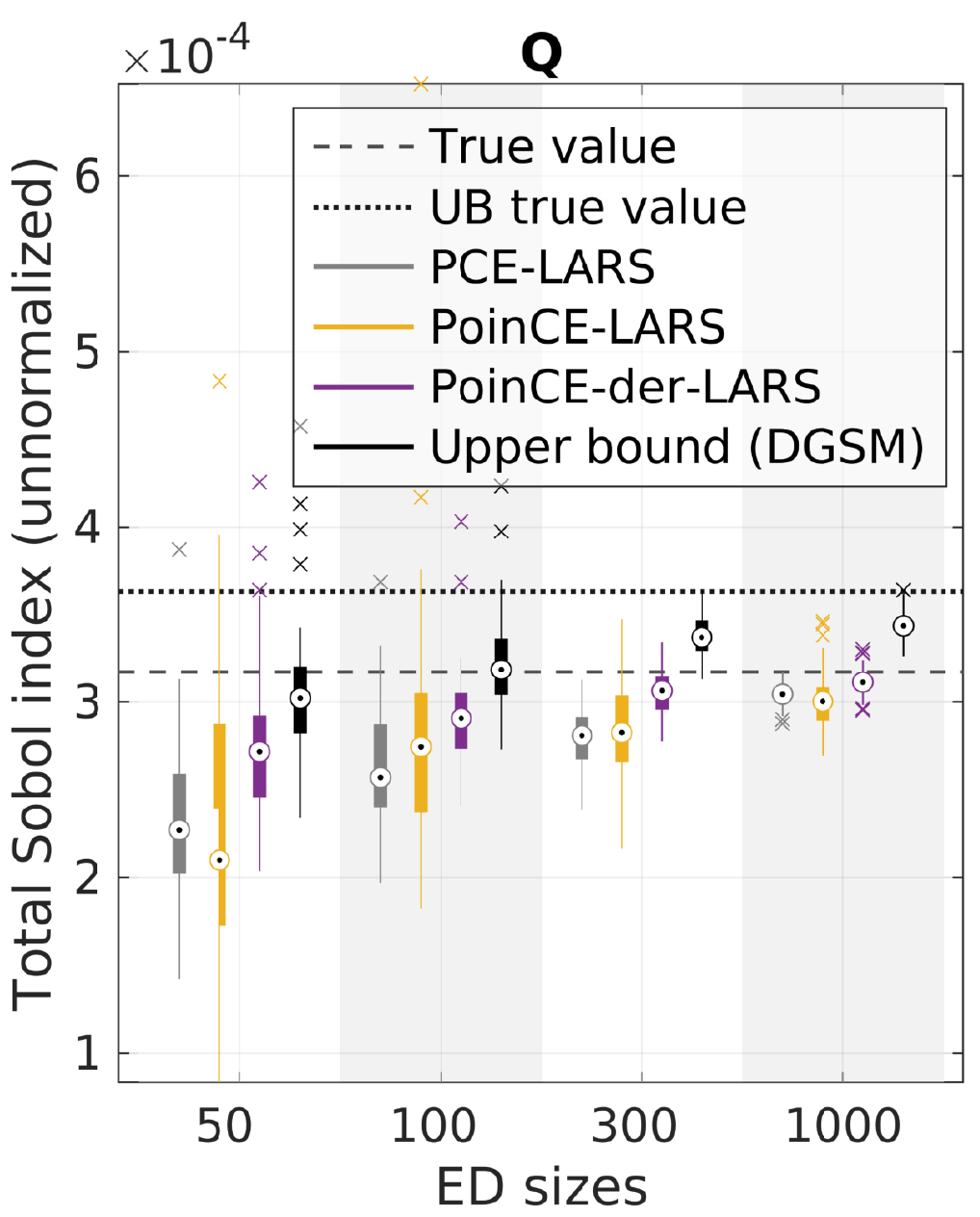}}
	\hfill
	{\includegraphics[width=.3\textwidth]
		{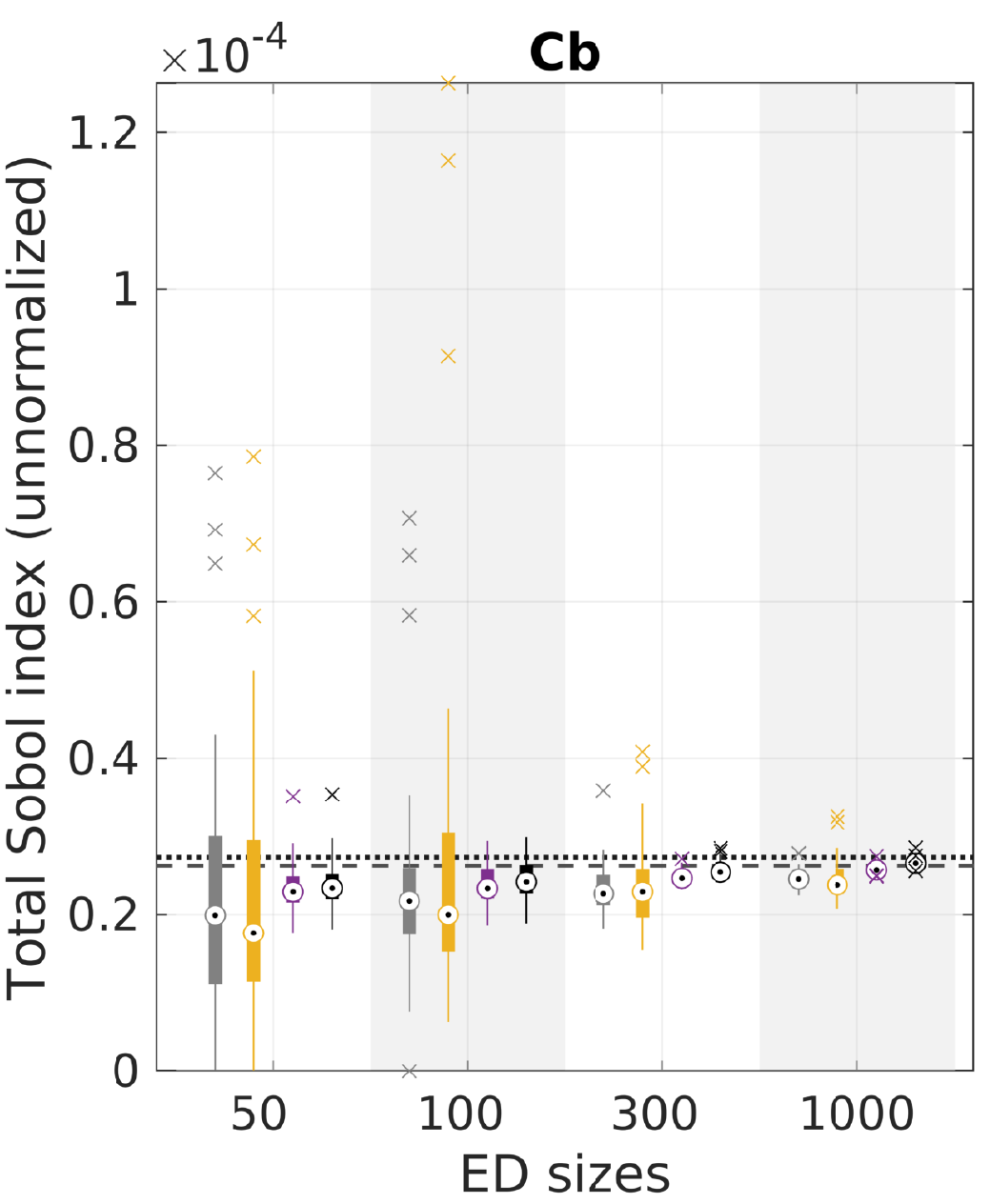}}
	\hfill
	{\includegraphics[width=.3\textwidth]
		{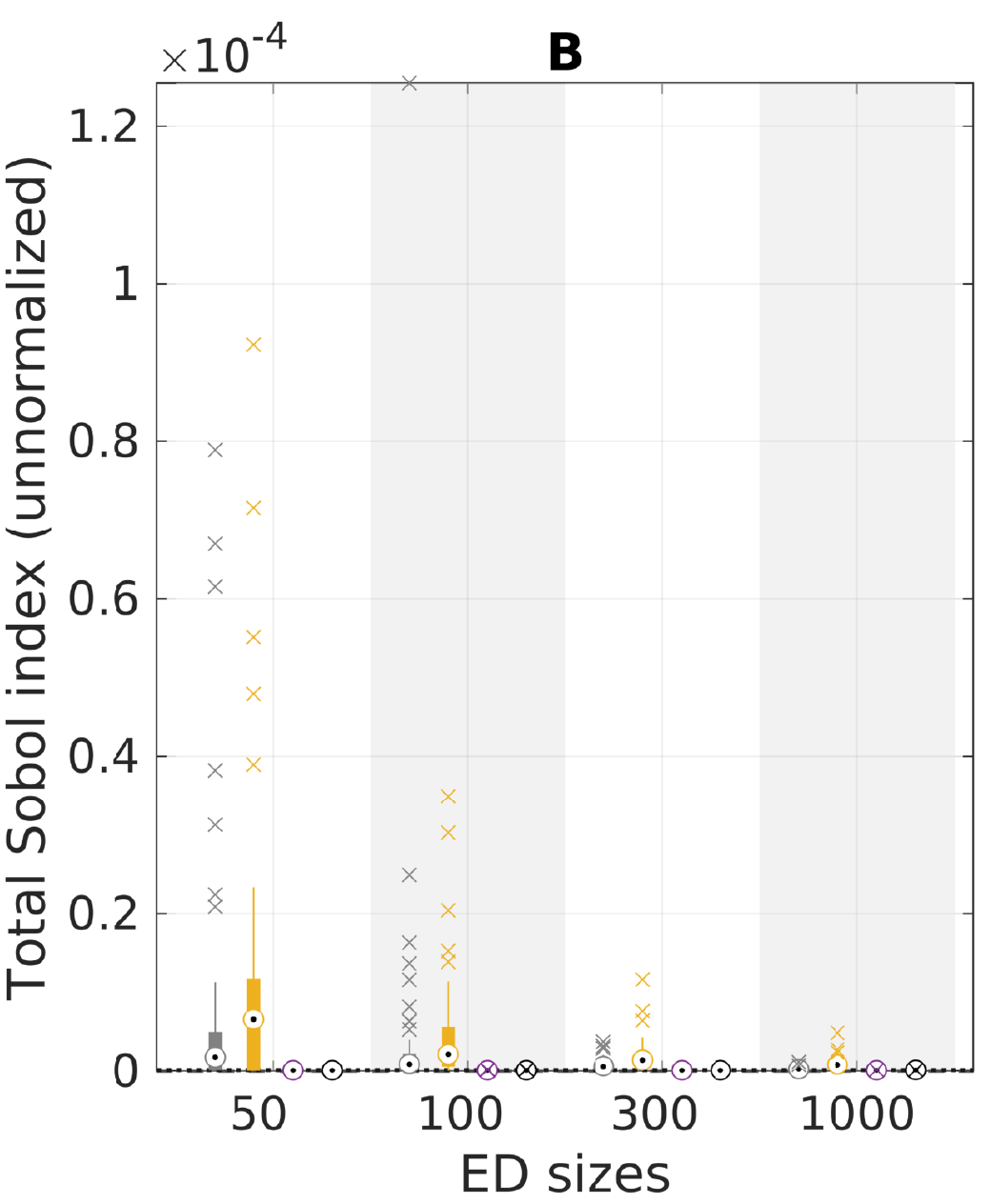}}
	\caption{Estimates of {unnormalized total Sobol' indices} for the dyke cost model ($p \leq 5$). Boxplots: in grey the PCE-based estimates and in black the DGSM-based upper bound from \eqref{eq:DGSM_poincare_upper_bound}. Lines: the dashed line (``True value'') denotes a high-precision estimate for the unnormalized total Sobol' index, while the dotted line (``UB true value'') is a MC-based high-precision estimate for the DGSM-based upper bound.
		Results for the remaining input variables can be found in \cref{fig:flood_unnormalized_Sobol_total_appendix} in the appendix.}
	\label{fig:flood_unnormalized_Sobol_total}
\end{figure}

\subsubsection{Comparison of total variance and relative mean-squared error}
To estimate Sobol' indices precisely, it is crucial to have a good estimate for the total variance. For PCE-LARS, PoinCE-LARS, and PoinCE-der-LARS, this value can directly be computed from the expansion coefficients, while for PoinCE computed by projection, we are using the empirical variance, as detailed in the beginning of \cref{sec:numerical}. In \cref{fig:flood_variance} we display the scatter of the variance estimates (50 replications). The empirical estimate has the largest variation, while PoinCE-der-LARS has the smallest. PCE-LARS and PoinCE-LARS underestimate the total variance more than PoinCE-der-LARS.
This is likely one reason for the good performance of PoinCE-der for the estimation of Sobol' indices: a more accurate total variance leads to more accurate Sobol' indices.

Interestingly, while PoinCE performs well for the estimation of Sobol' indices, this is not true for the generalization error, given by the relative mean-squared error
\begin{equation}
\text{RelMSE} = \frac{\Espe{X}{(f(X) - f^\text{surr}(X))^2}}{\Vare{X}{f(X)}}
\end{equation}
with the surrogate model $f^\text{surr}$. The RelMSE is computed by Monte Carlo integration on a validation set of size $10^6$ sampled from the input distribution $\mu$. In \cref{fig:flood_relmse} we display boxplots of estimates for the generalization error on a validation set of size $10^6$ (mean-squared error normalized by the variance of the validation set). 
PoinCE-der attains a smaller relative MSE than PoinCE. PCE shows faster convergence behavior than both, and attains a smaller relative MSE than PoinCE. PoinCE-der performs better than PCE for the two small experimental design sizes, which shows that the information brought by derivatives might be especially useful when data is scarce.

\begin{figure}[htbp]
	\centering
	\subfloat[Estimates of the variance]
	{\includegraphics[width=.3\textwidth]{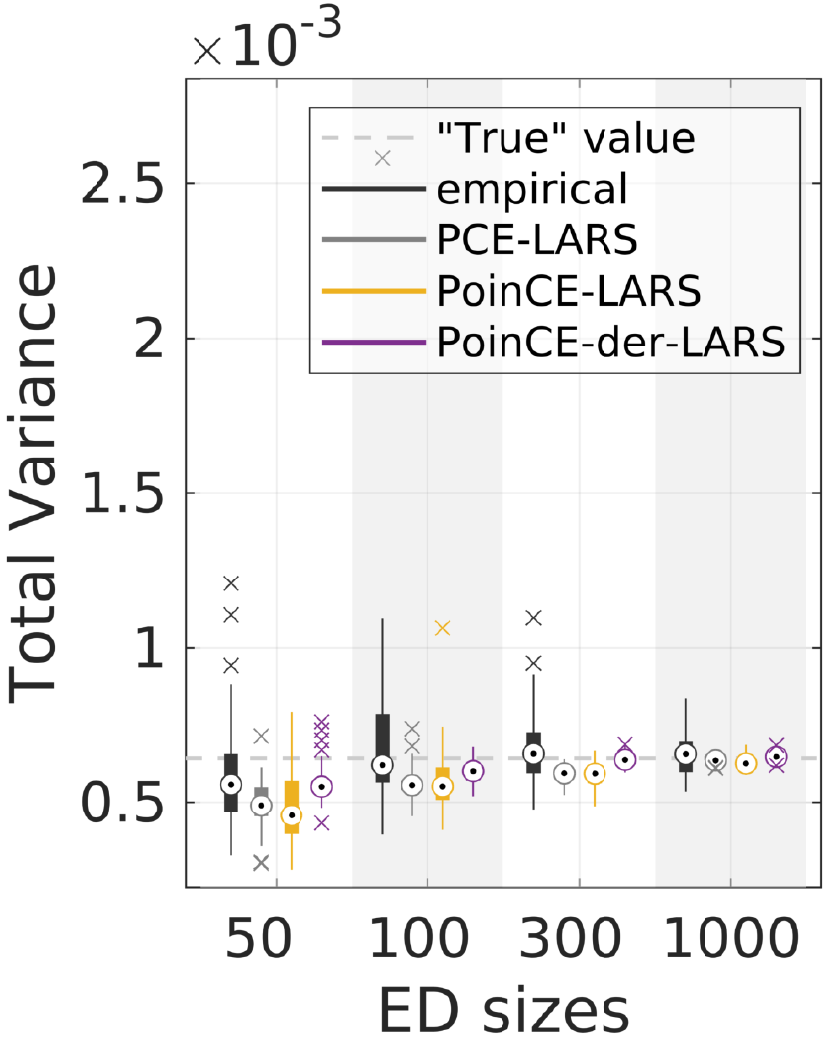}
		\label{fig:flood_variance}}
	\hspace{1cm}
	\subfloat[Relative mean-squared error]
	{\includegraphics[width=.3\textwidth]{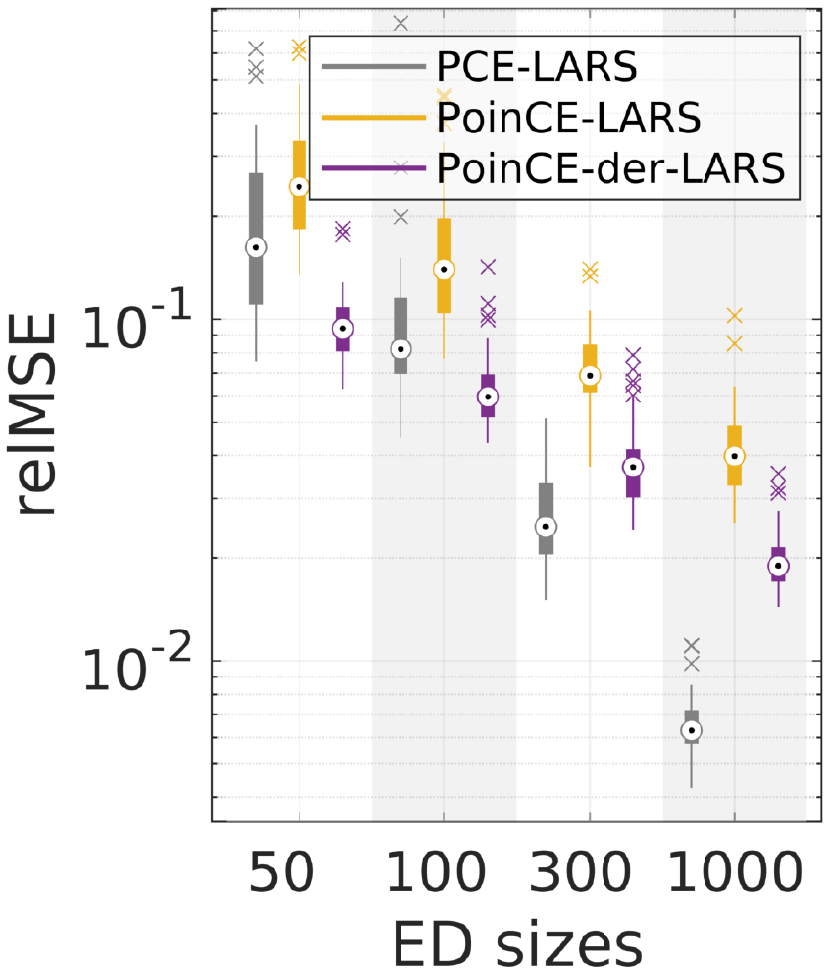}
		\label{fig:flood_relmse}}
	\caption{Dyke cost model: Comparison of PCE and PoinCE(-der) with respect to the following metrics: estimation of total variance, and relative mean-squared error.}
\end{figure}


\subsection{Mascaret data set}

Our second application focuses on a phenomenological and industrial simulation model, called Mascaret \citep{Goutal_et_al_2012}, based on a 1D solver of the Saint Venant equations and aiming at computing water height for river flood events.
The studied case, taken from Petit et al.~\citet{Petit2016} and also studied in Roustant et al.~\citet{roubar17}, is the French Vienne river in permanent regime whose uncertain input data concern flowrate, several physical parameters and geometrical data (transverse river profiles).
$37$ independent inputs have then been considered as random variables \citep{Petit2016}:
\begin{itemize}
	\item $12$ Strickler coefficients of the main channel $K_{s,c}^i$, uniform in $[20, 40]$;
	\item $12$ Strickler coefficients of the flood plain $K_{s,p}^i$, uniform in $[10, 30]$;
	\item $12$ slope perturbations $dZ^i$, standard Gaussian with bounds $[-3,3]$;
	\item $1$ discharge value $Q$, Gaussian with zero mean and standard deviation $50$, bounds $[-150,50]$.
\end{itemize}
The derivatives of the model output with respect to these $37$ inputs have been efficiently (with a cost independent of the number of inputs) computed by using the adjoint model of Mascaret \citep{demgoe16}.
This adjoint model has been obtained by automatic differentiation \citep{griwal08} using the automatic differentiation software Tapenade \citep{haspal13}.
A large-size Monte Carlo sample ($n=20\,000$) is available from the study of Petit et al.~\citet{Petit2016}.
This data set contains all the values of the $37$ inputs, the water height as output and the $37$ partial derivatives of the output (one derivative with respect to each input).
Note that this sample, which has a very large size, has been obtained during a research work for a demonstrative purpose.
In industrial practice, the aim is to use the minimal possible sample size: it is expected to use methods able to deal with sample sizes of the order of one hundred. 

Previous studies on this data set \citep{Petit2016,roubar17} have identified $32$ of the $37$ inputs as noninfluential.
In our study, we display results for the $5$ remaining inputs ($K_{s,c}^{11}$, $K_{s,c}^{12}$, $dZ^{11}$, $dZ^{12}$, and $Q$) and for one of the noninfluential inputs ($K_{s,c}^1$).
We choose a basis with hyperbolic truncation using $q = 0.5$, and degree adaptivity $p = 1,2 \enum 8$. We analyze several experimental design sizes ranging from $30$ to $300$. For each experimental design size, we run 30 replications, sampling the design randomly without replacement from the given full data set. 
``True'' values for Sobol' indices and total variance are computed from a PCE using all $20\,000$ points.

\subsubsection{Comparison of regression-based PoinCE(der) with PCE and the DGSM-based upper bound}
Estimates of first-order and total Sobol' indices are displayed in \cref{fig:mascaret_firstorder,fig:mascaret_total}. We display results for regression-based PCE, PoinCE, and PoinCE-der. In addition, we display the upper bound computed as in \eqref{eq:DGSM_poincare_upper_bound}, computed based on PoinCE-der coefficients and normalized by the PoinCE-der total variance.
We observe that for the non-influential variable $K_{s,c}^1$ (and indeed all other 31 non-influential variables), derivative-based PoinCE correctly identify a total and first-order Sobol' index of $0$. 
Overall, PoinCE and PCE show very similar results, with PoinCE having slightly larger variance in a few cases. 
For some variables such as $dZ^{11}$ and $Q$, the DGSM-based estimate of the upper bound almost coincides with the PoinCE-der estimate. 
Overall, we observe that PoinCE-der estimates have smaller variance than PCE and PoinCE for the important variables $K_{s,c}^{11}, dZ^{11}$, and $Q$, even already for 30 experimental design points. For low-importance variables such as $K_{s,c}^{12}$ and $dZ^{12}$, PoinCE-der correctly identifies a value away from zero already for the smallest experimental design, while half of the PCE and PoinCE estimates are zero. 
For small experimental design sizes, the PoinCE-der estimates also have a smaller bias than the PCE and PoinCE estimates. Sometimes the PoinCE-der estimates seem to systematically over- or underestimate the true Sobol' index by a small amount. However, note that the ``true'' value was computed by a PCE (based on all $20\,000$ points), and might therefore itself be slightly inaccurate.

\begin{figure}[htbp]
	\centering
	{\includegraphics[width=.3\textwidth]
		{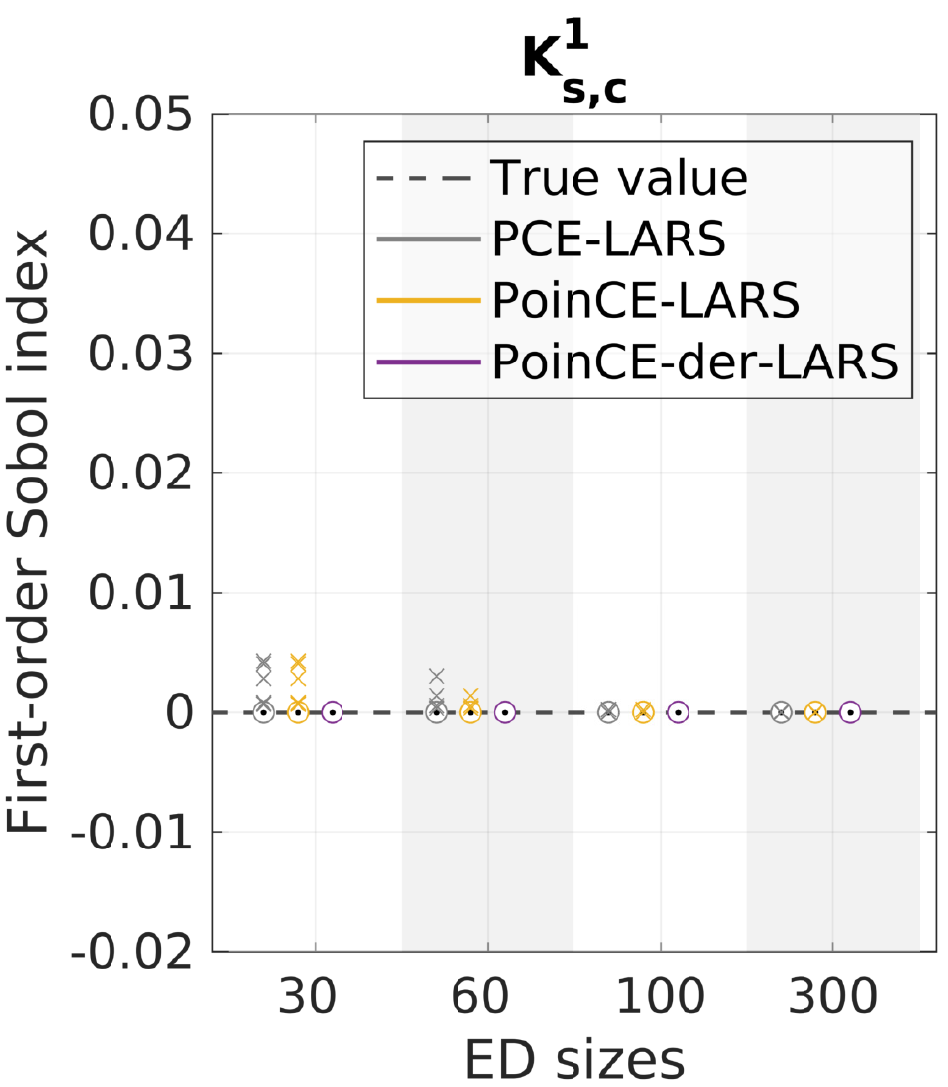}}
	\hfill
	{\includegraphics[width=.3\textwidth]
		{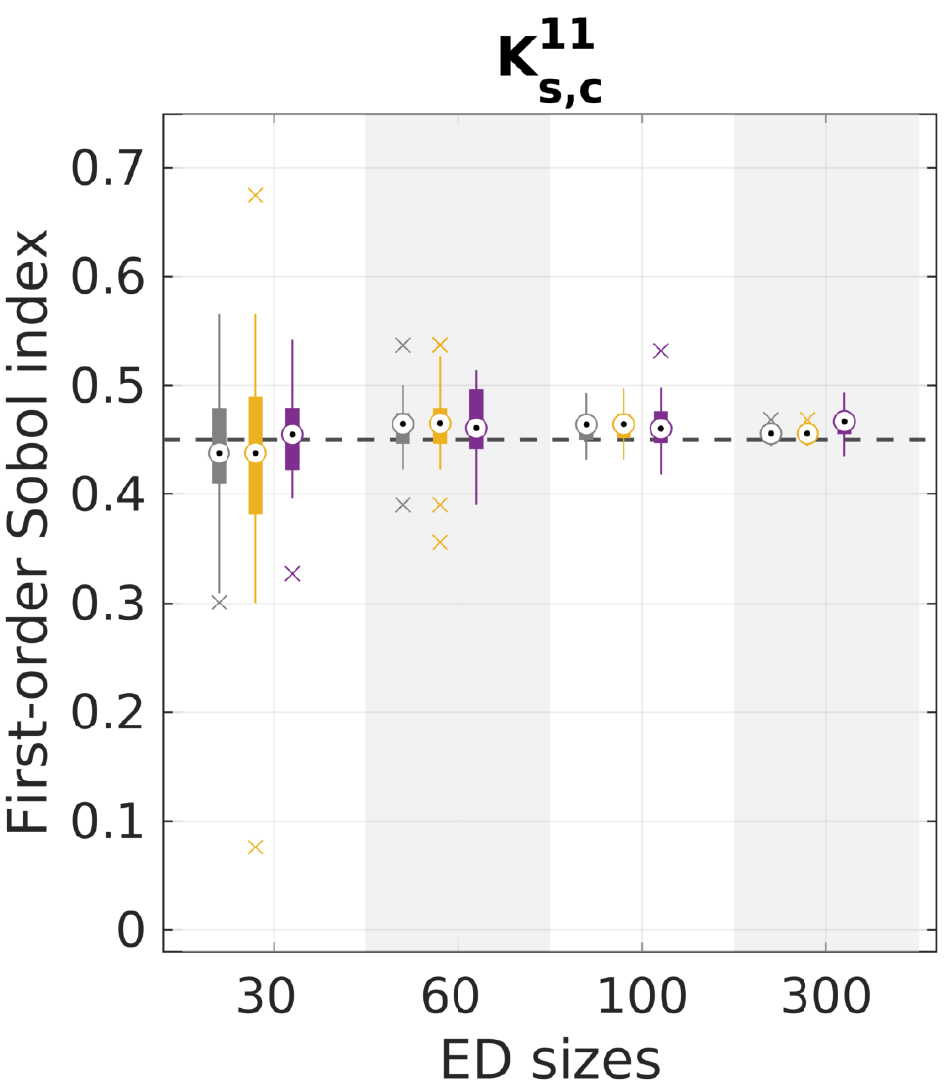}}
	\hfill
	{\includegraphics[width=.3\textwidth]
		{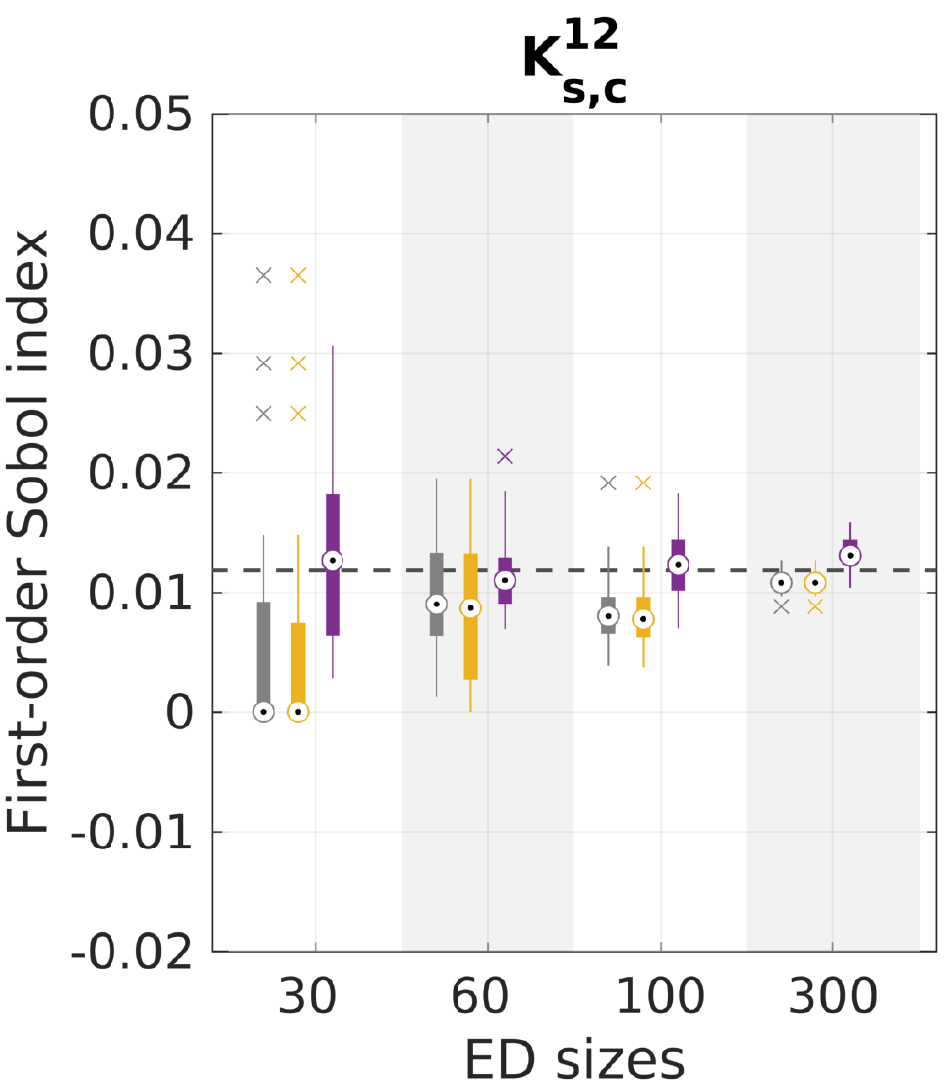}}
	\\
	{\includegraphics[width=.3\textwidth]
		{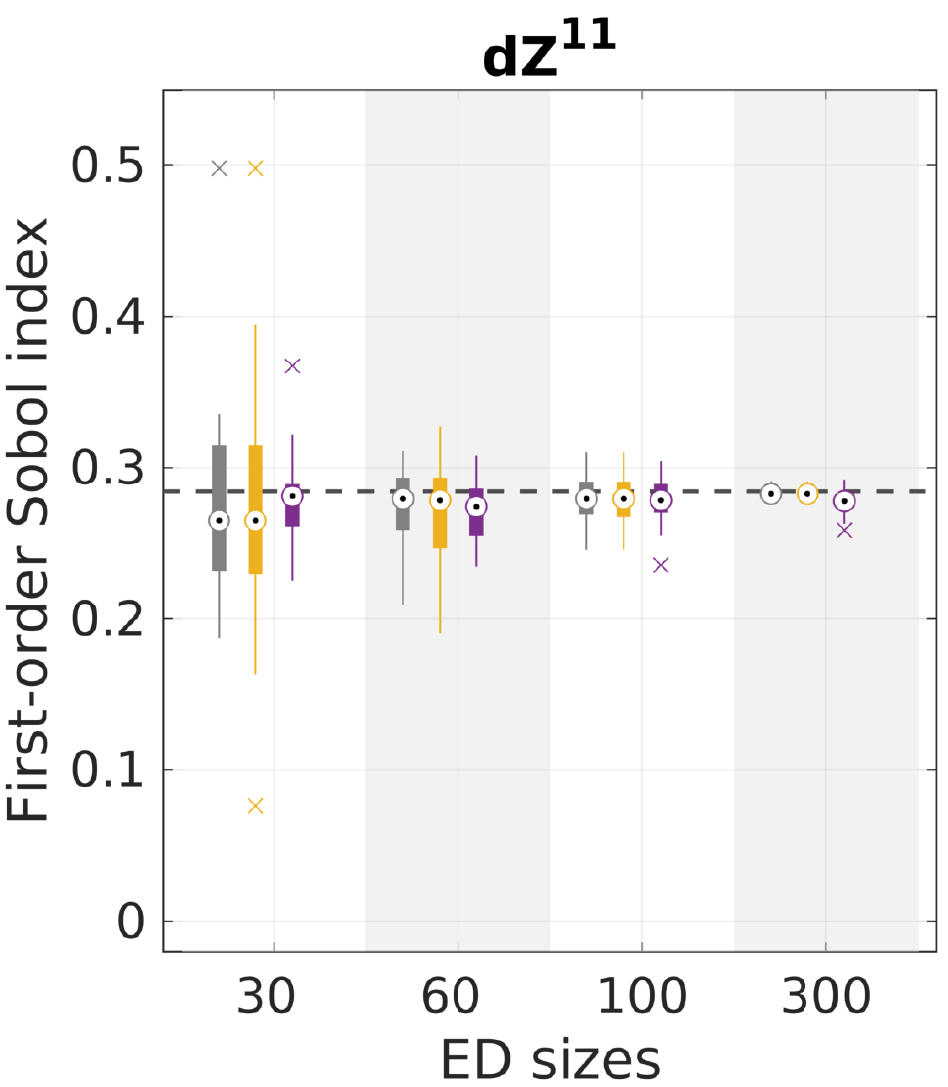}}
	\hfill
	{\includegraphics[width=.3\textwidth]
		{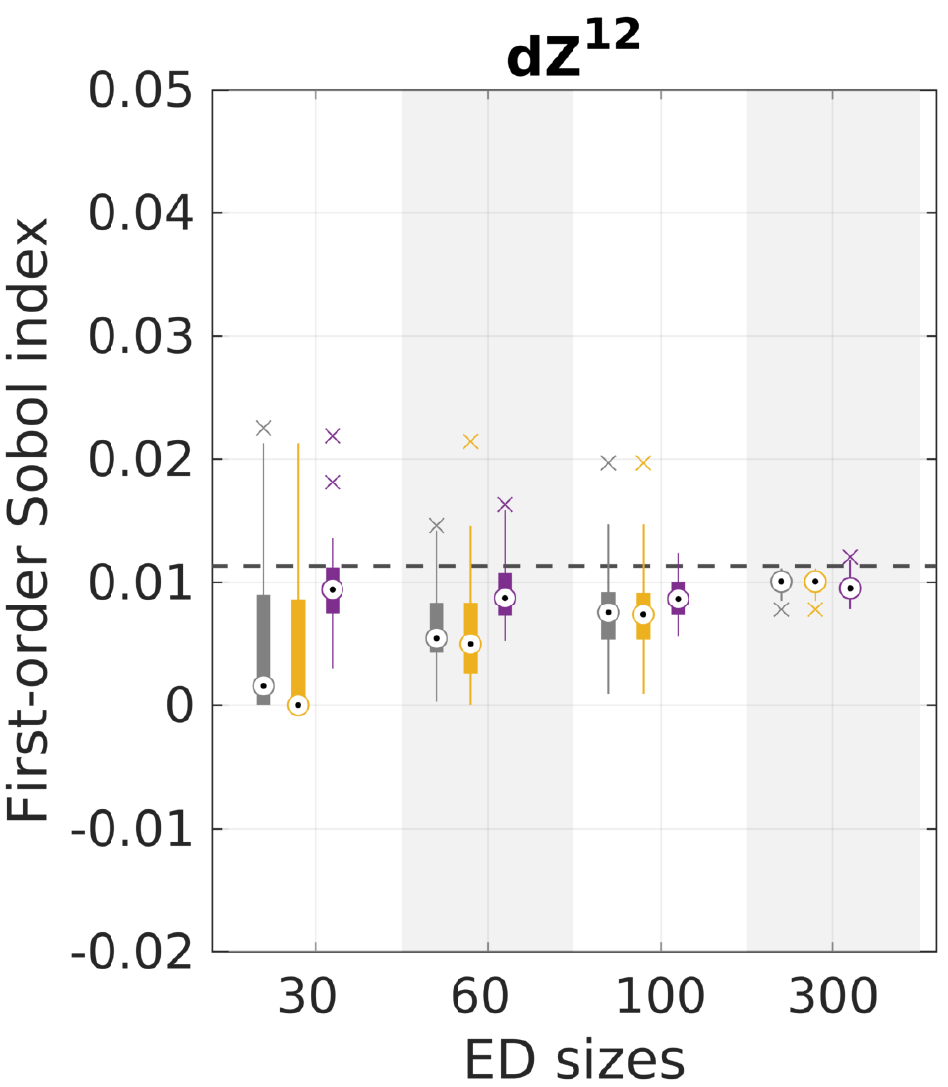}}
	\hfill
	{\includegraphics[width=.3\textwidth]
		{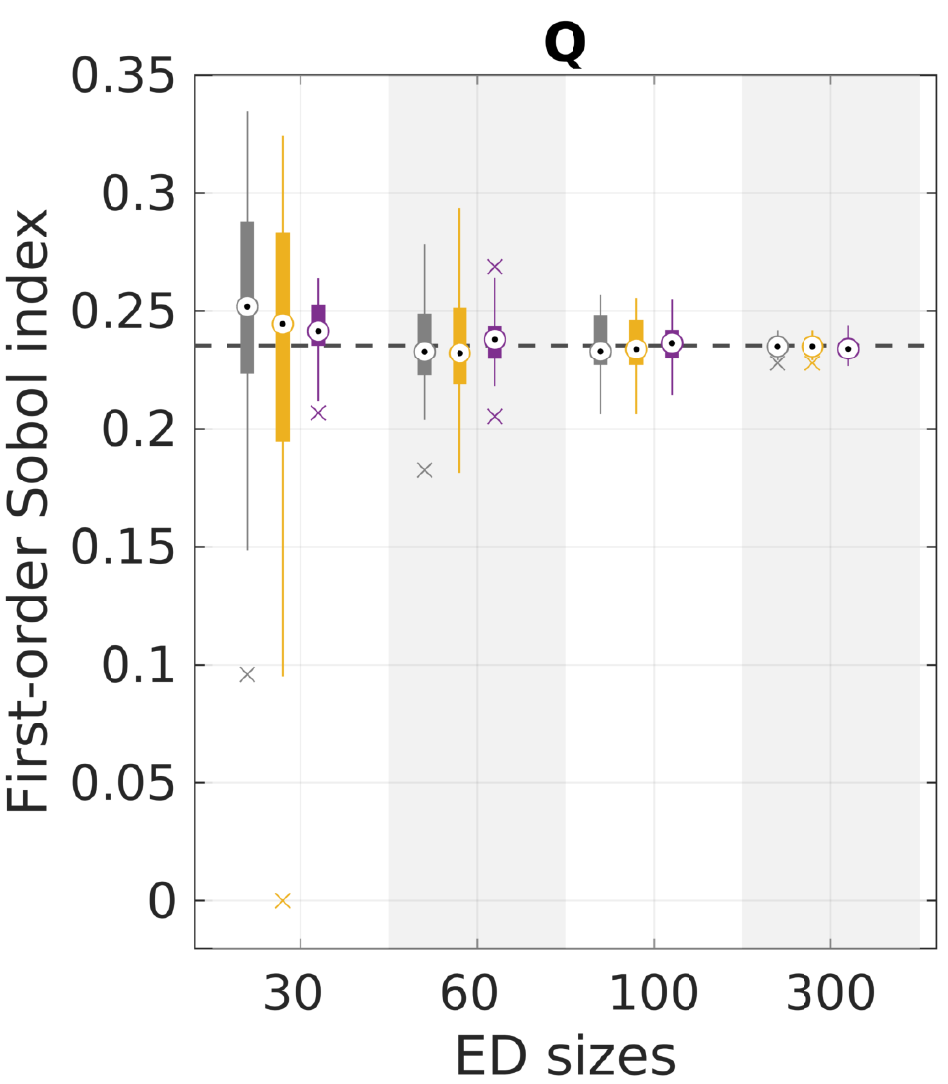}}

	\caption{First-order Sobol' indices for the Mascaret data set (30 replications). ``True'' values computed from a PCE using all $20\,000$ points.}
	\label{fig:mascaret_firstorder}
\end{figure}

\begin{figure}[htbp]
	\centering
	{\includegraphics[width=.3\textwidth]
		{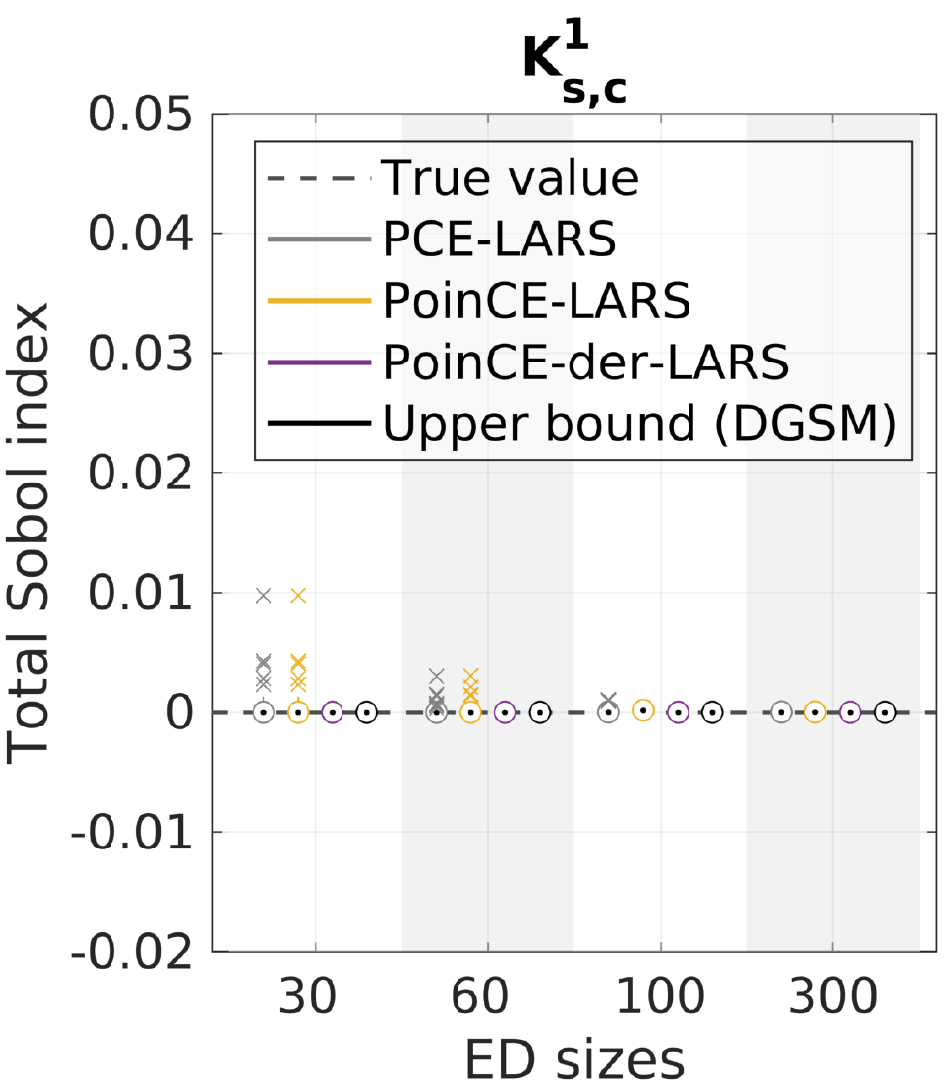}}
	\hfill
	{\includegraphics[width=.3\textwidth]
		{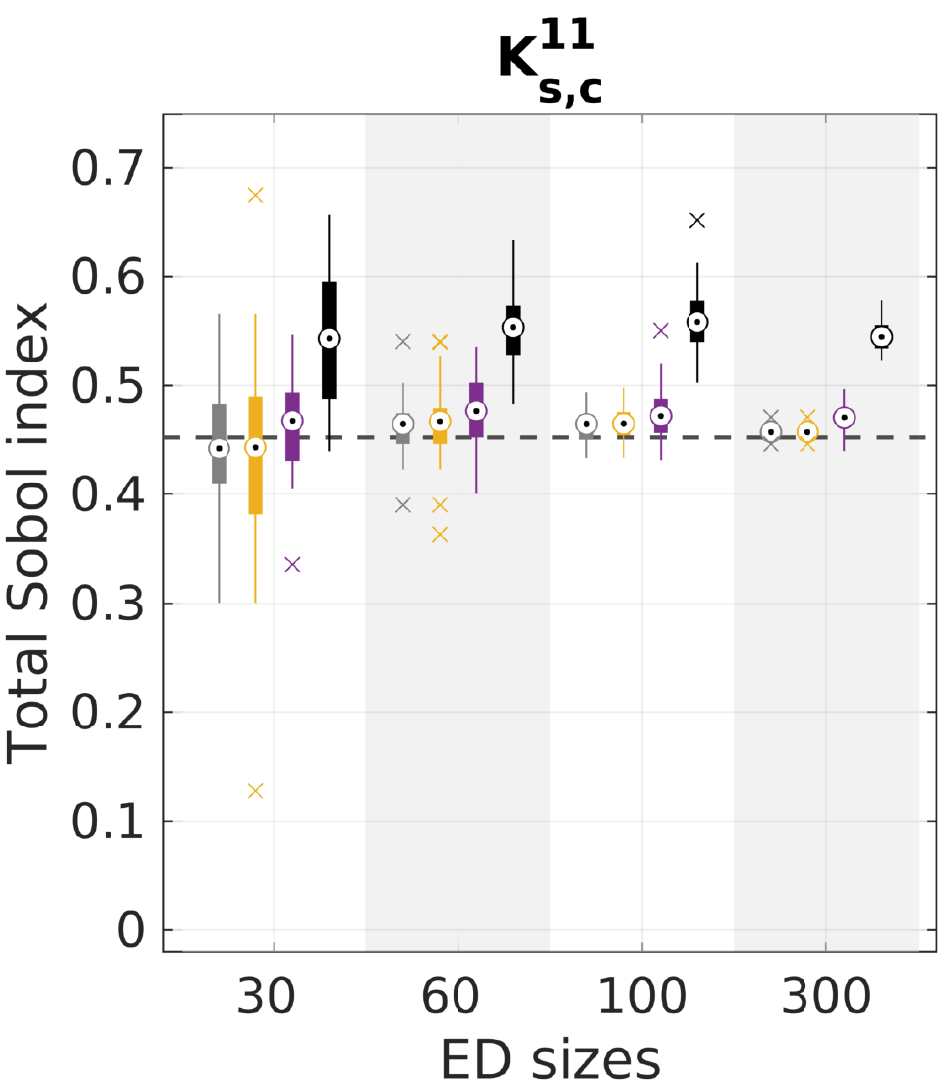}}
	\hfill
	{\includegraphics[width=.3\textwidth]
		{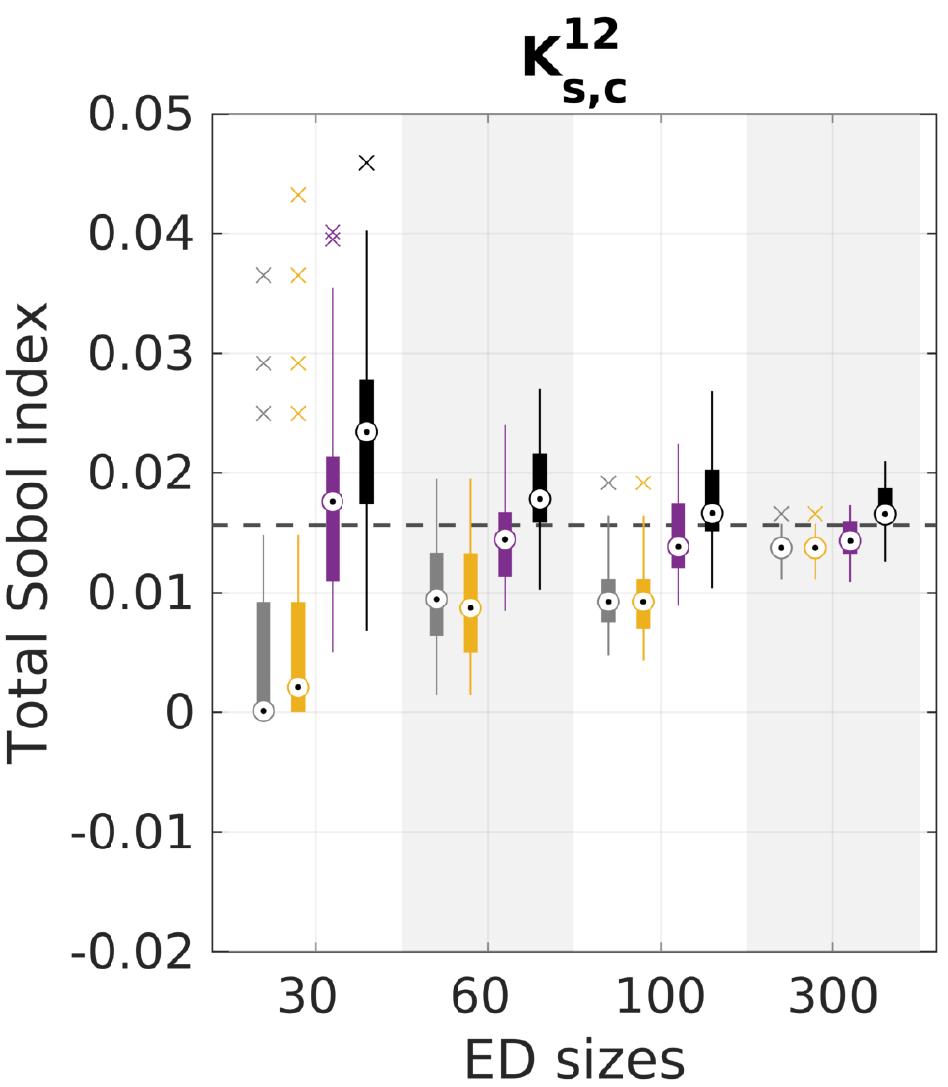}}
	\\
	{\includegraphics[width=.3\textwidth]
		{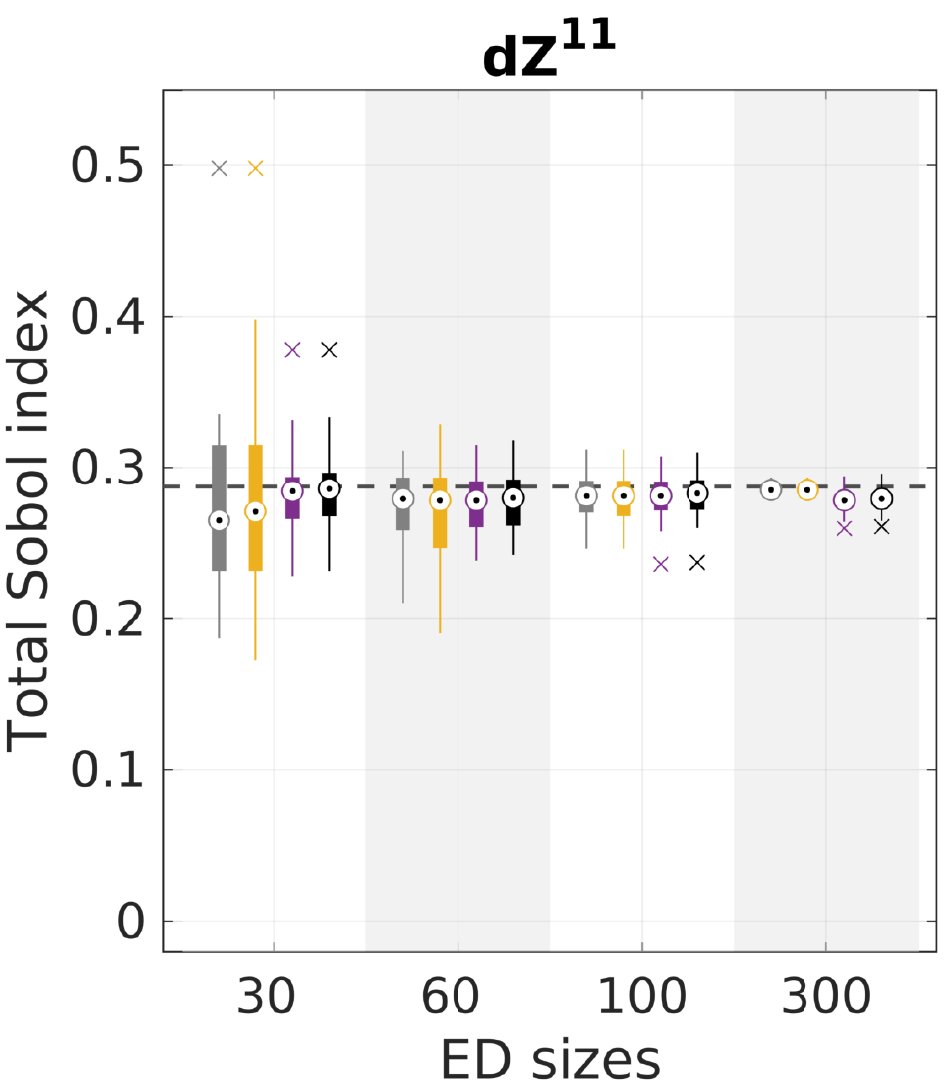}}
	\hfill
	{\includegraphics[width=.3\textwidth]
		{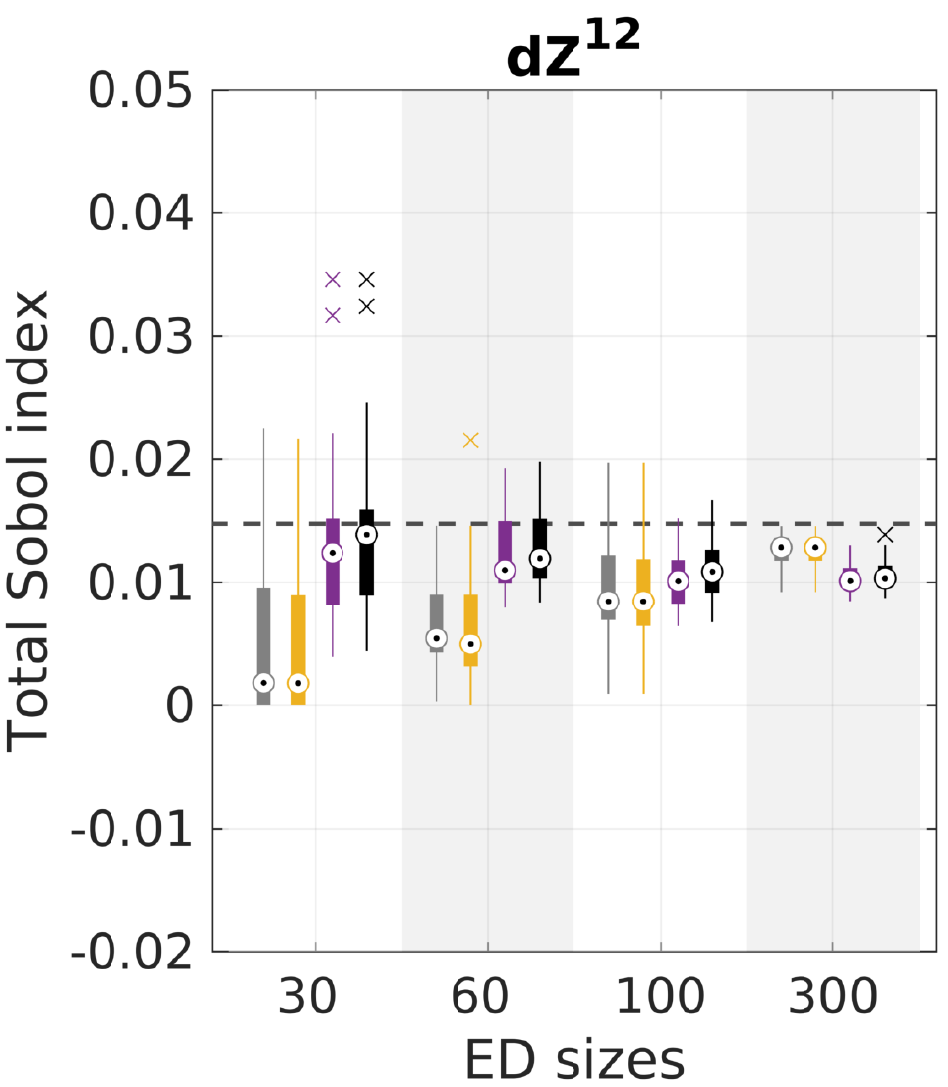}}
	\hfill
	{\includegraphics[width=.3\textwidth]
		{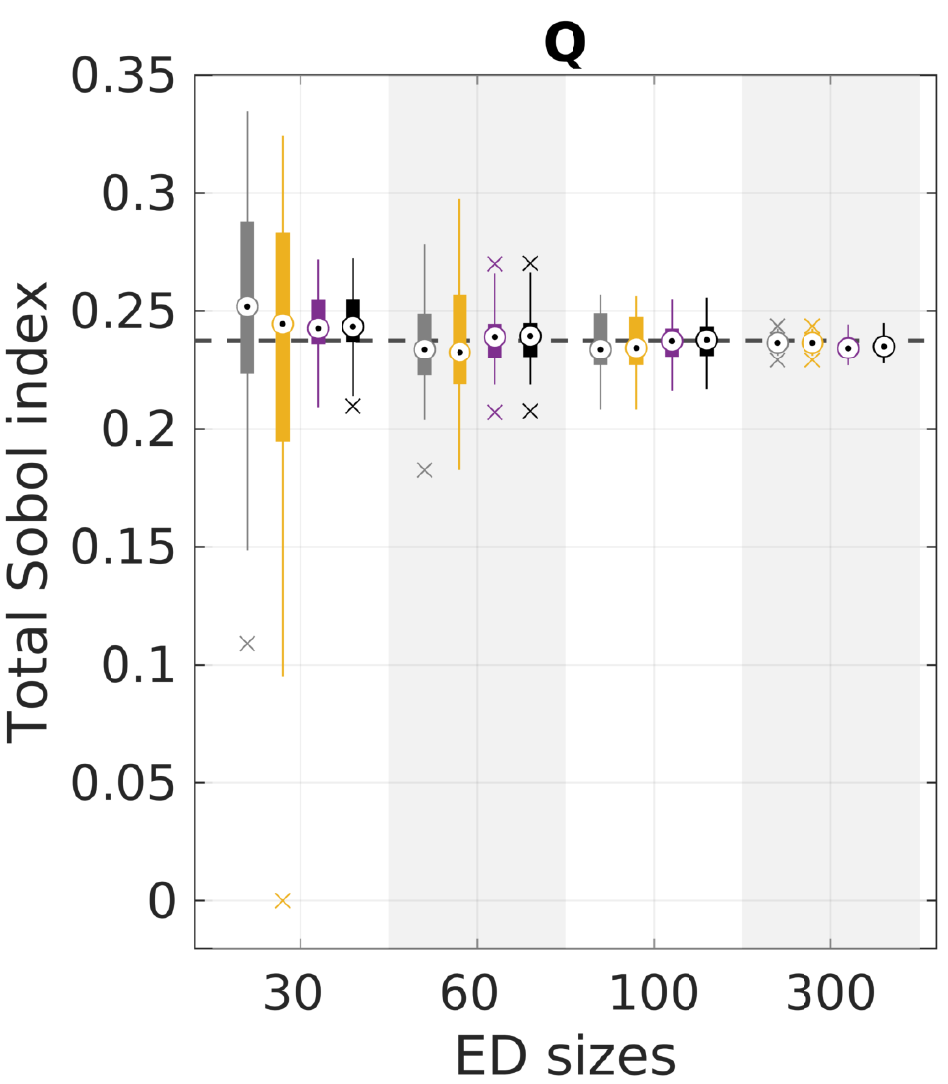}}

	\caption{Total Sobol' indices for the Mascaret data set (30 replications). ``True'' values computed from a PCE using all $20\,000$ points.}
	\label{fig:mascaret_total}
\end{figure}

\subsubsection{Comparison of total variance and relative mean-squared error}
In \cref{fig:mascaret_variance} we display various estimates for the total variance. We observe again that PoinCE-der yields an estimate with smaller variance and less bias than PoinCE and PCE. PoinCE and PCE both have smaller variance than the empirical estimate, but generally underestimate the total variance.

While PoinCE-der estimates Sobol' indices and total variance well, we observe in \cref{fig:mascaret_relmse} showing the relative MSE that PCE and PoinCE are performing better as global surrogate models: their model approximation error is for large experimental designs almost an order of magnitude better than for PoinCE-der.

\begin{figure}[htbp]
	\centering
	\subfloat[Estimates of the variance]{\includegraphics[width=.4\textwidth]{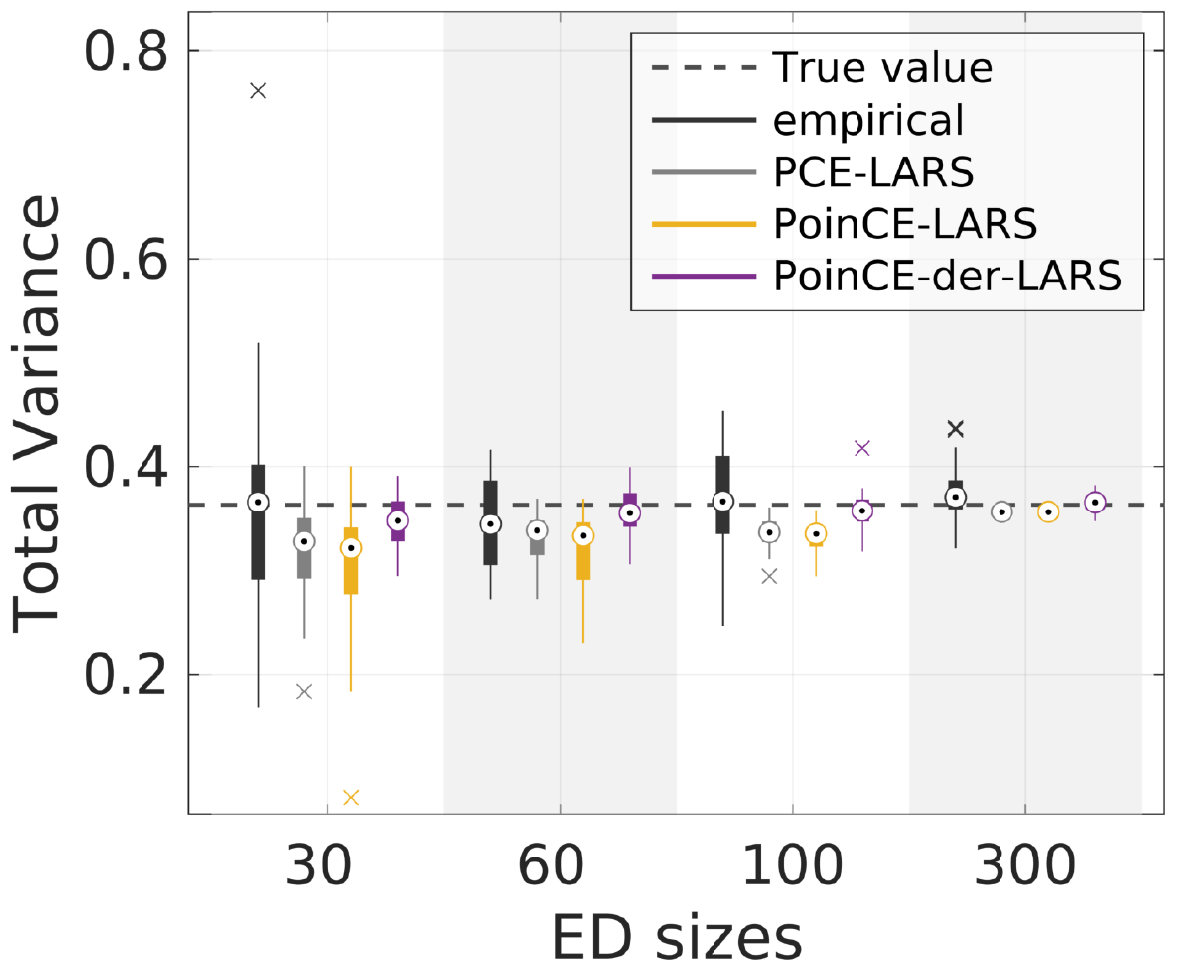}\label{fig:mascaret_variance}}
	\hspace{1cm}
	\subfloat[Relative mean-squared error]{\includegraphics[width=.4\textwidth]{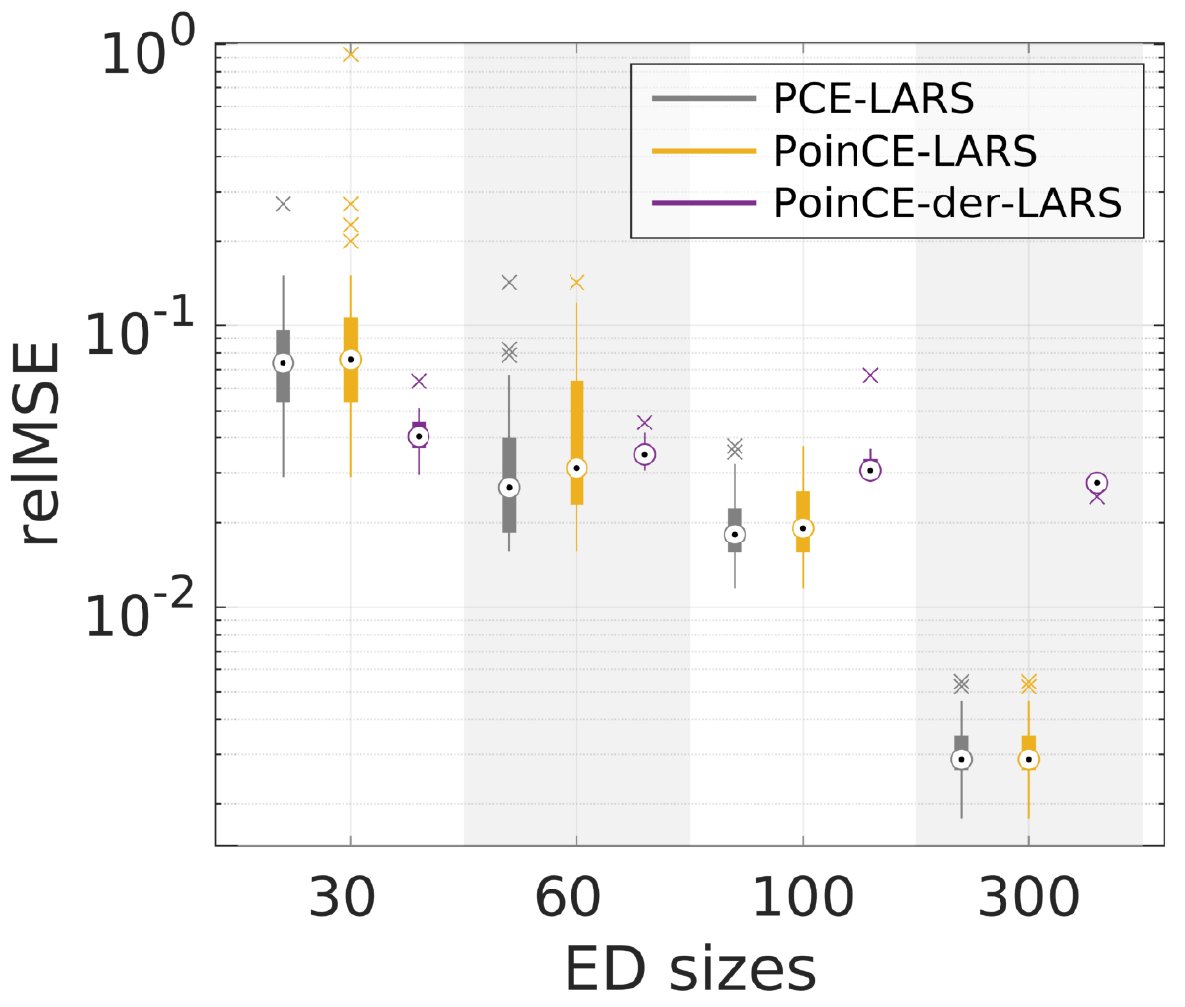}\label{fig:mascaret_relmse}}
	\caption{Mascaret data set. Comparison of PCE and PoinCE(-der) with respect to the following metrics: estimation of total variance, and relative mean-squared error. ``True'' value of total variance computed from a PCE using all $20\,000$ points.}
	\label{fig:mascaret_extra}
\end{figure}

\FloatBarrier 

\section{Conclusion}
\label{sec:conclusion}

In this paper we studied PoinCE, an expansion in terms of the Poincar\'e basis, which is an orthonormal basis of $L^2(\mu)$ with the unique property that all its partial derivatives form again an orthogonal basis for the same space.
We provided a proof of this property as well as a few analytical results as direct consequences.
In particular, we showed how upper and lower bounds for partial variances can be obtained analytically from PoinCE coefficients.

We described the computation of PoinCE and Poincar\'e derivative expansions by sparse regression and applied the method to two numerical examples. We found that while PoinCE does not outperform PCE in terms of validation error, it can be advantageous for estimating Sobol' indices in the low-data regime. 
PoinCE is therefore a valuable tool if model derivatives are cheaply available (e.g., by automatic differentiation or as a by-product of the simulation).
Taking partial derivatives reduces the size of the truncated basis especially for high-dimensional, low-order total-degree bases, which gives an advantage to derivative-based PoinCE over expansions relying on model evaluations.

Future work on the topic of PoinCE will investigate the simultaneous use of  model evaluations and derivatives for the computation of the coefficients, and compare to the related topic of gradient-enhanced PCE.

\section*{Acknowledgements}

This paper is a part of the project ``Surrogate Modeling for Stochastic Simulators (SAMOS)'' funded by the Swiss National Science Foundation (Grant \#200021\_175524), whose support is gratefully acknowledged.
Part of this research was conducted within the frame of the Chair in Applied Mathematics OQUAIDO, gathering partners in technological research (BRGM, CEA, IFPEN, IRSN, Safran, Storengy) and academia (CNRS, Ecole Centrale de Lyon, Mines Saint-Etienne, University of Grenoble, University of Nice, University of Toulouse) around advanced methods for computer experiments. Support from the ANR-3IA Artificial and Natural Intelligence Toulouse Institute is gratefully acknowledged.



\clearpage

\appendix

\section{Additional results}
In \crefrange{fig:flood_Sobol_firstorder_RegVsMC_appendix}{fig:flood_unnormalized_Sobol_total_appendix}, we show additional results for the dyke cost model, namely Sobol' index estimates (normalized and unnormalized) for the remaining five input variables.
For the corresponding discussion, see \cref{sec:dyke_cost}.

\label{app:additional}
\begin{figure}[htbp]
	\centering
	{\includegraphics[width=.3\textwidth]{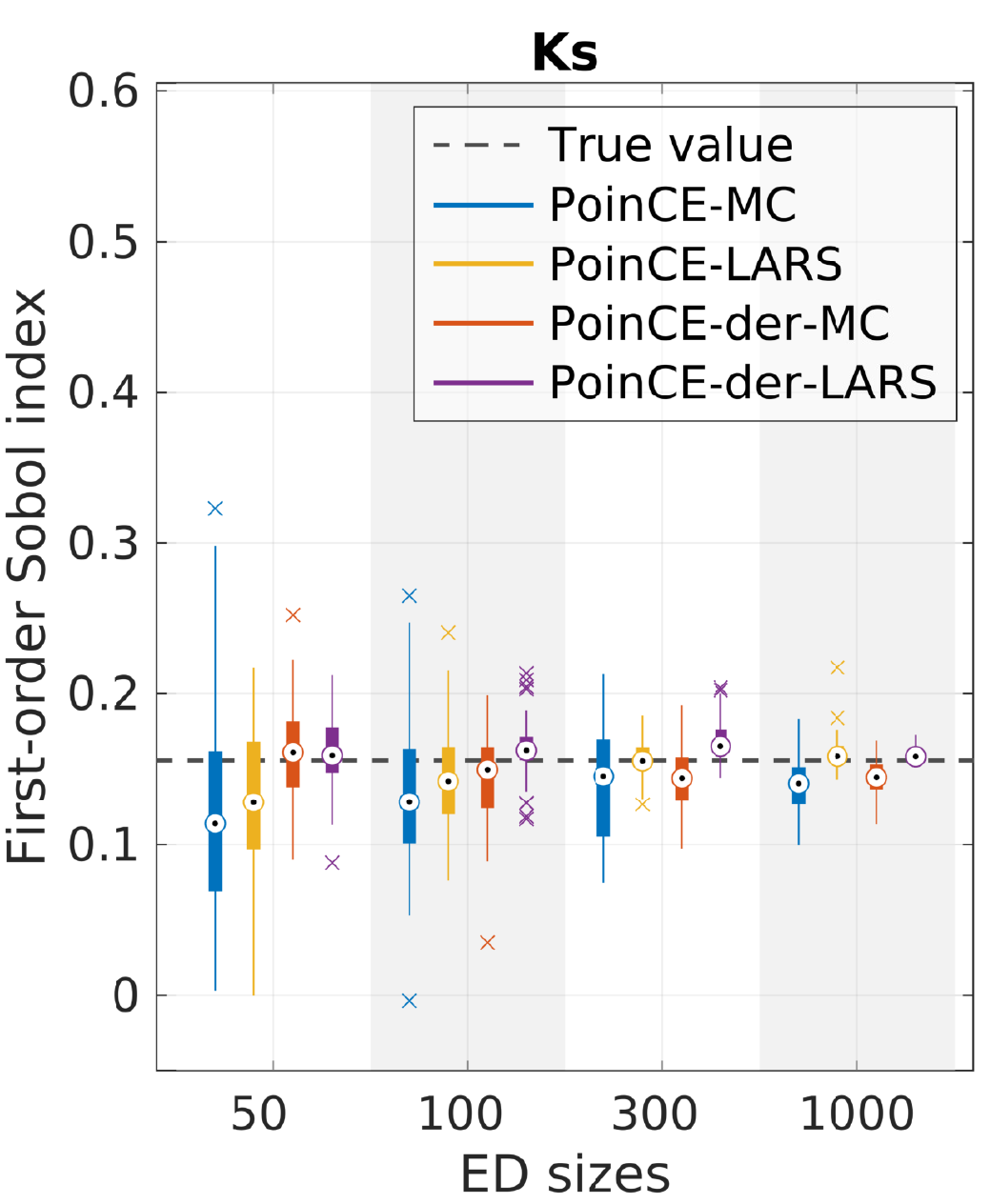}}
	\hfill
	{\includegraphics[width=.3\textwidth]{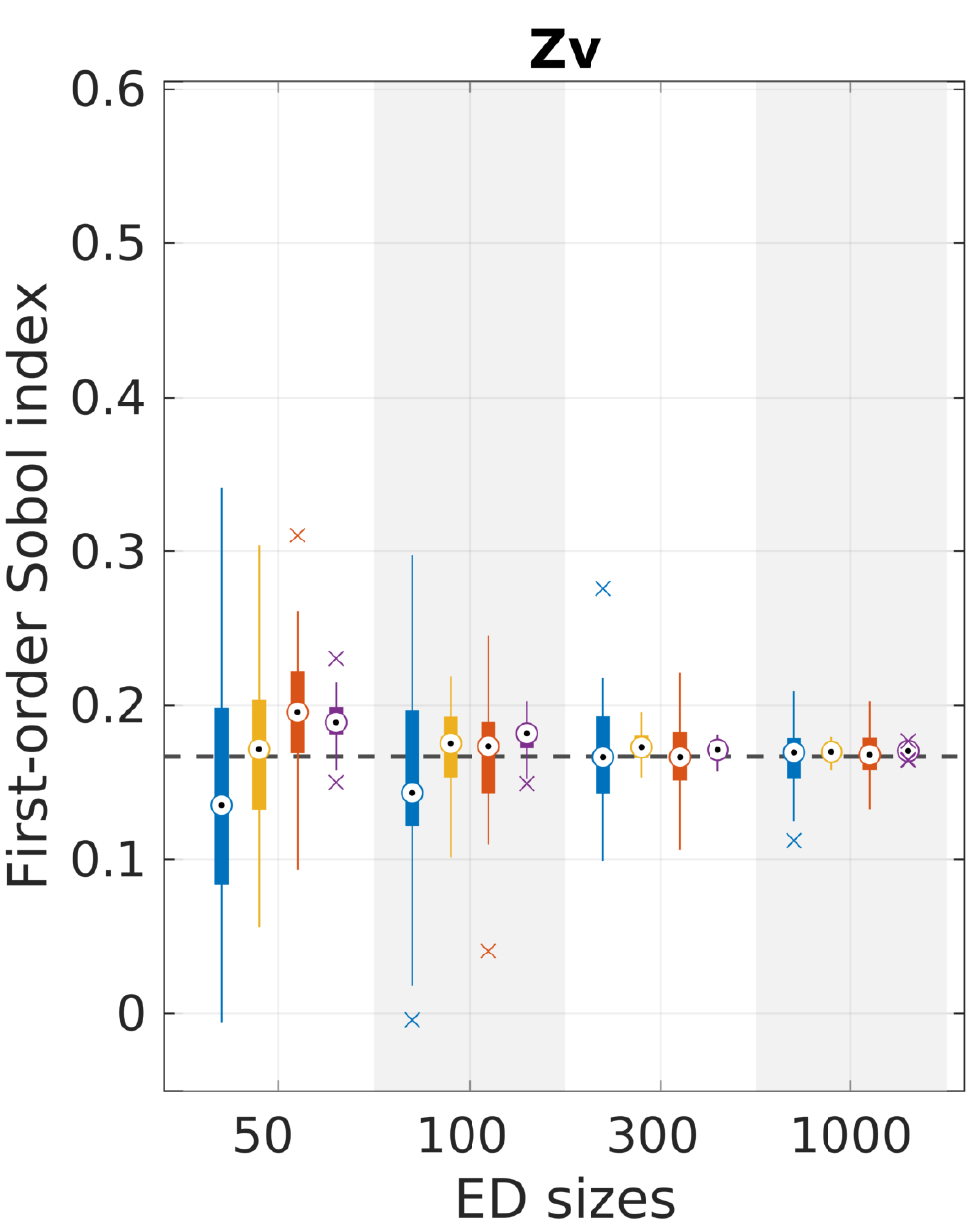}}
	\hfill
	{\includegraphics[width=.3\textwidth]{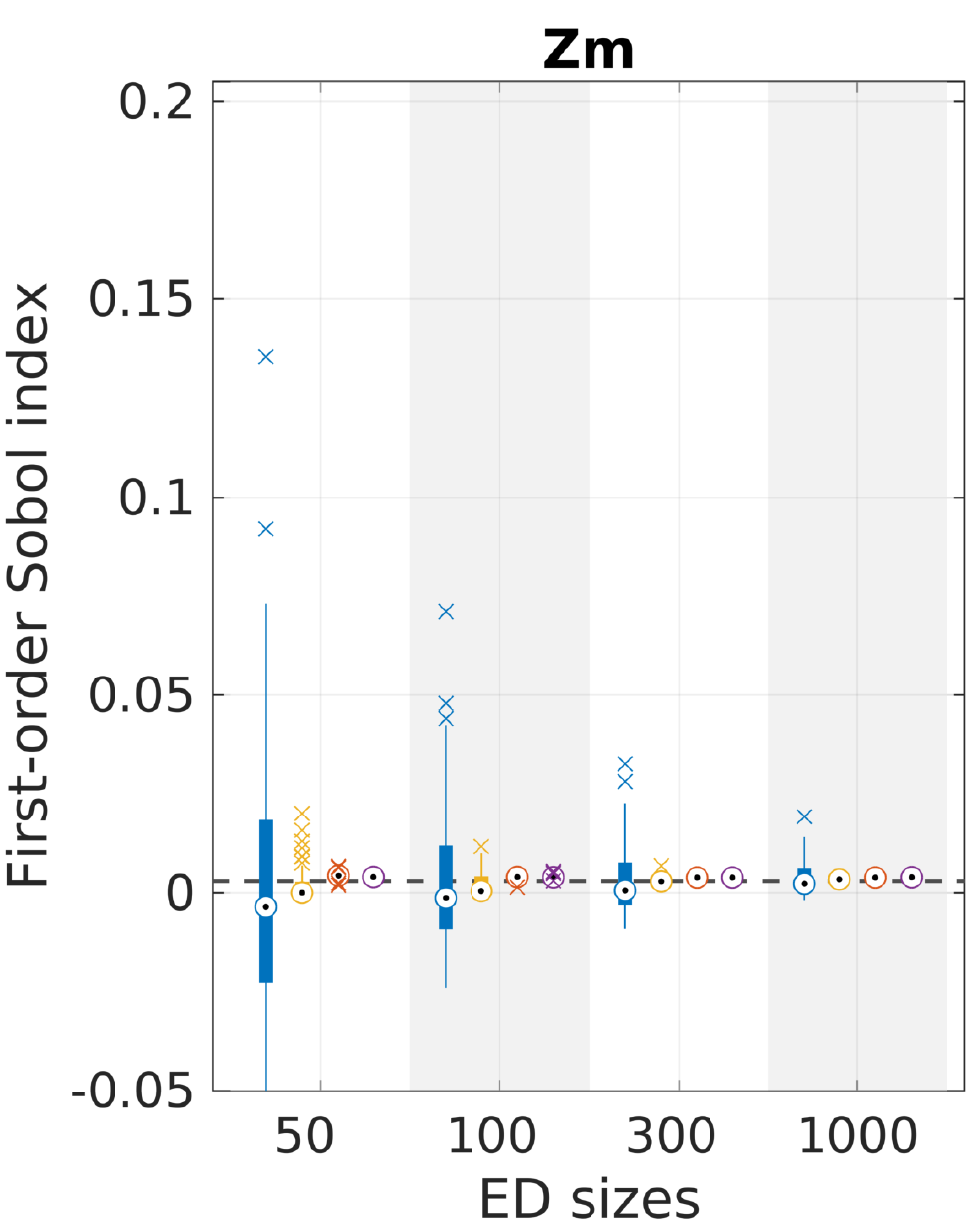}}
	\\
	{\includegraphics[width=.3\textwidth]{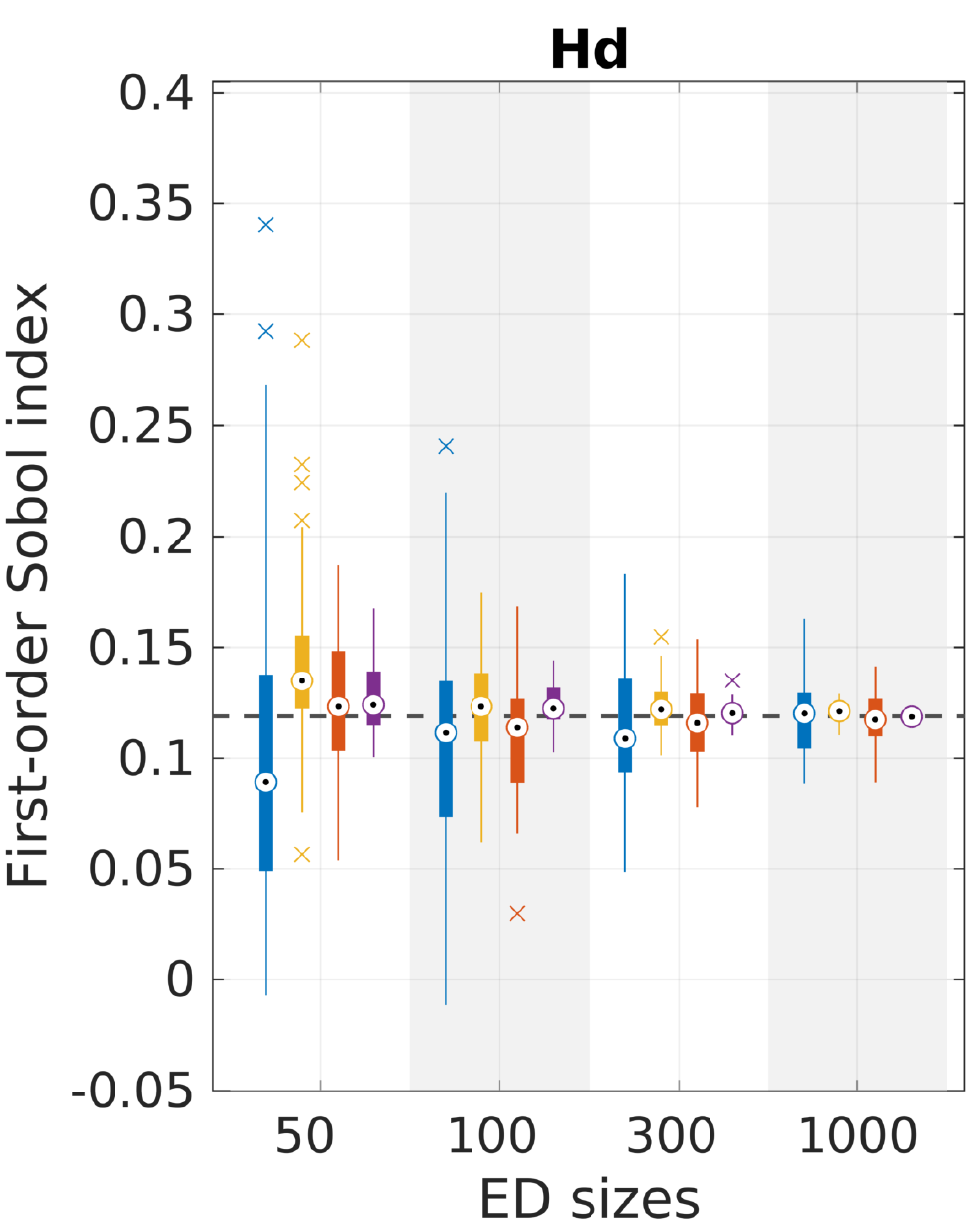}}
	\hspace{.05\textwidth}
	{\includegraphics[width=.3\textwidth]{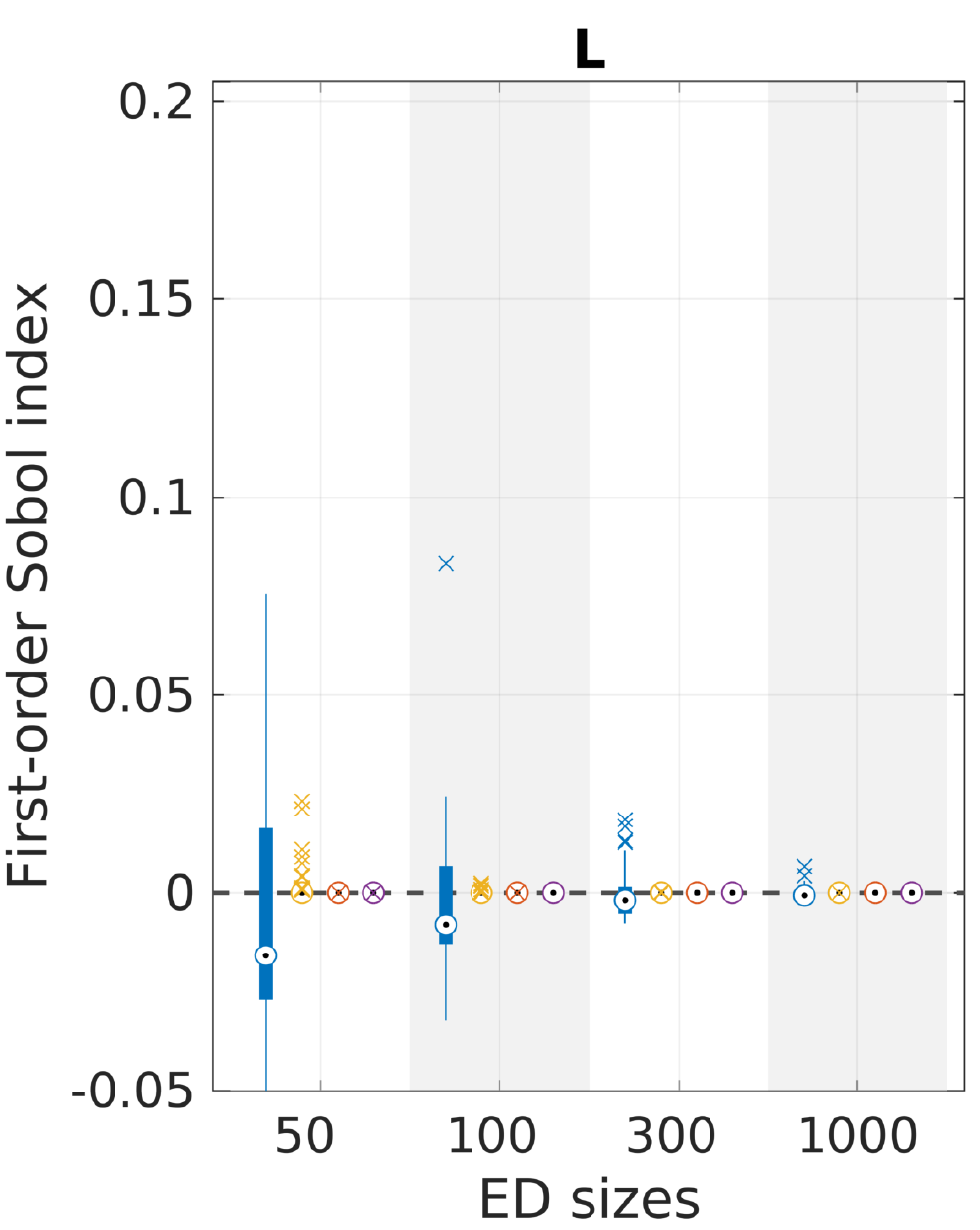}}
	\caption{Comparison of PoinCE estimates of {first-order Sobol' indices} for the dyke cost model. Degree $p = 2$ for the MC-based estimates and $p \leq 5$ (degree-adaptive) for the regression-based estimates. See also \cref{fig:flood_Sobol_firstorder_RegVsMC} in the main part of the paper.}
	\label{fig:flood_Sobol_firstorder_RegVsMC_appendix}
\end{figure}

\begin{figure}[htbp]
	\centering
	{\includegraphics[width=.3\textwidth]{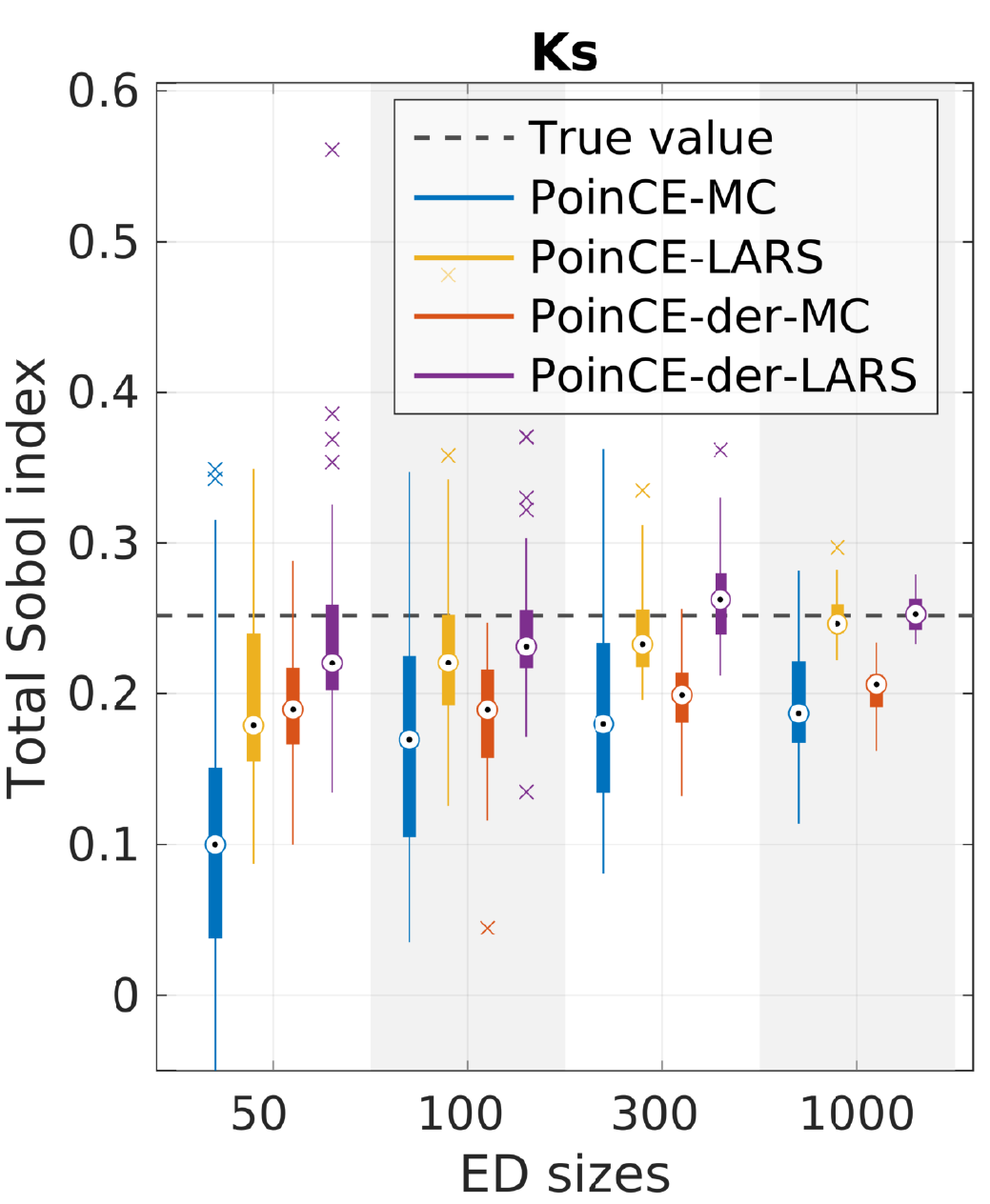}}
	\hfill
	{\includegraphics[width=.3\textwidth]{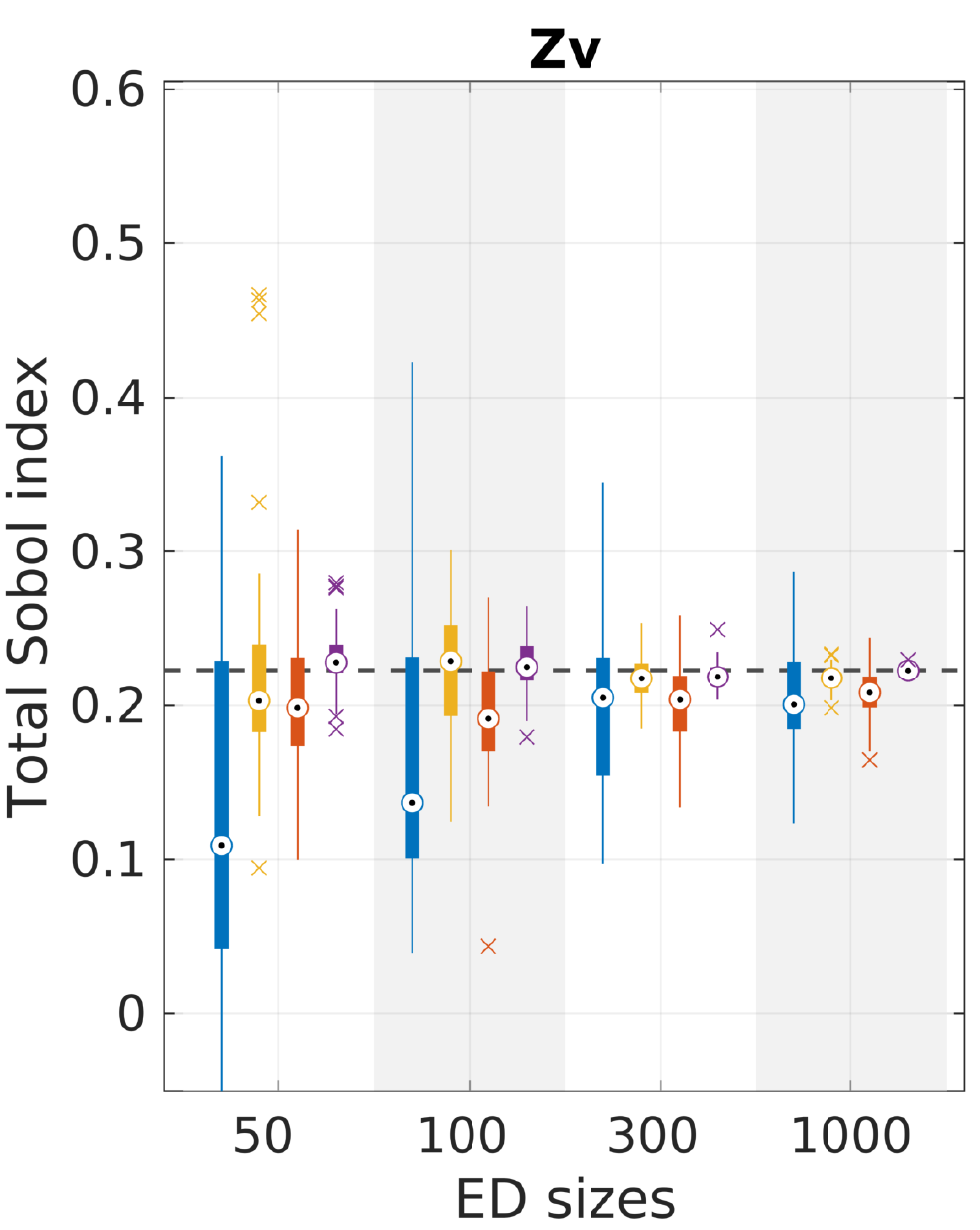}}
	\hfill
	{\includegraphics[width=.3\textwidth]{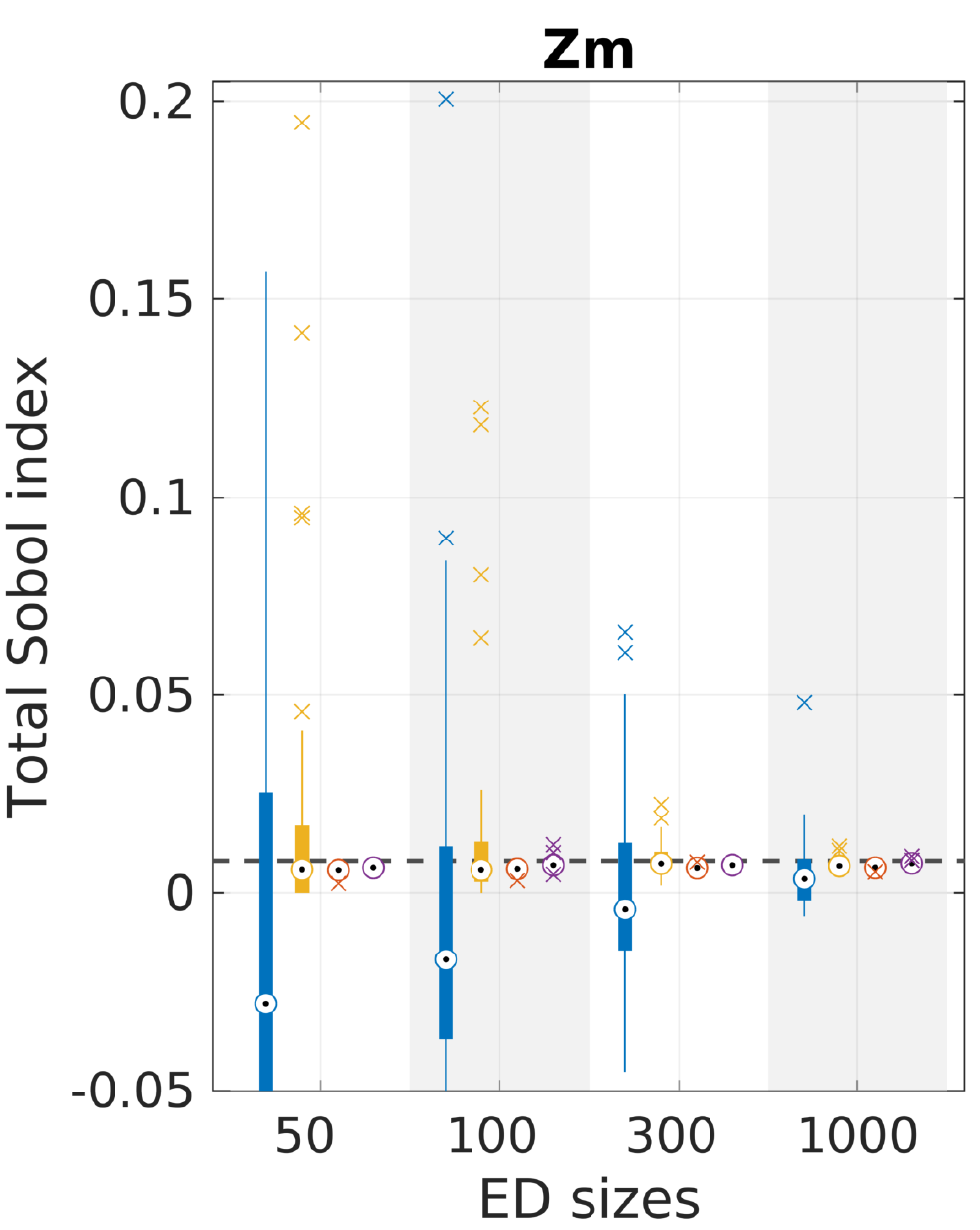}}
	\\
	{\includegraphics[width=.3\textwidth]{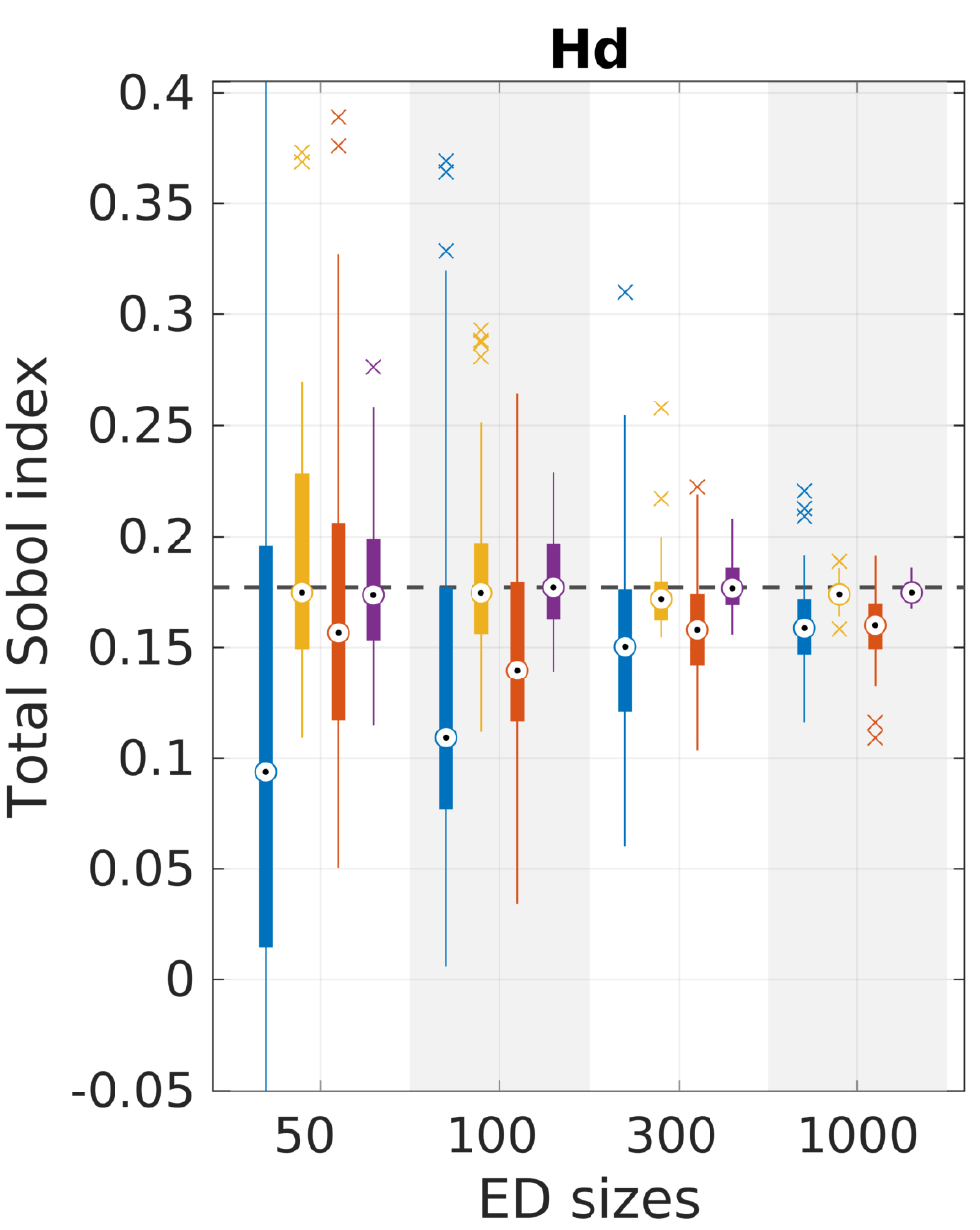}}
	\hspace{.05\textwidth}
	{\includegraphics[width=.3\textwidth]{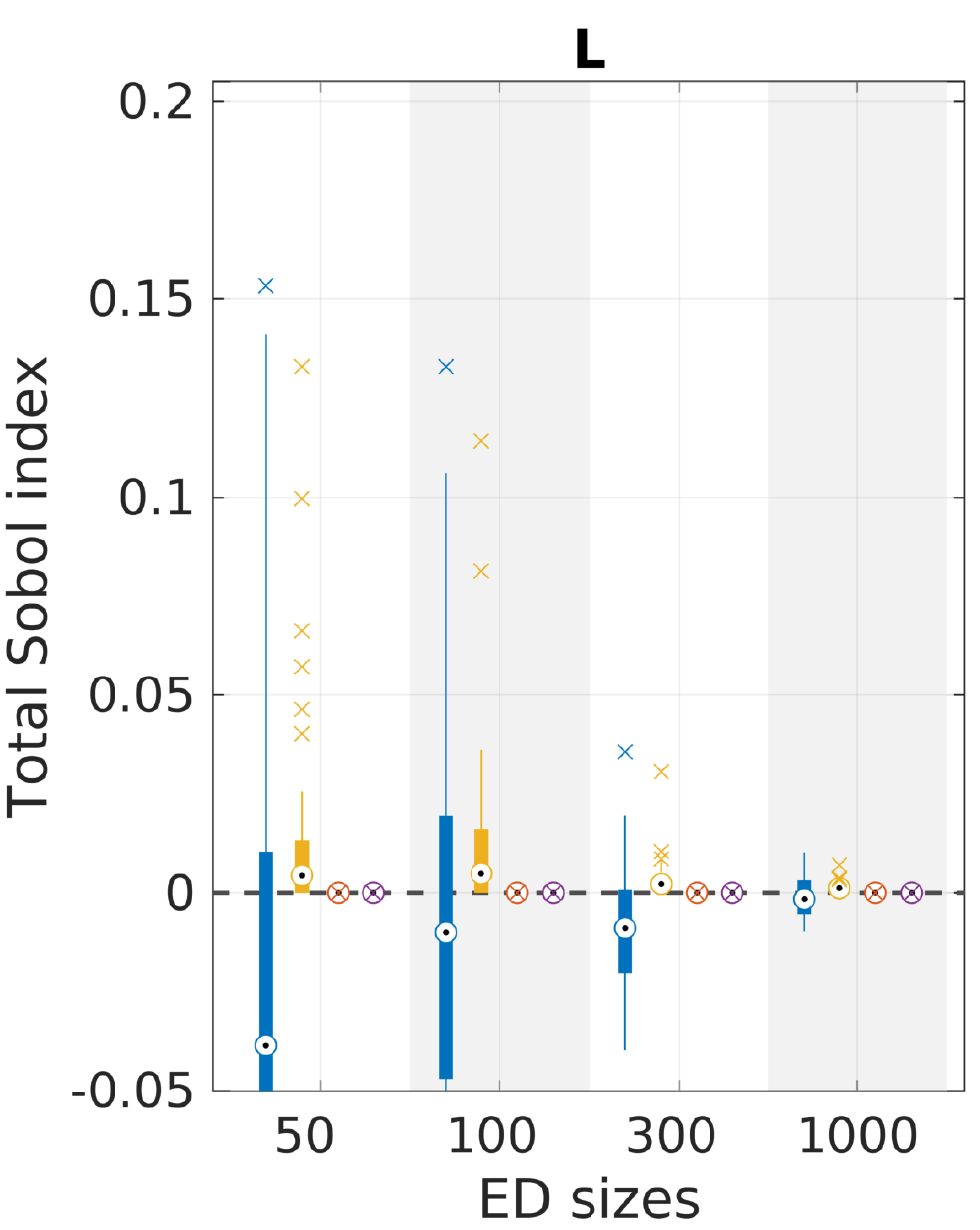}}
	\caption{Comparison of PoinCE estimates of {total Sobol' indices} for the dyke cost model. Degree $p = 2$ for the MC-based estimates and $p \leq 5$ (degree-adaptive) for the regression-based estimates. See also \cref{fig:flood_Sobol_total_RegVsMC} in the main part of the paper.}
	\label{fig:flood_Sobol_total_RegVsMC_appendix}
\end{figure}

\begin{figure}[htbp]
	\centering
	{\includegraphics[width=.3\textwidth]
		{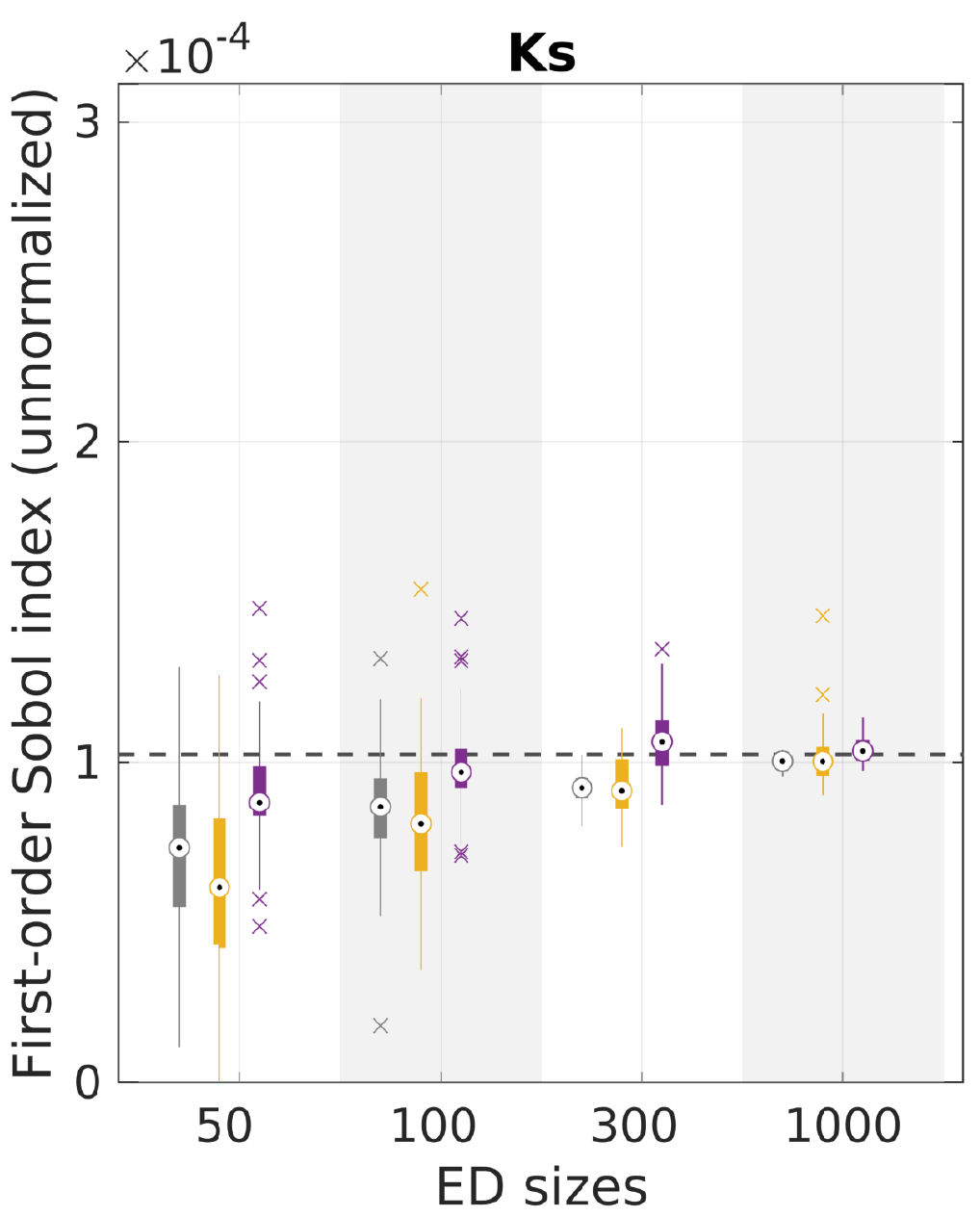}}
	\hfill
	{\includegraphics[width=.3\textwidth]
		{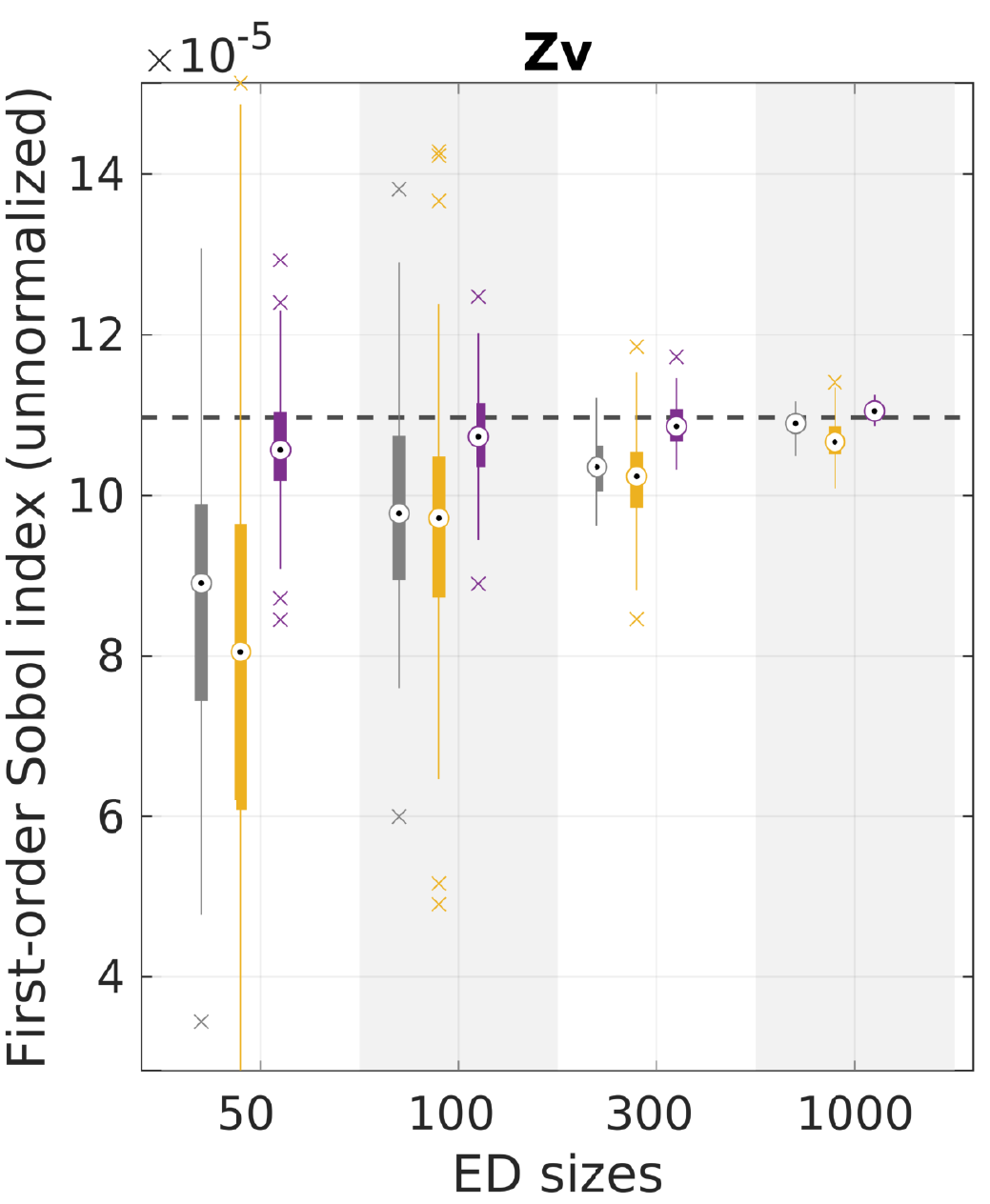}}
	\hfill
	{\includegraphics[width=.3\textwidth]
		{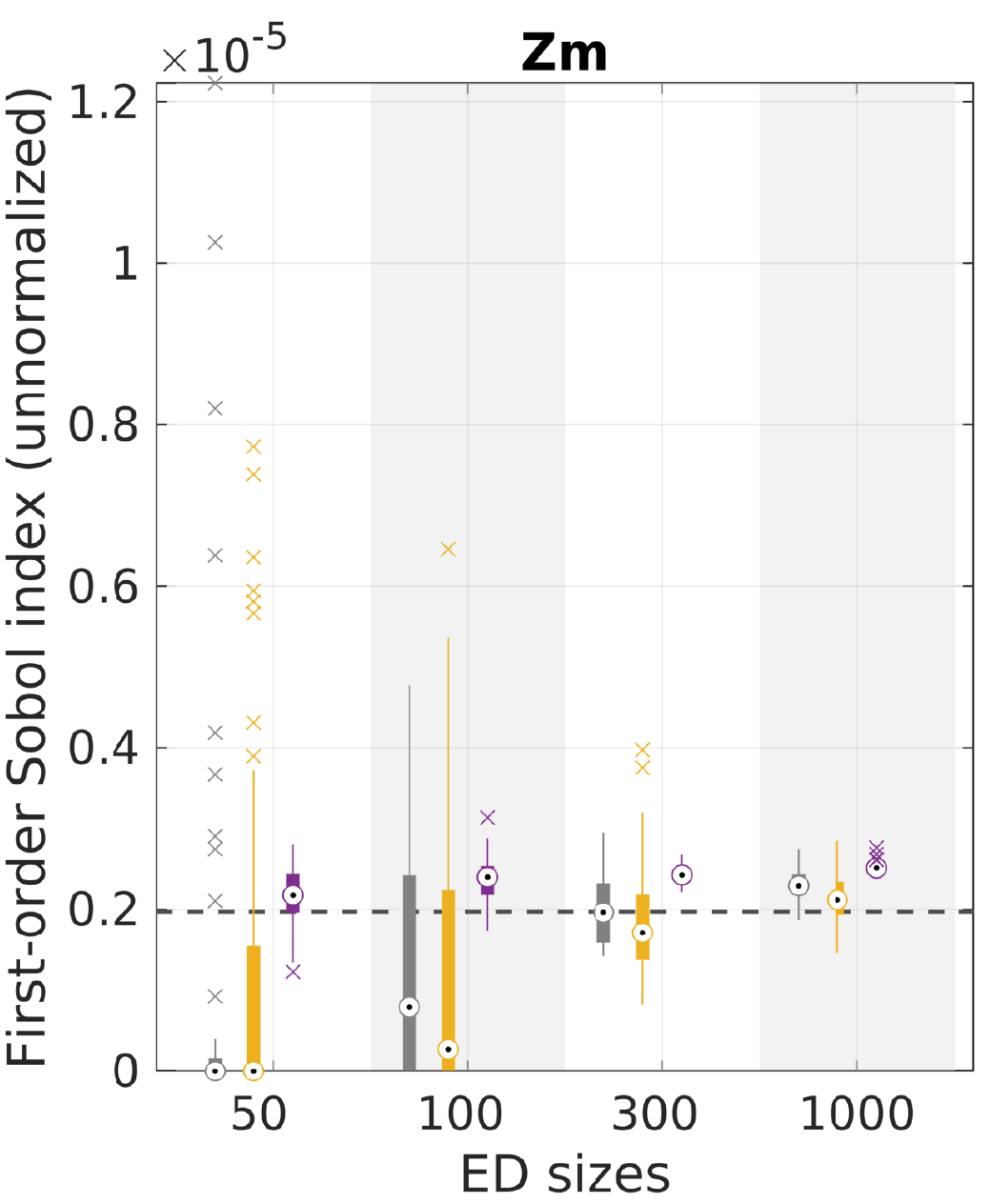}}
	\\
	{\includegraphics[width=.3\textwidth]
		{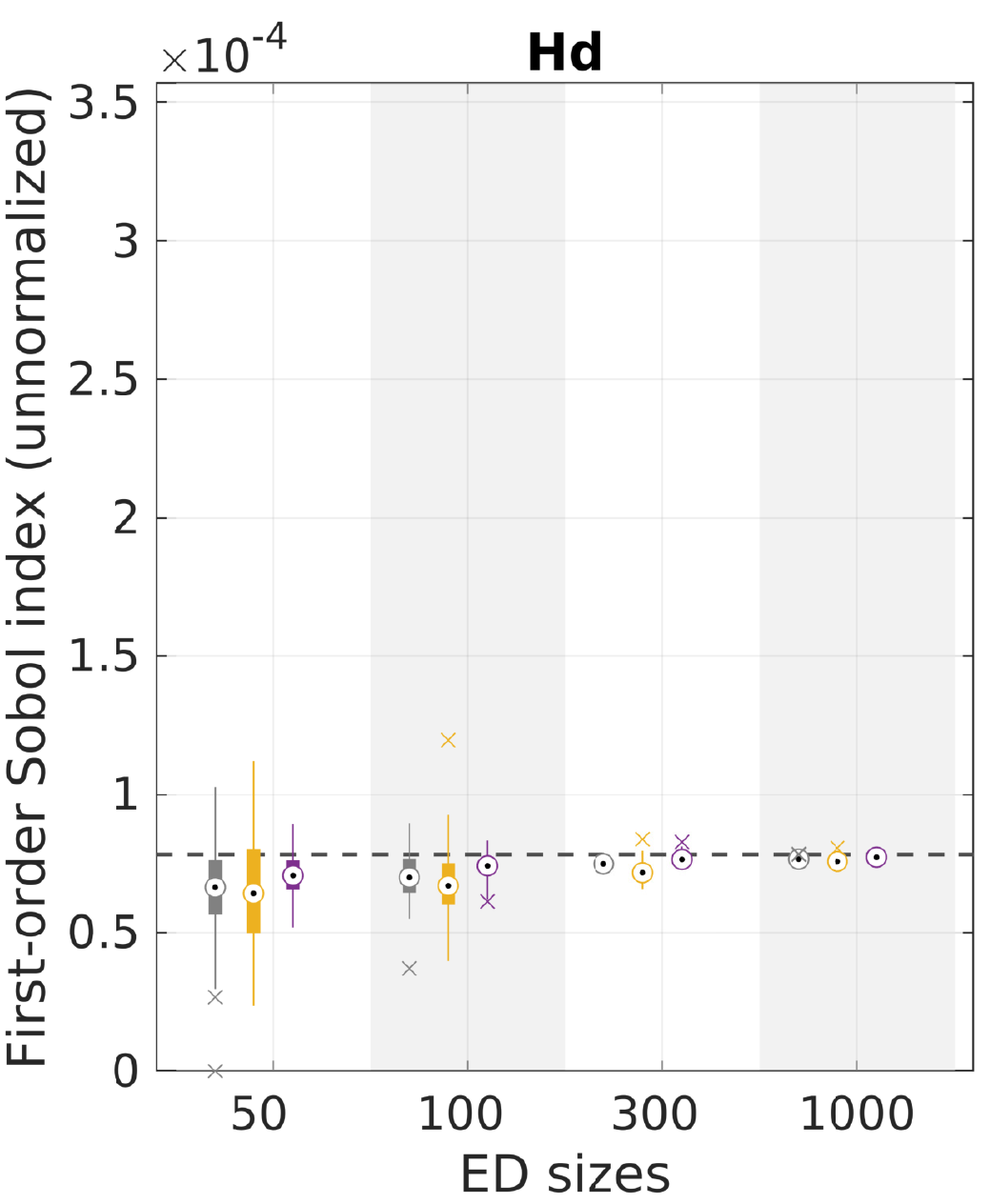}}
	\hspace{.05\textwidth}
	{\includegraphics[width=.3\textwidth]
		{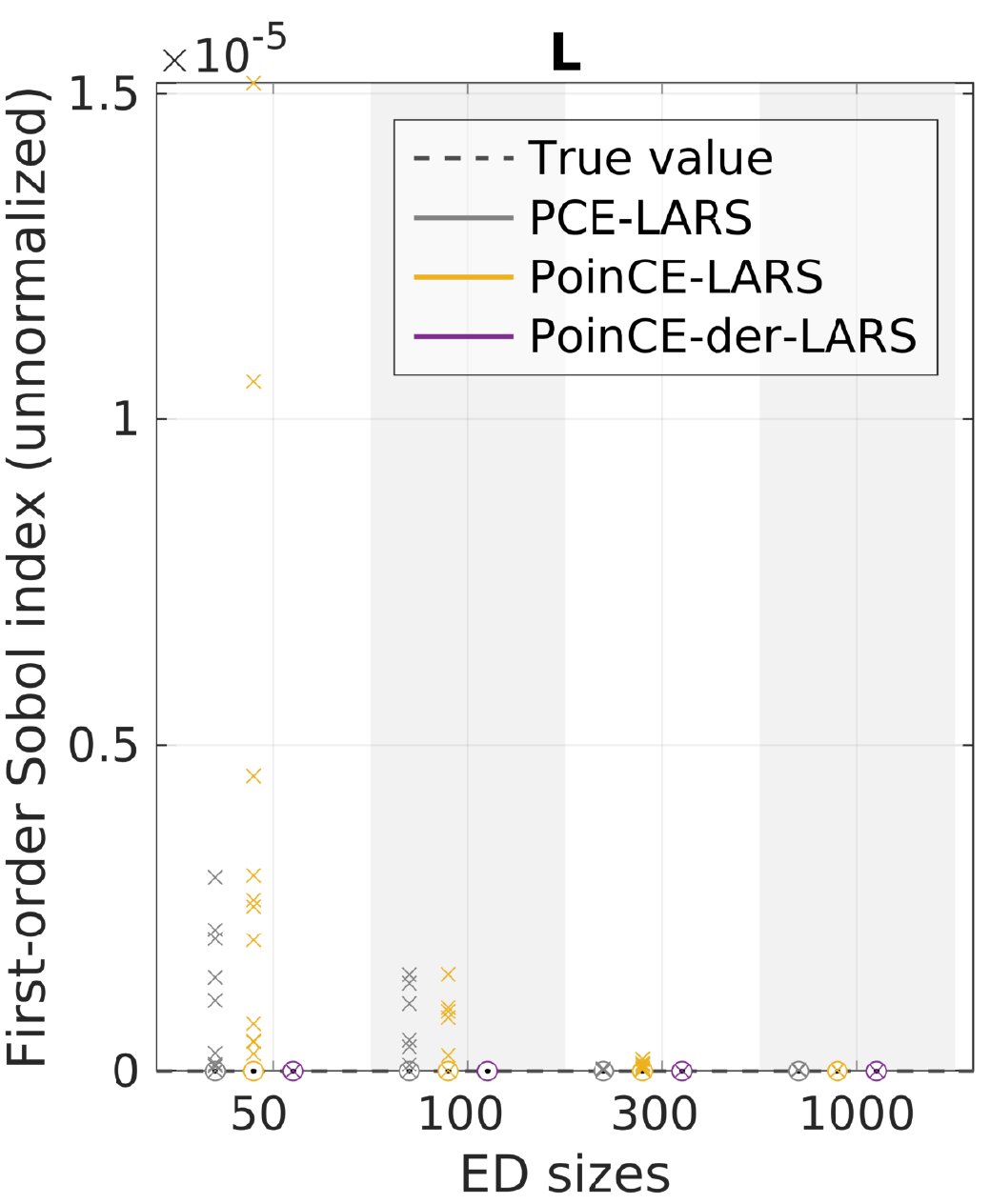}}
	\caption{Estimates of {unnormalized first-order Sobol' indices} for the dyke cost model ($p \leq 5$). 
		Boxplots: in grey the PCE-based estimates. The dashed line (``True value'') denotes a high-precision estimate for the unnormalized first-order Sobol' index.
		See also \cref{fig:flood_unnormalized_Sobol_firstorder}.}
	\label{fig:flood_unnormalized_Sobol_firstorder_appendix}
\end{figure}

\begin{figure}[htbp]
	\centering
	{\includegraphics[width=.3\textwidth]
		{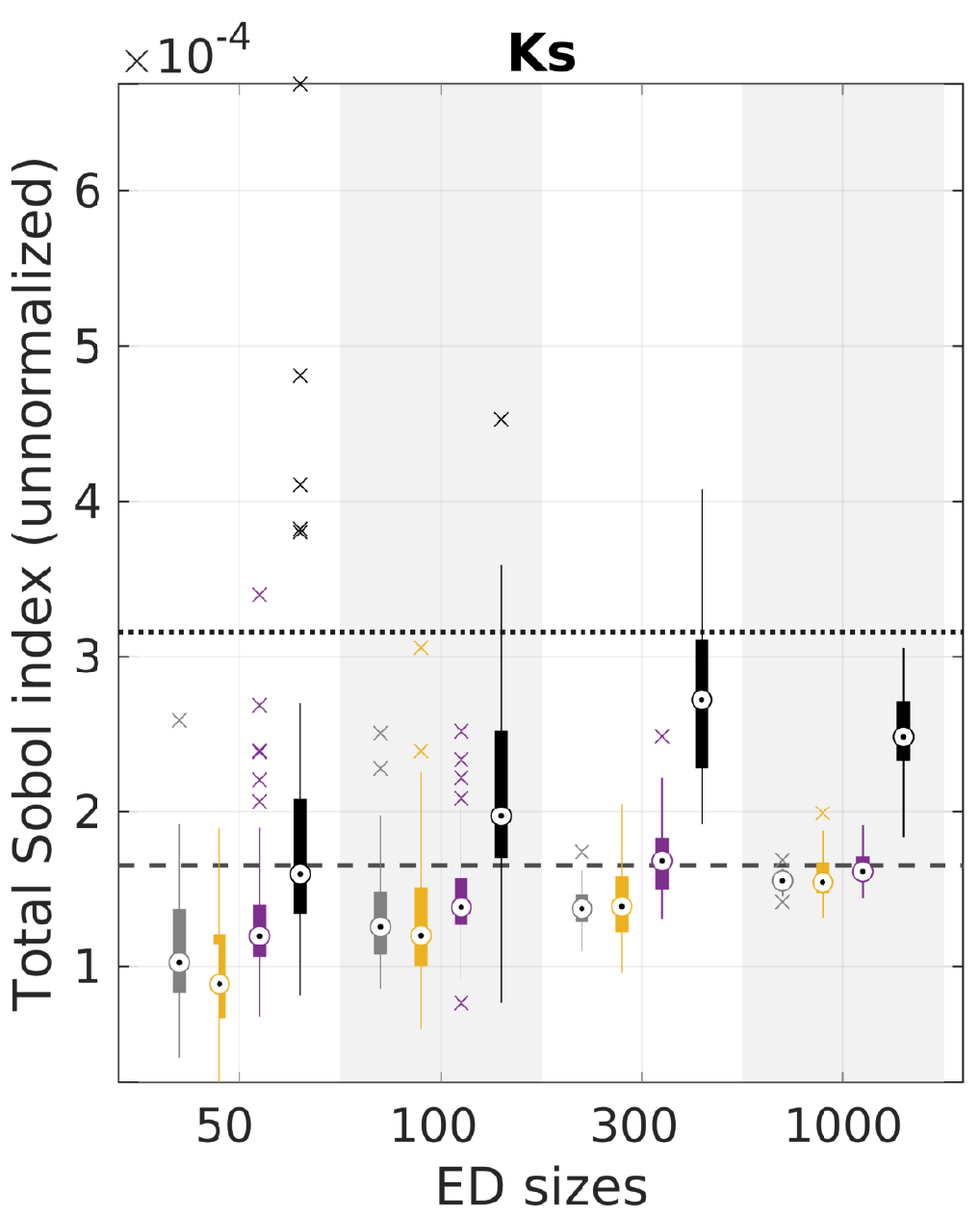}}
	\hfill
	{\includegraphics[width=.3\textwidth]
		{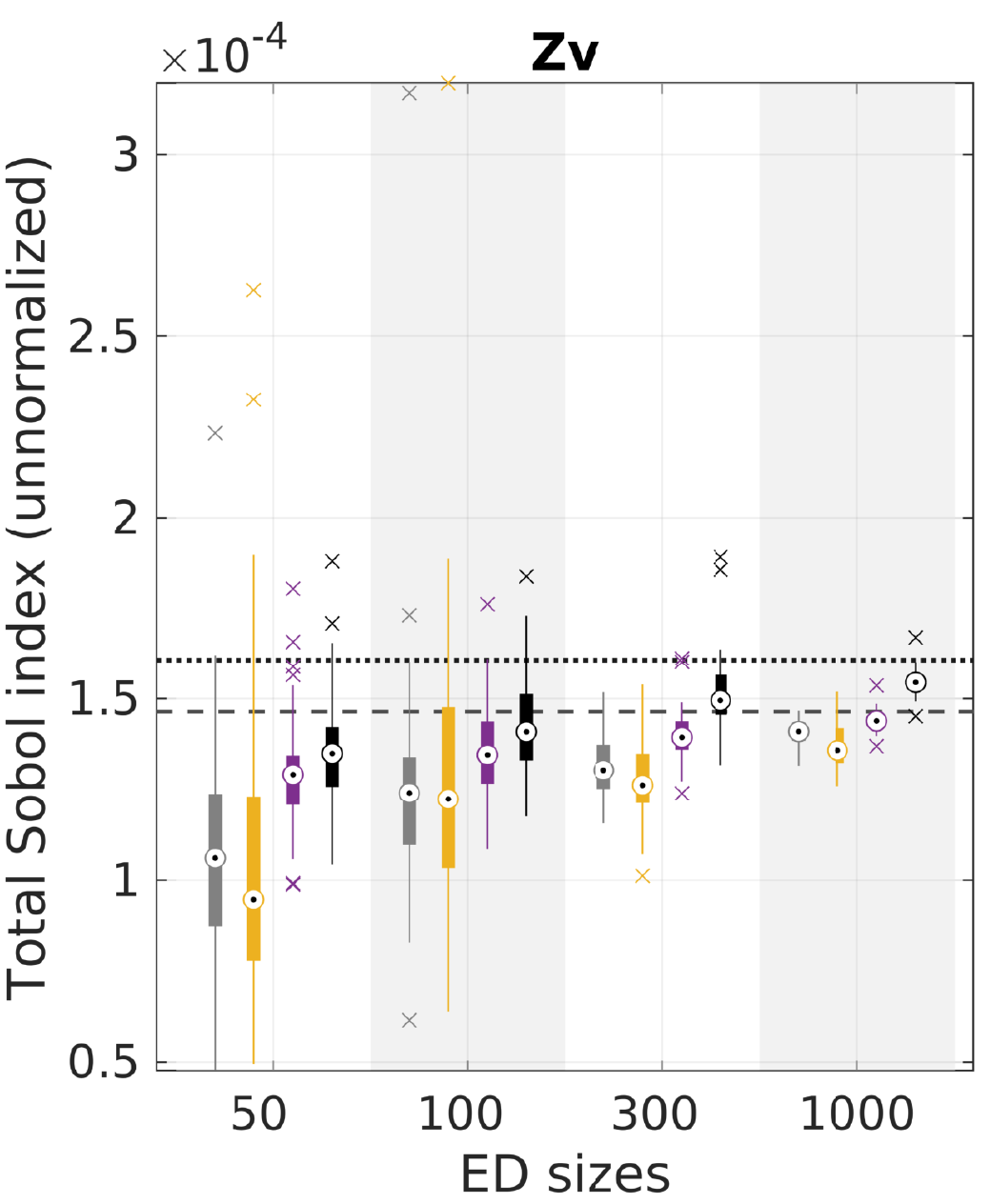}}
	\hfill
	{\includegraphics[width=.3\textwidth]
		{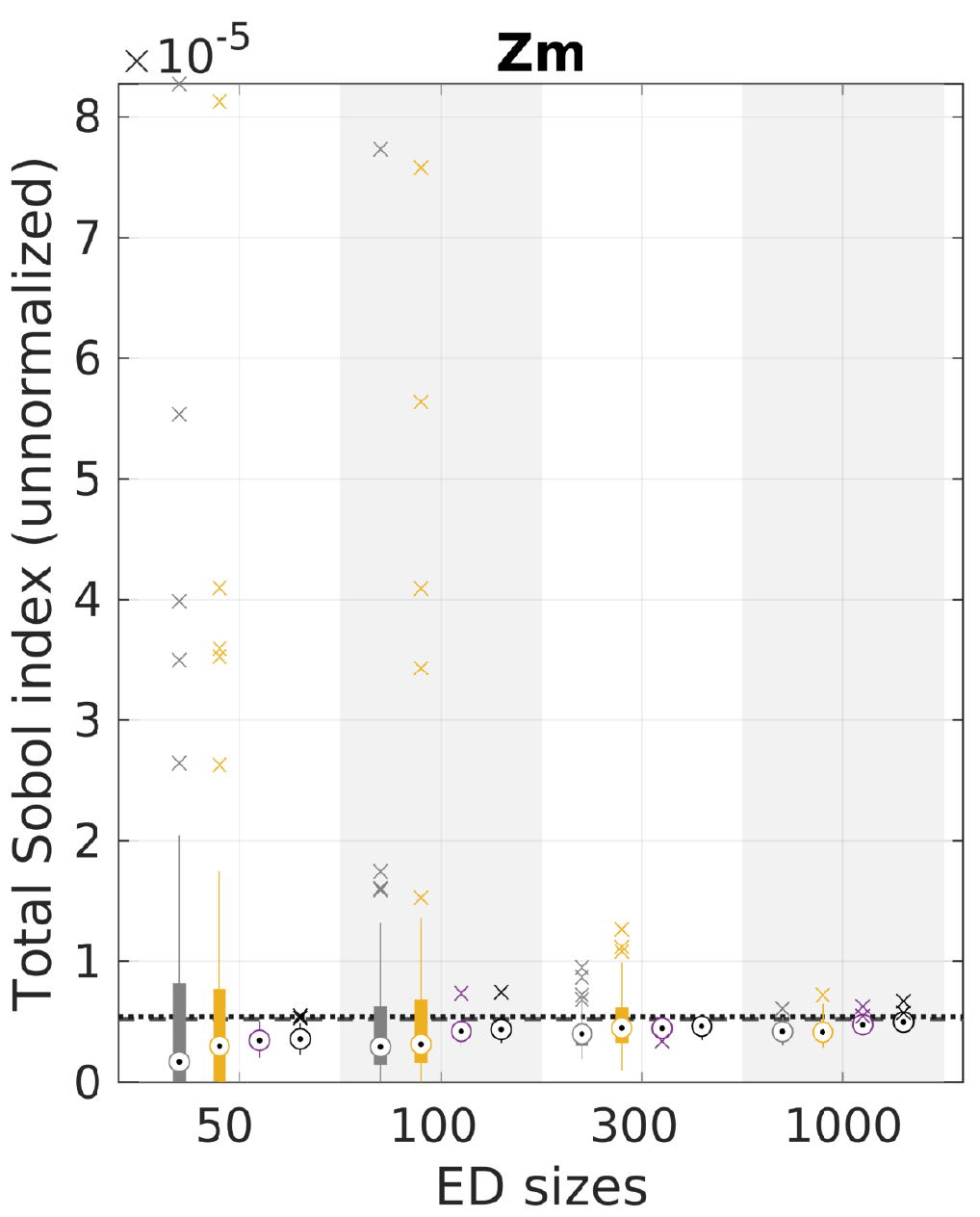}}
	\\
	{\includegraphics[width=.3\textwidth]
		{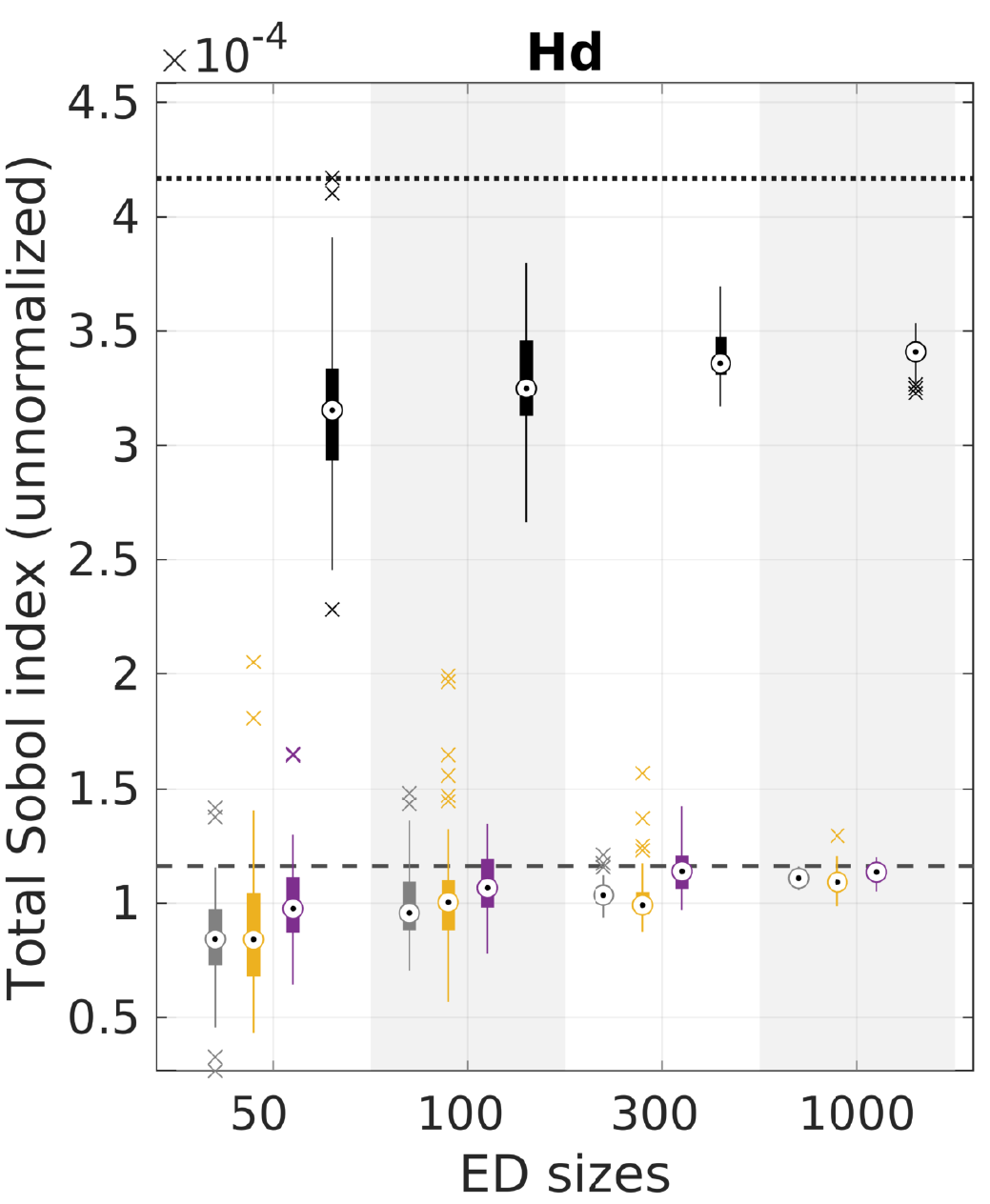}}
	\hspace{.05\textwidth}
	{\includegraphics[width=.3\textwidth]
		{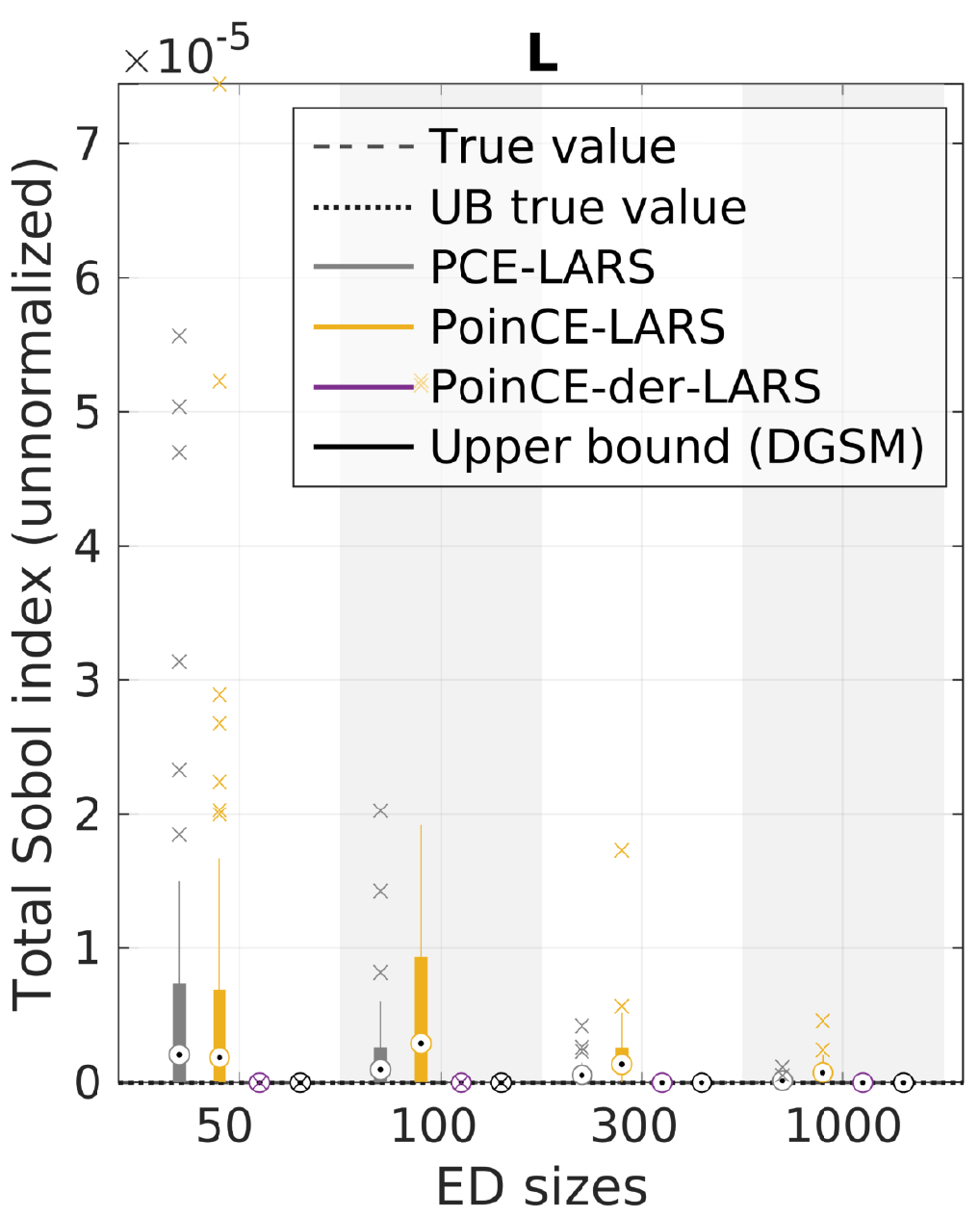}}
	\caption{Estimates of {unnormalized total Sobol' indices} for the dyke cost model ($p \leq 5$). 
		Boxplots: in grey the PCE-based estimates and in black the DGSM-based upper bound from \eqref{eq:DGSM_poincare_upper_bound}. 
		Lines: the dashed line (``True value'') denotes a high-precision estimate for the unnormalized total Sobol' index, while the dotted line (``UB true value'') is a MC-based high-precision estimate for the DGSM-based upper bound.
		See also \cref{fig:flood_unnormalized_Sobol_total}.}
	\label{fig:flood_unnormalized_Sobol_total_appendix}
\end{figure}

\end{document}